\documentclass[runningheads,envcountsame,a4paper]{elsarticle}

\usepackage[utf8]{inputenc}       
\usepackage[T1]{fontenc}

\pdfoutput=1 

\usepackage{subcaption}

\usepackage{amsfonts,amssymb,amsthm}
\usepackage{breqn}

\usepackage{tikz,tikz-cd}
\usetikzlibrary{calc}

\usepackage[margin=2.5cm]{geometry}

\usepackage{algorithm}
\usepackage{comment}
\usepackage[noend]{algpseudocode}

\usepackage{mathtools}
\usepackage{setspace}
\usepackage{graphicx}
\usepackage{hyperref}
\usepackage{aliascnt}

\usepackage{array}
\newcolumntype{P}[1]{>{\centering\arraybackslash}p{#1}}
\newcolumntype{M}[1]{>{\centering\arraybackslash}m{#1}}

\usepackage{xspace}
\newcommand{\W}[1]{\ensuremath{\textsf{W}[#1]}\xspace}
\newcommand\paraNP{\ensuremath{\textsf{paraNP}}\xspace}

\newcommand\NP{\ensuremath{\textsf{NP}}\xspace}

\newcommand\XP{\ensuremath{\textsf{XP}}\xspace}

\newcommand\FPT{\ensuremath{\textsf{FPT}}\xspace}

\newcommand{\yes}{\textsf{Yes}}

\newcommand{\diam}{\mathrm{diam}}

\newcommand{\nd}{\textsf{nd}}
\newcommand{\dnd}{\textsf{dnd}}
\newcommand{\vc}{\textsf{vc}}

\newtheorem{theorem}{Theorem}[section] 
\newaliascnt{proposition}{theorem}
\newtheorem{proposition}[proposition]{Proposition}
\aliascntresetthe{proposition}

\newaliascnt{lemma}{theorem}
\newtheorem{lemma}[lemma]{Lemma}
\aliascntresetthe{lemma}

\newaliascnt{claim}{theorem}
\newtheorem{claim}[claim]{Claim}
\aliascntresetthe{claim}

\newaliascnt{corollary}{theorem}
\newtheorem{corollary}[corollary]{Corollary}
\aliascntresetthe{corollary}

\newaliascnt{observation}{theorem}

\aliascntresetthe{observation}

\newaliascnt{definition}{theorem}

\aliascntresetthe{definition}

\newaliascnt{example}{theorem}
\newtheorem{example}[example]{Example}
\aliascntresetthe{example}

\newaliascnt{remark}{theorem}
\newtheorem{remark}[remark]{Remark}
\aliascntresetthe{remark}

\newcommand{\Oh}{\mathcal{O}}

\usepackage[most]{tcolorbox}

\newcommand{\problemdef}[4]{
	\begin{tcolorbox}[width = \textwidth,colback=white,arc=0pt,outer arc=0pt,boxrule=0.7pt,left =0.5em,right=0em]\textsc{#1} #2		\\[2pt]
		\begin{tabular}{ @{}l p{0.84\textwidth} c }
			\textsf{Input:} & #3 \\[.5pt]
			\textsf{Problem:} & #4
		\end{tabular}
	\vspace{-0.25em}
	\end{tcolorbox}
}

\makeatletter
\let\@upn\textbf
\makeatother

\makeatletter
\def\BState{\State\hskip-\ALG@thistlm}
\makeatother

\begin{document}

\title{Parameterizing Path Partitions}

%
%


\author[Trier]{Henning Fernau}
\ead{fernau@uni-trier.de}

\author[Clermont-Auvergne]{Florent Foucaud\corref{cor1}}
\ead{florent.foucaud@uca.fr}
\cortext[cor1]{Corresponding author.}
\author[Trier]{Kevin Mann}
\ead{mann@uni-trier.de}
\author[Bangalore]{Utkarsh Padariya}
\ead{utkarsh.prafulchandra@iiitb.ac.in}
\author[Bangalore]{Rajath Rao K.N.}
\ead{rajath.rao@iiitb.ac.in}

\address[Trier]{Universit\"at Trier, Fachbereich IV, Informatikwissenschaften, Germany}
\address[Clermont-Auvergne]{Universit\'e Clermont Auvergne, CNRS, Clermont Auvergne INP, Mines Saint-\'Etienne, LIMOS, 63000 Clermont-Ferrand, France}
 \address[Bangalore]{International Institute of Information Technology Bangalore, India}


%
\begin{abstract}
We study the algorithmic complexity of partitioning the vertex set of a given (di)graph into a small number of paths. The \textsc{Path Partition} problem (\textsc{PP}) has been studied extensively, as it includes \textsc{Hamiltonian Path} as a special case. The natural variants where the paths are required to be either \emph{induced} (\textsc{Induced Path Partition}, \textsc{IPP}) or \emph{shortest} (\textsc{Shortest Path Partition}, \textsc{SPP}), have received much less attention. Both problems are known to be \NP-complete on undirected graphs; we strengthen this by showing that they remain so even on planar bipartite directed acyclic graphs (DAGs), and that \textsc{SPP} remains \NP-hard on undirected bipartite graphs. When parameterized by the natural parameter ``number of paths'', both \textsc{SPP} and \textsc{IPP} are shown to be \W{1}-hard on DAGs. We also show that SPP is in \XP both for DAGs and undirected graphs for the same parameter, as well as for other special subclasses of directed graphs (\textsc{IPP} is known to be \NP-hard on undirected graphs, even for two paths). On the positive side, we show that for undirected graphs, both problems are in \FPT, parameterized by neighborhood diversity. 
We also give an explicit algorithm for the vertex cover parameterization of \textsc{PP}. When considering the dual parameterization (graph order minus number of paths), all three variants, \textsc{IPP}, \textsc{SPP} and \textsc{PP}, are shown to be in \FPT for undirected graphs.
We also lift the mentioned neighborhood diversity and dual parameterization results to directed graphs; here, we need to define a proper novel notion of directed neighborhood diversity. As we also show, most of our results also transfer to the case of covering by edge-disjoint paths, and purely covering.
\end{abstract}

%



\maketitle       
\section{Introduction}

Graph partitioning and graph covering problems are among the most studied problems in graph theory and algorithms. There are several types of graph partitioning and covering problems including covering the vertex set by stars (\textsc{Dominating Set}), covering the vertex set by cliques (\textsc{Clique Covering}), partitioning the vertex set by independent sets (\textsc{Coloring}), and covering the vertex set by paths or cycles \cite{manuel2018revisiting}. In recent years, partitioning and covering problems by paths have received considerable attention in the literature because of their connections with well-known graph-theoretic theorems and conjectures like the Gallai-Milgram theorem~\cite{GM60}, Berge’s path partition conjecture~\cite{berge1983path,hartman1988variations} and a conjecture by Magnant and Martin \cite{DBLP:journals/ajc/MagnantM09}. These studies are motivated by applications in diverse areas such as code optimization~\cite{boesch1977covering}, machine learning / AI~\cite{TG21}, transportation networks~\cite{9568702}, bioinformatics~\cite{LafMou2024}, parallel computing~\cite{pinter1987mapping}, and program testing~\cite{ntafos1979path}. There are several types of paths that can be considered: unrestricted paths, induced paths, shortest paths, or directed paths (in a directed graph). A path~$P$ is an \emph{induced} path in $G$ if the subgraph induced by the vertices of~$P$ is a path. An induced path is also called a \emph{chordless} path; \emph{isometric} path and \emph{geodesic} are other names for shortest path. Various questions related to the complexity of these path problems, even in standard graph classes, remained open for a long time (even though they have uses in various fields), a good motivation for this project.

In this paper, we mainly study the problem of partitioning the vertex set of a graph (undirected or directed) into the minimum number of disjoint \emph{paths}, focusing on three problems, \textsc{Path Partition} (\textsc{PP}), \textsc{Induced Path Partition} (\textsc{IPP}) and \textsc{Shortest Path Partition} (\textsc{SPP}) --- formal definitions are given in \autoref{def}. A \emph{path partition} (pp) of a (directed) raph~$G$ is a partitioning of the vertex set into (directed) paths. The \emph{path partition number} of $G$ is the smallest size of a pp of~$G$. Similar definitions apply to ipp and spp. \textsc{PP} is studied extensively (often under the names \textsc{Path Cover} or \textsc{Hamiltonian Completion}) on many graph classes, see for example~\cite{boesch1974covering,DBLP:journals/siamdm/ChangK96,corneil2013ldfs,DamDKS92,franzblau_raychaudhuri_2002,goodman1974hamiltonian,kundu1976linear,DBLP:journals/dam/PanC05} and references therein.

We give a wide range of results in this paper including complexity (\NP- or \W{1}-hardness) and algorithms (polynomial time or \FPT); see \autoref{tab:survey-PPP} for a summary of our results. The types of graphs we consider are general directed, directed acyclic (DAG), general undirected and bipartite undirected graphs. 
We also consider some structural parameters like the neighborhood diversity 
and vertex cover number.

The three problems considered are all \NP-hard. They are mostly studied on undirected graphs, and unless stated otherwise, the following references are for undirected graphs. \textsc{PP} can be seen as an extension of \textsc{Hamiltonian Path} and is thus \NP-hard, even for one path (while for one path, \textsc{SPP} and \textsc{IPP} are very easy, as it suffices to check whether the whole graph is a path). Similarly, \textsc{IPP} is \NP-hard, even for two paths~\cite{le2003splitting}. Recently, \textsc{SPP} was proved to be \NP-hard~\cite{pmanuelisometric}, but note that, as opposed to the two other problems, it is polynomial-time solvable for any constant number of paths by an \XP algorithm~\cite{dumas2024graphs}. On trees, the three problems are equivalent, and are solvable in polynomial time (as shown in the 1970s
in~\cite{boesch1974covering,goodman1974hamiltonian,kundu1976linear,skupien1974path}, see also~\cite{franzblau_raychaudhuri_2002} for a more recent improved algorithm). As a special case of \textsc{Hamiltonian Path}, it is known that \textsc{PP} is \NP-hard for many graph classes such as planar graphs~\cite{DBLP:journals/siamcomp/GareyJT76} or chordal bipartite graphs~\cite{DBLP:journals/dm/Muller96a}. On the other hand, it can be solved in polynomial time for some graph classes, such as cographs~\cite{DBLP:journals/siamdm/ChangK96}, cocomparability graphs~\cite{corneil2013ldfs} or graphs whose blocks are cycles, complete graphs or complete bipartite graphs~\cite{DBLP:journals/dam/PanC05} (this includes trees, cactii and block graphs). We refer to~\cite{DBLP:journals/dam/PanC05} for further references. \textsc{SPP} remains \NP-hard on split graphs but becomes polynomial-time solvable on cographs and chain graphs~\cite{DBLP:conf/mfcs/ChakrabortyMOPR24}. \textsc{IPP} can also be solved in linear time for graphs whose blocks are cycles, complete graphs or complete bipartite graphs~\cite{DBLP:journals/tcs/PanC07} and on cographs and chain graphs~\cite{DBLP:conf/mfcs/ChakrabortyMOPR24}, but no further positive results for special graph classes seem to be known. For directed graphs, it is known that \textsc{PP} can be solved in polynomial time for DAGs~\cite[Probl.\,26-2]{CLRS3} but \textsc{IPP} is \NP-hard on DAGs, even for three paths~\cite{LafMou2024}.

For a detailed survey on these types of problems (both partitioning and covering versions), see~\cite{manuel2018revisiting}, and also~\cite{DBLP:journals/dam/AndreattaM95}. 

Although every shortest path is an induced path and every induced path is a path, the reverse are not true. In this sense, induced paths can be seen as an intermediate notion between shortest paths and unrestricted paths. However, this intuition may fail in certain contexts. Indeed, some seemingly simple problems related to paths are easy for unrestricted or shortest paths, but become hard for induced paths. For example, given a graph $G$ and three vertices $x,y,z$, determining whether there is an induced path from vertex $x$ to vertex $y$ going through vertex $z$ is \NP-complete~\cite{M10}, while this task is polynomial-time both for unrestricted paths and for shortest paths.

\medskip

\noindent\textbf{Further related work.} As sketched in~\cite{franzblau_raychaudhuri_2002}, the PP problem is equivalent to the \textsc{Hamiltonian Completion} problem, asking to add at most~$k$ edges/arcs to a (directed) graph to guarantee the existence of a Hamiltonian path. Besides \textsc{Hamiltonian Path}, \textsc{PP} is also directly connected to other graph problems such as $L(2,1)$-labelings~\cite{PP-L21} or the metric dimension~\cite{MD-PP}. It is the object of the Gallai-Milgram theorem, stating that for any minimal path partition of a directed graph, there is an independent set intersecting each path~\cite{GM60}, see also~\cite[Section 2.5]{D17}.

Variations of PP and IPP where the paths must have a prescribed length are also studied, see for example~\cite{DBLP:journals/ol/LiYL24,monnot2007path,DBLP:journals/tcs/Steiner03} for undirected graphs, and~\cite{DBLP:journals/iandc/ChenCKLXZ24,DBLP:conf/iwoca/EtoKLMO24} for directed graphs.

The covering versions of PP and IPP (where the paths need not necessarily be disjoint) seem to be very little studied~\cite{manuel2018revisiting}, however we can mention some works for PP on DAGs, which is polynomial-time solvable and has numerous practical applications (see~\cite{dagPC} and references therein). On the other hand, the covering version of \textsc{SPP} was recently studied (often under the name \textsc{Isometric Path Cover}), see \cite{DBLP:journals/ipl/PanC05} for a linear-time algorithm on block graphs, \cite{DBLP:conf/mfcs/ChakrabortyMOPR24} for algorithms on cographs and chain graphs, \cite{dumas2024graphs} for an \XP algorithm, \cite{foucaud2022} for \NP-hardness on chordal graphs and approximation algorithms for chordal graphs and other classes, and \cite{TG21} for a $\log n$-factor approximation algorithm. This problem is connected to the \emph{Cops and Robber game}, where a robber and a set of cops alternatively move in a graph (a move is to reach a neighboring vertex); the robber tries to evade the cops indefinitely, while they try to reach (catch) the robber. As shown in~\cite{AF84,FF01}, if we are given a set of shortest paths covering all the vertices, we can assign a cop to each path and let him patrol along it (while staying as close as possible to the robber). Eventually, they will catch the robber, and thus the smallest size of such a shortest path cover is an upper bound to the smallest number of cops required to catch the robber. This problem is also connected to decomposition results in structural graph theory~\cite{dujmovic2020planar} and to machine learning applications~\cite{TG21,TG22}.

The versions of PP and SPP where, on the other hand, covering is not required (and the endpoints of the solution paths, called \emph{terminals}, are prescribed in the input) are widely studied 
as \textsc{Disjoint Paths (DP)} \cite{ROBERTSON199565} (sometimes also called \textsc{Linkage}) and 
\textsc{Disjoint Shortest Paths (DSP)} \cite{DBLP:journals/siamdm/BentertNRZ23,lochet2021polynomial}. Both these problems are \NP-complete and have applications in network routing problems and VLSI design~\cite{PFVLSIbook}. \textsc{DP} in particular has been extensively studied, due to its deep connections with structural graph theory: as part of the celebrated Graph Minor project, Robertson and Seymour showed that for undirected graphs, \textsc{DP} is in \FPT, parameterized by the number of paths~\cite{ROBERTSON199565}, contrasting \textsc{PP}. An improved \FPT algorithm was given in~\cite{kawarabayashi2012disjoint}. Algorithms were also designed for planar directed graphs~\cite{ForHopWyl80} (\FPT), graphs of bounded directed treewidth~\cite{DBLP:journals/jct/JohnsonRST01} (\XP) and other classes~\cite{DBLP:journals/jgt/Bang-JensenCM17}. Recently, \textsc{DSP} on undirected graphs, which has applications in artificial intelligence~\cite{GDLHSH18}, was shown to have an \XP algorithm and to be \W{1}-hard when parameterized by the number of paths~\cite{lochet2021polynomial}; an improved algorithm was given in~\cite{DBLP:journals/siamdm/BentertNRZ23}, and the directed graph case is studied in~\cite{DBLP:conf/esa/Berczi017}.

A strengthening of the DP problem (motivated by the task of detecting induced subgraphs or induced subdivisions of graphs) is studied under the name \textsc{Induced Disjoint Paths}~\cite{GolPauLee2022,KK12,MarPSL2023}: here, given the terminals, one is looking for disjoint unrestricted paths connecting the terminals, but moreover, no edge should connect two vertices on different paths. This problem is \NP-hard even for two paths~\cite{KK12}.

Some further restricted variations around disjoint paths are studied for example in~\cite{DBLP:journals/algorithmica/BelmonteHKKKKLO22} (for unrestricted paths) and in~\cite{DBLP:journals/algorithmica/AraujoCMSS20} (for induced paths).

\begin{table}[tbh] \centering  
 \begin{tabular}{|c|c|c|c|}
 \hline
 \text{parameter}  & \textsc{PP}  & \textsc{SPP} & \textsc{IPP} \\ \hline\hline
    none (UG/DG)
    & \text{\NP-c.}\,\cite{GarJoh79}         & \text{\NP-c.\,\cite{pmanuelisometric} } & \text{\NP-c.}\,\cite{le2003splitting} \\ \hline
       none (bipartite UG)      & \text{\NP-c. \cite{krishnamoorthy1975np,le2003splitting}}   & \textbf{\NP-c.} & open                                                     \\ \hline
    solution size~$k$ (UG) & \text{\paraNP-h.} \cite{GarJoh79}    & \textbf{in \XP}                                                           & \text{\paraNP-h.} \cite{le2003splitting}                                               \\ \hline
        solution size~$k$ (DG) & \text{\paraNP-h.} \cite{GarJoh79}    & \begin{tabular}[c]{@{}c@{}}open\\ \textbf{in \XP(SG)}\end{tabular}                                                           & \text{\paraNP-h.} \cite{le2003splitting}                                               \\ \hline
    solution size~$k$ (DAG)        & \begin{tabular}[c]{@{}c@{}}polynomial,\\ see~\cite{CLRS3}, \\Problem 26-2 \end{tabular}   & \begin{tabular}[c]{@{}c@{}}\textbf{\NP-c.}\\ \textbf{\W{1}-h.}\\ \textbf{in \XP}\end{tabular} & \begin{tabular}[c]{@{}c@{}}\textbf{\NP-c.}\\ \textbf{\W{1}-h.} \\ \text{\paraNP-h.} \cite{LafMou2024}\end{tabular} \\ \hline

    \scalebox{.9}{neighborhood diversity (UG)}             & {\FPT} \cite{GajLamOrd2013}          & \textbf{\FPT}                                                        & 
    \textbf{\FPT}                                                   \\ 
    \hline
        \scalebox{.9}{neighborhood diversity (DG)}     &  open         & \textbf{\FPT}                                                        &     \textbf{\FPT}                                                   \\     \hline
        \scalebox{.86}{vertex cover number (UG/DG)}     &  \textbf{\FPT}         & \textbf{\FPT}                                                        &     \textbf{\FPT}                                                   \\\hline

dual~$n-k$  (UG/DG)         & \textbf{\FPT}         & \textbf{\FPT}                                                        & \textbf{\FPT}                                                 \\ \hline
    \end{tabular}

    \caption{Summary of known results concerning path partitioning problems. The abbreviations c. and h. refer to completeness and hardness, respectively. In parentheses, we put further input specifications, with UG, DG and SG referring to undirected graphs, directed graphs and special graph classes, respectively. Our results are highlighted in bold face.\\[-5ex]}
    \label{tab:survey-PPP}
\end{table}

\medskip

\noindent\textbf{Our contribution.} 
\autoref{tab:survey-PPP} summarizes known results about the three problems, with our results are highlighted in bold. We fill in most of the hitherto open questions concerning variations of \textsc{PP}, \textsc{SPP} and \textsc{IPP}, e.g., 
we show that \textsc{SPP} has a poly-time algorithm for a fixed number of paths (on undirected graphs, DAGs, and special classes of directed graphs). This is surprising as both \textsc{PP} and \textsc{IPP} are \NP-hard on undirected graphs for $k = 1$ and for $k = 2$, respectively, see~\cite{GarJoh79} and~\cite{le2003splitting}. To prove this, we use existing \XP (or faster) algorithms for the related \textsc{Disjoint (Shortest) Paths} problem. 
Many of our results concern our problems restricted to DAGs. Notice that \textsc{PP} has a polynomial-time algorithm when restricted to DAGs (using Maximum Matching~\cite[Probl.\,26-2]{CLRS3}), but the complexity for \textsc{IPP} and \textsc{SPP} was open for such inputs. We show that \textsc{IPP} and \textsc{SPP} are \NP-hard even when restricted to planar DAGs whose underlying graph is bipartite and have maximum degree~4.\footnote{We remark that in the conference version of this paper~\cite{FernauFMPN23}, this result was stated by mistake for graphs of maximum degree~3, but the proof is the same here and the statement is corrected.} We strengthen this result using a similar construction as in the classic proof of \textsc{DP} being \W{1}-hard on DAGs by Slivkins~\cite{Slivkins2010}, to show that \textsc{IPP} and \textsc{SPP} are \W{1}-hard on DAGs. As mentioned above,  Manuel recently showed that \textsc{SPP} is \NP-hard on undirected graphs \cite{pmanuelisometric}. We extend this result to show that \textsc{SPP} is \NP-hard even when restricted to bipartite graphs that are sparse (they have degeneracy at most~5). The complexity of these problems has not yet been studied when parameterized by structural parameters. 
We show that \textsc{IPP} and \textsc{SPP} both belong to \FPT when parameterized by standard structural parameters like vertex cover and neighborhood diversity using the technique of \textsc{Integer Linear Programming}. 
Also, we lift these considerations to directed graphs. To this end, we introduce the notion of \emph{directed neighborhood diversity} in this paper, which may be useful for other problems on directed graphs. We also obtain easy \FPT\ algorithms for \textsc{PP}, \textsc{IPP} and \textsc{SPP} on both undirected and directed graphs when parameterized by the vertex cover number of the (underlying) graph. We further obtain a non-trivial \FPT\ algorithm for \textsc{PP} on undirected graphs~$G$ with vertex cover number $\vc(G)$, with the improved running time $\Oh^*\left(2^{\Oh(\vc(G)\log(\vc(G)))}\right)$. 
Moreover, when considering the dual parameterization (graph order minus number of paths), we show that all three variants, \textsc{PP}, \textsc{IPP},  and \textsc{SPP}, are in \FPT for undirected graphs. This is also the case for directed graphs (except for \textsc{PP}, which is left open). Finally, we also remark that many of our results apply to the purely ``vertex covering'' versions of \textsc{PP}, \textsc{IPP},  and \textsc{SPP}, and to their intermediate ``vertex covering using edge-disjoint paths'' versions.

It is interesting to note the differences in the complexities of the three problems on different input classes. For example, we will see that \textsc{SPP} can be solved in \XP-time on undirected graphs when parameterized by solution size, while this is not possible for the other two problems. For DAGs, \textsc{PP} is polynomial-time solvable, but the other two problems are \NP-hard. In a way, as any shortest path is induced, \textsc{IPP} can be seen as an intermediate problem between \textsc{PP} and \textsc{SPP}. Thus, it is not too surprising that, when the complexity of the three problems differs, \textsc{IPP} sometimes behaves like \textsc{PP}, and sometimes, like \textsc{SPP}. It would be interesting to find a situation in which all three problems behave differently.

\medskip

\noindent\textbf{Organization of the paper.} We start with formally defining the studied problems in \autoref{def}.  Then we focus our attention on \NP-hardness results in \autoref{np-hard}, showing that \textsc{IPP} and \textsc{SPP} are \NP-hard when the input graph is a DAG, and that \textsc{SPP} is \NP-hard when the input graph is a bipartite graph. Our main result showing that \textsc{IPP} and \textsc{SPP} are \W{1}-hard for the standard parameter~$k$ (number of paths) is found in \autoref{w1hard}. We give an \XP algorithm for \textsc{SPP} in \autoref{xp-algo}. We then give an \FPT algorithm for \textsc{IPP} and \textsc{SPP} when parameterized by neighborhood diversity in \autoref{nd-param}, as well as a non-trivial direct \FPT algorithm for \textsc{PP}, parameterized by vertex cover. These results are then applied in \autoref{sec:duals} to prove \FPT results for the dual parameterization. In \autoref{sec:ED}, we discuss the (edge-disjoint) covering versions of the problems, and argue that most of our results translate to this setting as well. We conclude in \autoref{conclusion}.

\section{Definitions} \label{def}

We are using standard terminology concerning graphs, classical and parameterized complexity and we will not iterate this standard terminology here. In particular, a \emph{path} $P$ can be described by a sequence of non-repeated vertices such that there is an edge between vertices that are neighbors in this sequence. Sometimes, it is convenient to consider $P$ as a set of vertices. 
We are next defining the problems considered in this paper. All problems can be considered on undirected or directed graphs, or also on directed acyclic graphs (DAG). We will specify this by prefixing \textsc{U}, \textsc{D}, or \textsc{DAG} to our problem name abbreviations.

We say that a sub-graph~$G'$ of $G=(V,E)$ \emph{spans} $G$ if its vertex set is~$V$. 
In other words, $G'$ is a partial graph of~$G$. 
\problemdef{Path Partition}{(\textsc{PP} for short)}{A graph $G$, a non-negative integer~$k$}{Are there pairwise vertex-disjoint paths $P_1, \dots , P_{k'}$, with $k'\leq k$, such that, together, these paths span~$G$\,?}

A path~$P$ is an \emph{induced path} in~$G$ if the induced graph $G[P]$ is a path; here, $P$ is considered as a vertex set.   

\problemdef{Induced Path Partition}{(\textsc{IPP} for short)}{A graph $G$, a non-negative integer~$k$}{Are there pairwise vertex-disjoint paths $P_1, \dots , P_{k'}$, with $k'\leq k$, that are induced paths, and such that, together, these paths span~$G$\,?}

A \emph{shortest path} is a path with end-points $u,v$ that is shortest among all paths from $u$ to~$v$. The greatest length of any shortest path in a graph~$G$ is also known as its \emph{diameter}, written as $\diam(G)$.

\problemdef{Shortest Path Partition}{(\textsc{SPP} for short)}{A graph $G$, a non-negative integer~$k$}{Are there pairwise vertex-disjoint paths $P_1, \dots , P_{k'}$, with $k'\leq k$, that are shortest paths, and such that, together, these paths span~$G$\,?}


\begin{remark}
\label{rem:combi}
The following known combinatorial properties are helpful to connect the different yet related problems:
\begin{enumerate}
 \item 
Every shortest path is also an induced path.
\item 
 Every induced path of length at most two is also a shortest path.
\item  In bipartite graphs, an induced path of length three is a shortest path.
\item 
If $G=(V,E)$ is a subgraph of $H$ and $P$ is a shortest path in $H$ with only vertices of $V$, then $P$ is also a shortest path in~$G$.
\end{enumerate}    
\end{remark}

Omitting the vertex-disjointness condition, we arrive at three \emph{covering versions} instead:
\problemdef{Path Cover}{(\textsc{PC} for short)}{A graph $G$, a non-negative integer~$k$}{Are there  paths $P_1, \dots , P_{k'}$, with $k'\leq k$, such that, together, these paths span~$G$\,?}

\problemdef{Induced Path Cover}{(\textsc{IPC} for short)}{A graph $G$, a non-negative integer~$k$}{Are there  paths $P_1, \dots , P_{k'}$, with $k'\leq k$, that are induced paths and such that, together, these paths span~$G$\,?}

\problemdef{Shortest Path Cover}{(\textsc{SPC} for short)}{A graph $G$, a non-negative integer~$k$}{Are there  paths $P_1, \dots , P_{k'}$, with $k'\leq k$, that are shortest paths and such that, together, these paths span~$G$\,?}


Let us also define the non-covering versions of \textsc{PP}, \textsc{IPP} and \textsc{SPP}, that will be used in the paper.

\problemdef{Disjoint Paths}{(\textsc{DP} for short)}{A graph $G$, pairs of terminal vertices $\{(s_1, t_1) , \dots, (s_k , t_k) \}$}{Are there pairwise vertex-disjoint paths $P_1, \dots , P_k$ such that, for $1\leq i\leq k$, the end-points of~$P_i$ are $s_i$~and~$t_i$?}

Analogously to \textsc{DP}, we can define the problems \textsc{Disjoint Induced Paths}\footnote{This problem should not be confused with \textsc{Induced Disjoint Paths}, see \cite{GolPauLee2022,KK12,MarPSL2023}, where it is required that any two (unrestricted) solution paths have no adjacent vertices.} and \textsc{Disjoint Shortest Paths} as follows.

\problemdef{Disjoint Induced Paths}{(\textsc{DIP} for short)}{A graph $G$, pairs of terminal vertices $\{(s_1, t_1) , \dots, (s_k , t_k) \}$}{Are there pairwise vertex-disjoint induced paths $P_1, \dots , P_k$ such that, for $1\leq i\leq k$, the end-points of~$P_i$ are $s_i$~and~$t_i$?}

\problemdef{Disjoint Shortest Paths}{(\textsc{DSP} for short)}{A graph $G$, pairs of terminal vertices $\{(s_1, t_1) , \dots, (s_k , t_k) \}$}{Are there pairwise vertex-disjoint shortest paths $P_1, \dots , P_k$ such that, for $1\leq i\leq k$, the end-points of~$P_i$ are $s_i$~and~$t_i$?}

\begin{remark} \label{rem:DP-DIP}
The problems \textsc{DIP} and \textsc{DP} are \emph{equivalent}  when the inputs are undirected graphs or DAGs (in the sense that the \yes-instances are the same for both problems), but not for general directed graphs.
\end{remark}
\begin{proof}
Namely, if an input is a \yes-instance of \textsc{DIP}, then trivially it is also a \yes-instance of \textsc{DP}. But also if an input is a \yes-instance of \textsc{UDP} or \textsc{DAGDP}, then it is also a \yes-instance of \textsc{UDIP} or \textsc{DAGDIP}, as we could take each path~$P$ in the solution~$\mathbb P$ of \textsc{UDP} or \textsc{DAGDP} and find a path~$P'$ induced by some subset of vertices of~$P$; this can be done by starting a breadth-first search at one endpoint of the path~$P$ to reach the other endpoint and choosing a shortest path. This way, we get a set~$\mathbb{P}'$ of induced paths as a solution to \textsc{UDIP} or \textsc{DAGDIP}. This does not hold true for general directed graphs, as a solution to \textsc{DDP} does not imply a solution to \textsc{DDIP}. For example, consider a directed graph with three vertices $a_1,a_2,a_3$ and arcs $(a_1,a_2), (a_2,a_3), (a_3,a_1)$ with $(s_1,t_1) = (a_1,a_3)$ as the prescribed path endpoints.
\end{proof}

Interestingly, and contrasting \autoref{rem:DP-DIP}, we will see that the complexities of \textsc{DAGPP} and \textsc{DAGIPP} differ drastically, see \autoref{tab:survey-PPP}. 

\subsection{Graph notions}

A graph is \emph{$d$-degenerate} if every induced subgraph has a vertex of degree at most $d$.
The \emph{vertex cover number} of a graph $G$ is the smallest size of a vertex cover of $G$, that is, a set $S$ of vertices of $G$ such that each edge of $G$ intersects $S$.
The \emph{neighborhood diversity}~\cite{lampis2012algorithmic} of a graph $G$ (or just $\nd(G)$) is the number of equivalence classes of the following equivalence relation: two vertices $u, v \in V$ are equivalent (we also say that they have the \emph{same type}) if they have the same neighborhoods except for possibly themselves, i.e., if $N(v) \setminus \{u\} = N(u) \setminus \{v\}$. The equivalence classes, which we call \emph{neighborhood diversity classes}, form either cliques or independent sets and all vertices in one class are pairwise \emph{twins} (vertices with either the same closed neighborhood or the same open neighborhood)~\cite{Kou2013}. More information on this parameter can be found in~\cite{Kou2013}. Notice that $V$ can be partitioned into vertices of the same type in linear time using partition refinement~\cite[Algorithm 2]{DBLP:conf/stacs/HabibPV98}.

\section{\NP-hardness results} \label{np-hard}

\noindent
This section contains several \NP-hardness reductions for the studied problems.

\subsection{SPP and IPP on DAGs}

A well-known result related to $\textsc{PC}$ in DAGs is Dilworth's theorem: the minimal size of a directed path cover equals the maximal cardinality of an independent set (or, in the original language of posets, the minimum size of a chain cover is equal to the maximum size of an antichain)~\cite{dilworth2009decomposition}. Fulkerson~\cite{fulkerson1956note} gave a constructive proof of this theorem, from which it follows that the \textsc{DAGPC} problem can be reduced to a maximum matching problem in a bipartite graph. See also~\cite{dagPC} for further information. Even \textsc{DAGPP} can be solved in polynomial time by a similar method~\cite[Problem 26-2]{CLRS3}. We show that, in contrast, \textsc{DAGSPP} and \textsc{DAGIPP} are \NP-hard even when restricted to planar bipartite DAGs.

\begin{theorem}\label{thm:DAGSPP-NPhard}
\textsc{DAGSPP} and \textsc{DAGIPP} are \NP-hard even when the inputs are restricted to planar bipartite DAGs of maximum degree~4.
\end{theorem}

\begin{proof}
Our reduction is adapted from \cite{monnot2007path,DBLP:journals/tcs/Steiner03}.
We  reduce from the \textsc{Planar $3$-Dimensional Matching} problem, or \textsc{Planar $3$-DM}, which is \NP-complete (see~\cite{dyer1986planar}), even when each element occurs in either two or three triples.
 A \textsc{$3$-DM} instance consists of three disjoint sets $X, Y, Z$ of equal cardinality $p$ and a set~$T$ of triples from $X \times Y \times Z$. Let $q= |T|$. The question is if there are $p$ triples which contain all elements of $X$, $Y$ and~$Z$. We associate a bipartite graph with this instance.
We assume that the four sets $T$, $X$, $Y$ and $Z$ are pairwise disjoint. We also assume that each element of $X\cup Y\cup Z$ belongs to at most three triples. We have a vertex for each element in $X, Y, Z$ and each triple in $T$. There is an edge connecting triples to elements if and only if the element belongs to the triple. This graph~$G$ is bipartite with vertex bipartition of $T, X \cup Y \cup Z$, and has maximum degree~3. We say the instance is planar if $G$ is planar. Given an instance of \textsc{Planar $3$-DM}, $G = (T , X \cup Y \cup Z , E)$,  and a planar embedding of it,\footnote{This planar embedding can be computed in linear time if necessary~\cite{DBLP:journals/tcs/ShihH99}.}.  we build an instance $G' = (V', E')$ of \textsc{DAGSPP}. 

\smallskip\noindent
\textbf{Construction:} We replace each $v_i = (x, y, z) \in T$, where $x \in X$, $y \in Y$, $z \in Z$, with a gadget $H(v_i)$ that consists of $9$ vertices named $l_{jk}^{i}$ where $1 \leq j,k \leq 3 $ and with edges as shown in  \autoref{fig:DAGSPP-vertex-gadget}; if the planar embedding has $x, y, z$ in clockwise order seen as neighbors of~$v_i$, then we add the arcs $(l_{12}^{i},x)$, $(l_{22}^{i},z)$ and $(l_{32}^{i},y)$, otherwise, we add the arcs  $(l_{12}^{i},x)$, $(l_{22}^{i},y)$ and $(l_{32}^{i},z)$.  

\begin{figure}[tbh]
    \centering
   \scalebox{1.1}{\tikzset{every picture/.style={line width=0.75pt}} 

\begin{tikzpicture}[
roundnode/.style={circle, draw=black,  thick, minimum size=5mm, scale=0.5}]

\node[roundnode,scale = 1.4] at (-5.5,1)     (x1)                              {$x$};

\node[roundnode,scale = 1.4] at (-4.3,1)     (y1)                              {$y$};

\node[roundnode,scale = 1.4] at (-2.6,1)     (z1)                              {$z$};

\node[roundnode] at (-6.5,2)     (l11)                              {$l_{11}^{i}$};

\node[roundnode] at (-5.5,2)     (l12)                              {$l_{12}^{i}$};

\node[roundnode] at (-5.5,3)     (l13)                              {$l_{13}^{i}$};

\node[roundnode] at (-3.3,2)     (l21)                              {$l_{21}^{i}$};

\node[roundnode] at (-4.3,2)     (l22)                              {$l_{22}^{i}$};

\node[roundnode] at (-4.3,3)     (l23)                              {$l_{23}^{i}$};

\node[roundnode] at (-1.6,2)     (l31)                              {$l_{31}^{i}$};

\node[roundnode] at (-2.6,2)     (l32)                              {$l_{32}^{i}$};

\node[roundnode] at (-2.6,3)     (l33)                              {$l_{33}^{i}$};

\draw[-{Stealth[scale=0.8]}] (l11) to (l12);
\draw[-{Stealth[scale=0.8]}] (l12) to (l13);
\draw[-{Stealth[scale=0.8]}] (l12) to (x1);

\draw[-{Stealth[scale=0.8]}] (l21) to (l22);
\draw[-{Stealth[scale=0.8]}] (l22) to (l23);
\draw[-{Stealth[scale=0.8]}] (l22) to (y1);
\draw[-{Stealth[scale=0.8]}] (l21) to (l33);

\draw[-{Stealth[scale=0.8]}] (l31) to (l32);
\draw[-{Stealth[scale=0.8]}] (l32) to (l33);
\draw[-{Stealth[scale=0.8]}] (l32) to (z1);

\draw[-{Stealth[scale=0.8]}] (l13) to (l23);
\draw[-{Stealth[scale=0.8]}] (l23) to (l33);

\draw [-{Stealth[scale=0.8]}] (l11.north) to [out=90,in=150] (l23.north);

\end{tikzpicture}}   
    \caption{The vertex gadget, replacing $v_i$ in~$G$ with nine vertices in $G'$,} as defined in the proof of \autoref{thm:DAGSPP-NPhard}.
    \label{fig:DAGSPP-vertex-gadget}
\end{figure}
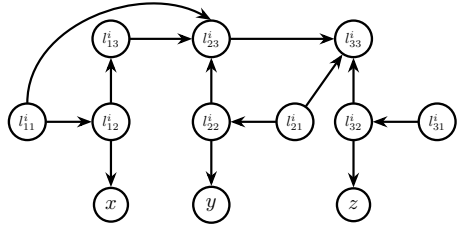

\smallskip\noindent
We observe the following two properties of $G'$.

\begin{claim}\label{claim:DAG-reduction-props}
$G'$ is a planar DAG with maximum degree~4 in which every shortest/induced path is of length at most $3$, and the underlying undirected graph of~$G'$ is bipartite.
\end{claim}
\begin{proof}
    The first four conditions are easy to see in the given construction. We will prove the bipartiteness claim by giving a proper 2-coloring. We will color all vertices which represent the elements of sets $X$ and~$Z$ with color $1$, and all vertices which represent elements of set~$Y$ with color~$2$. Then we observe that within the gadget $H(v_i)$, the vertices $l_{12},l_{23},l_{21}, l_{32}$ have color $2$, and the remaining vertices have color~$1$. This  coloring is a proper $2$-coloring of~$G'$. \renewcommand{\qedsymbol}{$\Diamond$}
\end{proof}

\begin{claim}\label{claim:DAG-reduction-mainclaim}
The \textsc{Planar $3$-DM} instance has a solution if and only if $G'$ can be partitioned into $p + 3q$ shortest/induced paths.
\end{claim}
\begin{proof}
Since $|V(G')| = 3p + 9q $, if $G'$ can be partitioned into $p + 3q$ shortest/induced paths, we observe that the average path length has to be 2. Observe also that none of the paths can contain more than one of the original elements of $X \cup Y \cup Z$. 

The only possible shortest/induced paths of length 3 are of the form $P^i = l_{12}^{i} l_{13}^{i} l_{23}^{i} l_{33}^{i}$, but if one takes this path into the solution, to cover vertex $l_{11}^{i}$, one has to include this vertex as a single vertex path, hence the average path length of such a solution would be strictly less than 2. We can also eliminate paths $P^i_1 = l_{11}^{i} l_{23}^{i} l_{33}^{i}$ and $P^i_2 = l_{22}^{i} l_{23}^{i} l_{33}^{i}$ from the solution of $G'$ by similar arguments. More precisely, taking~$P^i_1$ would force us to also take the path $l_{12}^{i}l_{13}^{i}$, and taking~$P^i_2$ even makes us select the singe vertex path $l_{21}^{i}$. Again, the average path length of such a solution would be strictly less than~$2$.

In summary, each shortest or induced path in the solution contains exactly three vertices, and each gadget $H(v_i)$ is partitioned into $P_3$-paths in one of the two ways shown in \autoref{twopart}. 

We now describe how to read off a solution of the original \textsc{Planar $3$-DM} instance as described by~$G$ from a shortest/induced path partition of~$G'$ into $p + 3q$ shortest/induced paths.
If a gadget $H(v_i)$ is partitioned as in $(2)$ in \autoref{twopart}, we select that triple $v_i = (x, y, z) \in T$ into the solution~$S$ of~$G$,  
else we do not select the triple into the solution. Since the solution covers all vertices of $G'$, it will cover in particular all vertices from $X\cup Y\cup Z$; thus, the selected set~$S$ of triples will form a \textsc{Planar $3$-DM} solution.

Conversely, if $S$ is a solution of the \textsc{Planar $3$-DM} instance specified by~$G$, then for $v_i = (x, y, z) \in S$, partition the corresponding $H(v_i)$ as described in  $(2)$ of \autoref{twopart}, and partition every other $H(v_j)$ as shown in $(1)$ of \autoref{twopart}. This provides a shortest/induced path partition of~$G'$ into $p + 3q$ shortest paths.\renewcommand{\qedsymbol}{$\Diamond$}
\end{proof}
 
This completes our proof of the theorem.\end{proof}

\begin{figure}[!tbp] 
  \centering
  
\scalebox{1.1}{
\tikzset{every picture/.style={line width=0.75pt}} 

\begin{tikzpicture}[
roundnode/.style={circle, draw=black,  thick, minimum size=5mm, scale=0.55}]

\node[roundnode,scale = 1.4] at (-5,1)     (x1)                              {$x$};

\node[roundnode,scale = 1.4] at (-3.3,1)     (y1)                              {$y$};

\node[roundnode,scale = 1.4] at (-1.6,1)     (z1)                              {$z$};

\node[roundnode] at (-6,2)     (l11)                              {$l_{11}^{i}$};

\node[roundnode] at (-5,2)     (l12)                              {$l_{12}^{i}$};

\node[roundnode] at (-5,3)     (l13)                              {$l_{13}^{i}$};

\node[roundnode] at (-4.3,2)     (l21)                              {$l_{21}^{i}$};

\node[roundnode] at (-3.3,2)     (l22)                              {$l_{22}^{i}$};

\node[roundnode] at (-3.3,3)     (l23)                              {$l_{23}^{i}$};

\node[roundnode] at (-2.6,2)     (l31)                              {$l_{31}^{i}$};

\node[roundnode] at (-1.6,2)     (l32)                              {$l_{32}^{i}$};

\node[roundnode] at (-1.6,3)     (l33)                              {$l_{33}^{i}$};

\node[roundnode,scale = 1.4] at (1,1)     (x2)                              {$x$};

\node[roundnode,scale = 1.4] at (2.7,1)     (y2)                              {$y$};

\node[roundnode,scale = 1.4] at (4.4,1)     (z2)                              {$z$};

\node[roundnode] at (0,2)     (a11)                              {$l_{11}^{i}$};

\node[roundnode] at (1,2)     (a12)                              {$l_{12}^{i}$};

\node[roundnode] at (1,3)     (a13)                              {$l_{13}^{i}$};

\node[roundnode] at (1.7,2)     (a21)                              {$l_{21}^{i}$};

\node[roundnode] at (2.7,2)     (a22)                              {$l_{22}^{i}$};

\node[roundnode] at (2.7,3)     (a23)                              {$l_{23}^{i}$};

\node[roundnode] at (3.4,2)     (a31)                              {$l_{31}^{i}$};

\node[roundnode] at (4.4,2)     (a32)                              {$l_{32}^{i}$};

\node[roundnode] at (4.4,3)     (a33)                              {$l_{33}^{i}$};

\draw[-{Stealth[scale=1]}] (l11) to (l12);
\draw[-{Stealth[scale=1]}] (l12) to (l13);
\draw[-{Stealth[scale=1]}] (l21) to (l22);
\draw[-{Stealth[scale=1]}] (l22) to (l23);
\draw[-{Stealth[scale=1]}] (l31) to (l32);
\draw[-{Stealth[scale=1]}] (l32) to (l33);

\draw[-{Stealth[scale=1]}] (a11) to (a12);
\draw[-{Stealth[scale=1]}] (a12) to (x2);
\draw[-{Stealth[scale=1]}] (a21) to (a22);
\draw[-{Stealth[scale=1]}] (a22) to (y2);
\draw[-{Stealth[scale=1]}] (a31) to (a32);
\draw[-{Stealth[scale=1]}] (a32) to (z2);

\draw[-{Stealth[scale=1]}] (a13) to (a23);
\draw[-{Stealth[scale=1]}] (a23) to (a33);

\node at (-3.3,0.3) {$(1)$};
\node at (2.7,0.3) {$(2)$};
\end{tikzpicture}}

   \caption{Two different vertex partitions of a $H(v_i)$ gadget into 3-vertex paths, corresponding to different triple selections in the construction of \autoref{thm:DAGSPP-NPhard}.} \label{twopart}
\end{figure}
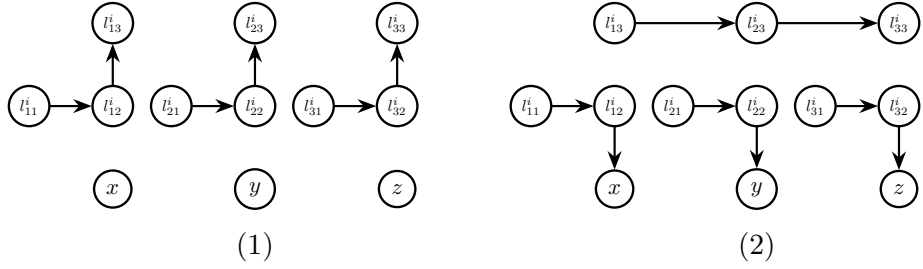

\subsection{SPP for bipartite undirected graphs}

Next, we prove that \textsc{USPP} is \NP-hard even when the input graph is restricted to bipartite $5$-degenerate graphs with diameter at most~4. To prove this, we reduce from \textsc{$4$-USPP} on bipartite graphs to \textsc{USPP} on bipartite graphs. \textsc{$4$-USPP} asks, given $G=(V,E)$, if there exists a partition~$\mathbb{P}$ of~$V$ such that each set in~$\mathbb{P}$ induces a shortest path of length 3 in~$G$. First, we show that \textsc{$4$-USPP} is \NP-hard on bipartite graphs (\autoref{exact4np}) by a reduction from \textsc{$4$-UIPP} (also known as 
\textsc{Induced $P_4$-Partition}) on bipartite graphs. 
\textsc{$4$-UIPP} asks if there exists a partition~$\mathbb{P}$ of~$V$ such that each set in $\mathbb{P}$ induces a path of length~3 in~$G$.

\begin{lemma}[\cite{monnot2007path}] \label{induced}
\textsc{$4$-UIPP} is \NP-hard for bipartite graphs of maximum degree~$3$.
\end{lemma} 

\begin{lemma} \label{exact4np}
\textsc{$4$-USPP} is \NP-hard for bipartite graphs of maximum degree~$3$.
\end{lemma}
\begin{proof}
This follows from \autoref{induced} and the fact that, in a bipartite graph, a path of length three is an induced path if and only if  
it is a shortest path, see \autoref{rem:combi}. 
\end{proof} 

\begin{theorem}\label{thm:SPP-NP-hardness}
\textsc{USPP} is \NP-hard, even for bipartite 5-degenerate graphs with diameter 4.
\end{theorem}
\begin{proof}
To prove this claim, we use  \autoref{exact4np}. Given an instance of \textsc{$4$-USPP}, say, $G = (V, E)$, with bipartition $V=A \cup B$ and $|V| = 4k$, as the number of vertices must be divisible by 4, we create an instance $G'=(V',E')$ of \textsc{USPP}.

\smallskip\noindent
\textbf{Construction:}  We add $10$ new vertices to $G$, getting $$V' = V \cup \{x_1,x_2,x_3,x_4,x_5,y_1,y_2,y_3,y_4,y_5\}\,.$$ We add edges from $x_2$ and $x_4$ to all vertices of $B\cup\{y_2,y_4\}$. Also, add edges from $y_2$ and~$y_4$ to all vertices of~$A$, 
add further edges to form paths $x_1x_2x_3x_4x_5$ and $y_1y_2y_3y_4y_5$.
The remaining edges all stem from~$G$. See \autoref{bipartitefig} for an illustration. 
This describes $E'$ of $G'$.


\begin{claim}\label{obs:bipartite_degenerate}
$G'=(V',E')$ is bipartite and 5-degenerate.
\end{claim}
\begin{proof}
We can make the claimed bipartition explicit by writing  $V'=A' \cup B'$, where $A'=A \cup \{x_2,x_4,y_1,y_3,y_5\}$ and $B'=B \cup \{x_1,x_3,x_5,y_2,y_4\}$ (see \autoref{bipartitefig}).

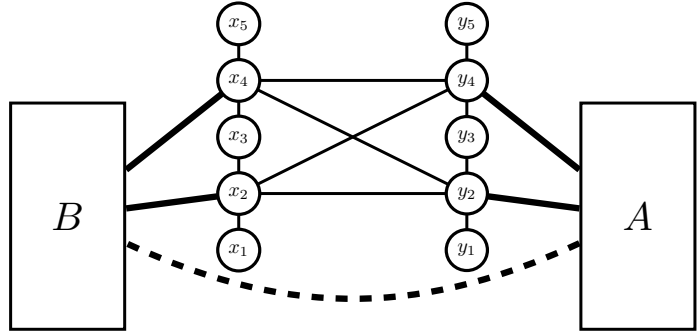
\begin{figure}[tbh]
    \centering
   \scalebox{1.5}{
\tikzset{every picture/.style={line width=0.75pt}} 
\begin{tikzpicture}[
roundnode/.style={circle, draw=black,  thick, minimum size=5mm, scale=0.5},
squarednode/.style={rectangle, draw=black, thick, minimum size=5mm},
]

\node[roundnode] at (0.5,0)     (x1)                              { $x_1$};

\node[roundnode] at (0.5,0.5)     (x2)                              { $x_2$};

\node[roundnode] at (0.5,1)     (x3)                              { $x_3$};

\node[roundnode] at (0.5,1.5)     (x4)                              { $x_4$};
\node[roundnode] at (0.5,2)     (x5)                              { $x_5$};

\node[roundnode] at (2.5,0)     (y1)                              { $y_1$};

\node[roundnode] at (2.5,0.5)     (y2)                              { $y_2$};

\node[roundnode] at (2.5,1)     (y3)                              { $y_3$};

\node[roundnode] at (2.5,1.5)     (y4)                              { $y_4$};
\node[roundnode] at (2.5,2)     (y5)                              { $y_5$};

\node[rectangle,draw,minimum width = 1cm, 
    minimum height = 2cm] (b) at (-1,0.3) {$B$};
    
    \node[rectangle,draw,minimum width = 1cm, 
    minimum height = 2cm] (a) at (4,0.3) {$A$};


;

\draw (x1) to (x2);
\draw (x2) to (x3);
\draw (x3) to (x4);
\draw (x4) to (x5);
\draw (y1) to (y2);
\draw (y2) to (y3);
\draw (y3) to (y4);
\draw (y4) to (y5);

\draw (x2) to (y2);
\draw (x2) to (y4);

\draw (x4) to (y2);
\draw (x4) to (y4);

\draw [ultra thick](x4) to (b);
\draw [ultra thick](x2) to (b);

\draw[ultra thick] (y4) to (a);
\draw[ultra thick] (y2) to (a);

\draw[ultra thick,dashed,bend left=25] (a) to (b);

\end{tikzpicture}}   
    \caption{A sketch of $G' = (V',E')$; thick lines mean that all the edges across two sets are present. The thick dashed line represents all edges that have been present between $A$ and $B$ already in~$G$.}
    \label{bipartitefig}
\end{figure}

Hence, $G'$ is bipartite. Now we prove that $G'$ is 5-degenerate. 
As the maximum degree of~$G$ is 3 and as the construction of $G'$ will add two more neighbors to each vertex of $A\cup B$, vertices from $A\cup B$ have degree at most~5 in~$G'$. Moreover, vertices $x_1$, $x_3$, $x_5$, $y_1$, $y_3$ and $y_5$ have degree~1 or 2 in $G'$. Thus, any induced subgraph containing a vertex other than $x_2$, $x_4$, $y_2$, $y_4$ has some vertex of degree at most~5. Moreover, the subgraph induced by these four vertices is a 4-cycle. Thus, $G'$ is 5-degenerate. 
\end{proof}

\begin{claim} \label{obs:diameter}
 Any shortest path of $G'$, not starting and ending with one of the vertices from $\{y_1,y_5,x_1,x_5\}$ has at most $4$ vertices. 
Also, $G'$ has a diameter of $4$.
\end{claim} 

\begin{proof} 
In $G'$, any vertex $v \in A' \setminus \{y_1,y_5\}$, is at a distance at most~$2$ away from any vertex in~$A'$ and at a distance at most~$3$ from a vertex in~$B'$. Similarly, any vertex $u \in B' \setminus \{x_1,x_5\}$ is at a distance at most~$2$ away from a vertex in~$B'$ and at a distance at most~$3$ from a vertex in~$A'$. From the above arguments, we can see that any shortest path with $5$ vertices has to start and end with one of the following vertices: $\{y_1,y_5,x_1,x_5\}$. We can see the distance between any two vertices in the set is~$4$.\renewcommand{\qedsymbol}{$\Diamond$}
\end{proof}

\noindent
The arguments leading to the previous claim also give rise to the following one.
\begin{claim} \label{obs:two_paths_in_partition}
There can be only two shortest paths in $G'$ that have five vertices and that can simultaneously exist in a path partition.
\end{claim}

\begin{proof}
    The above claim is true, as the starting and ending vertices of shortest path having $5$ vertices must start and end at $\{y_1,y_5,x_1,x_5\}$.  \renewcommand{\qedsymbol}{$\Diamond$}
\end{proof} 
 
Notice that \autoref{obs:bipartite_degenerate} and \autoref{obs:diameter} together guarantee the additional properties of the constructed graph $G'$ that have been claimed in \autoref{thm:SPP-NP-hardness}. 

\begin{claim}  \label{obs:dist}
Let $u,v\in V$ have distance $d \leq 3$ in $G$. Then, they also have distance~$d$ in~$G'$.
Hence, if $P$ is a shortest path on at most four vertices in~$G$, then $P$ is also a shortest path in~$G'$.
\end{claim} 

\begin{proof}
This observation is trivial if $d=0,1$; it follows for $d=2$ as (also) the graph~$G'$ is bipartite by \autoref{obs:bipartite_degenerate}. For $d=3$, the distance remains the same, as no edge is added between $u \in A$ and $v \in B$ in~$G'$. 
\renewcommand{\qedsymbol}{$\Diamond$}
\end{proof}

Now, we claim that $G$ is a \yes-instance of \textsc{$4$-USPP} if and only if $G'$ has a shortest path partitioning of cardinality $k'=k + 2$, where $|V|=4k$. For the forward direction, let $\mathbb{P}$ be any solution of \textsc{$4$-USPP} for~$G$ containing~$k$ shortest paths. To construct a solution~$\mathbb{P}'$ of \textsc{USPP} for the instance~$(G',k')$, we just need to add the two paths $x_1x_2x_3x_4x_5$ and $y_1y_2y_3y_4y_5$ to~$\mathbb{P}$. 
By \autoref{obs:dist}, every path $P \in \mathbb{P}$ is in fact a shortest path in~$G'$. Hence, $\mathbb{P}'$ is a set of shortest paths with cardinality $k'=k + 2$ that covers all vertices of~$G'$.

For the backward direction, assume~$G'$ has a solution $\mathbb{P}'$, where $|\mathbb{P}'|=k'=k+2$. As $|V'| = 4k + 10$ by construction, we know by \autoref{obs:diameter} and \autoref{obs:two_paths_in_partition} that $\mathbb{P}'$ contains $k$ paths of length three and two paths of length~4. The only two possible paths of length~$4$ following the condition are $\{x_1x_2x_3x_4x_5,y_1y_2y_3y_4y_5\}$, otherwise we would create independent vertices $y_3$ or $x_3$. Hence, $$\{x_1x_2x_3x_4x_5,y_1y_2y_3y_4y_5\} \subseteq \mathbb{P}'$$ and the rest of the paths of $\mathbb{P}'$ are of length~3 and consists of vertices from~$V$ only. Let $\mathbb{P} = \mathbb{P}' \setminus \{x_1x_2x_3x_4x_5,y_1y_2y_3y_4y_5\}$.
As the $k$ paths of~$\mathbb{P}$ are each of length~3 also in~$G$ (by \autoref{rem:combi}) and as they cover $V$ completely, $\mathbb{P}$ provides a solution to the \textsc{$4$-USPP} instance~$G$. 
\end{proof}

\section{\W{1}-hardness for DAGSPP and DAGIPP}\label{w1hard}

The natural or standard parameterization of a parameterized problem stemming from an optimization problem is its solution size. We will study this type of parameterization in this section for path partitioning problems. More technically speaking, we are parameterizing these problems by an upper bound on the number of paths in the partitioning. Unfortunately, our results show that for none of the variations that we consider, we can expect \FPT-results.

\begin{theorem} \label{thm:w-hard-dag}
    \textsc{DAGSPP} and \textsc{DAGIPP} are \W{1}-hard when parameterized by solution size.
\end{theorem}

The following reduction is non-trivially adapted from~\cite{Slivkins2010}. The reduction starts out from \textsc{Clique} and produces an array-like layout of the input graph as an instance of \textsc{DAGSPP} (or of \textsc{DAGIPP}), as can be seen in the figures on the following pages. The construction is quite technical and is now described in details.

\begin{proof}
We define a parameterized reduction from \textsc{Clique} to \textsc{DAGSPP} (or to \textsc{DAGIPP}), all parameterized by solution size. We focus on \textsc{DAGSPP}, but all the arguments work for \textsc{DAGIPP} as well. Let $(G,k)$ be an instance of \textsc{Clique}, where $k\in\mathbb{N}$ and $G=(V,E)$, with $V=[n]$ for simplicity. We construct an equivalent instance $(G',k')$ of the \textsc{DAGSPP} problem, where $G'=(V',E')$ is a DAG and $k' = \frac{k\cdot (k-1)}{2} + 3k$. We will  now formally describe $V'$ and $E'$ in the following and later explain their functionalities. To avoid trivialities as well as many corner cases, we tacitly assume $k\geq 2$ in the following.

We define the vertex set first. For $i \in [k]$ and $u \in [n]$, let  $$V^{i,u} = \{a_l^{i,u},b_l^{i,u}\mid l \in [k] \setminus i\} $$ be the vertex set of our \emph{main gadgets} $G^{i,u}$. For  $i \in [k]$, the \emph{dummy gadgets} $G^{i,0}$ and  $G^{i,n+1}$ have vertex sets $$V^{i,0} = \{a_l^{i,u},b_l^{i,u}\mid l \in \{1,2\}\} \text{ and } V^{i,n+1} = \{a_l^{i,u},b_1^{i,n+1},b_2^{i,n+1}\mid l \in [k+2]\}\,,$$respectively. Then, define $$V_{\text{gadgets}}\coloneqq \bigcup_{i \in [k]}\bigcup_{u\in \{0\}\cup [n+1]} V^{i,u}\,. $$  As we think of the gadgets being arranged in an array, we also speak of the $i^{\text{th}}$ \emph{row} (containing $V^{i,u}$ and more) in the following, as well as of the $u^{\text{th}}$ \emph{column}.
For each row $i \in [k]$, we also have \emph{row start terminals} $s_i$, $s'_i$ (collected in the set $V_{\text{rst}}$) and \emph{row target terminals}  $t_i$, $t'_i$ (collected in the set $V_{\text{rtt}}$). Also, for each $i \in [k]$, we have $2(k-i)$ \emph{column start terminals} $s_{i,j}$ and $s'_{i,j}$, where $k \geq j \geq i+1$, (collected in the set $V_{\text{cst}}$), as well as $2(i-1)$ \emph{column target terminals}
$t_{j,i}$ and $t'_{j,i}$, with $1\leq j \leq i-1$ (collected in the set $V_{\text{ctt}}$). This gives the set of terminals $V_{\text{terminals}}\coloneqq V_{\text{rst}}\cup V_{\text{rtt}} \cup  V_{\text{cst}}\cup V_{\text{ctt}}$. Altogether, we have $V' \coloneqq V_{\text{gadgets}}\cup V_{\text{terminals}}$.

Now, we describe the arc set. Let us first describe the arcs within the gadgets. For the main gadgets  $G^{i,u}$, with $i \in [k]$ and $u \in [n]$, we set $$E^{i,u} = E_{\text{down}}^{i,u} \cup E_{\text{right}}^{i,u} \cup E_{\text{bridge}}^{i,u}\,,$$ where $E_{\text{down}}^{i,u} = 
\{ (a_j^{i,u}, b_j^{i,u})\mid j \in [k] \setminus \{i\} \}\,,$  
$E_{\text{right}}^{i,u} = 
\{(a_j^{i,u},a_{j+1}^{i,u}), (b_j^{i,u},b_{j+1}^{i,u})\mid j \in [k] \setminus \{i-1,i\}\}\,,$ and 
$E_{\text{bridge}}^{i,u} = \{(a_{i-1}^{i,u},a_{i+1}^{i,u}), (b_{i-1}^{i,u},b_{i+1}^{i,u})\}\,.$ As visualized in \autoref{fig:I-SPP-W-gadget}, one can think of an array layout also for $G^{i,u}=(V^{i,u},E^{i,u})$; then $e\in E_{\text{down}}^{i,u}$ go downwards, $e\in E_{\text{right}}^{i,u}$ point to the right, in principle always to the next vertex in the index ordering, except for the $i^{\text{th}}$ column within $G^{i,u}$ that is bridged by $E_{\text{bridge}}^{i,u}$.

For the {dummy gadgets} $G^{i,0}$, we have $E^{i,0} = \{(a_1^{i,0},a_2^{i,0}), (b_1^{i,0},b_2^{i,0}) \}$ and for the {dummy gadgets} $G^{i,n+1}$, $E^{i,n+1} = 
\{ (a_l^{i,n+1},a_{l+1}^{i,n+1})\mid l \in [k + 1]\} \cup \{(b_{1}^{i,n+1}, b_2^{i,n+1})\}\,.$
This describes all arcs within the gadgets $G^{i,u}=(V^{i,u},E^{i,u})$, i.e., $G^{i,u}=G'[V^{i,u}]$ for all $i\in [k]$ and $u\in \{0\}\cup [n+1]$.

Next, we describe (most of) the connections between terminals. For the row terminals, let $E_{\text{rst}} = 
\{ (s_i,s'_i)\mid i \in [k] \}$ and $E_{\text{rtt}} = 
\{(t'_i,t_i)\mid i \in [k]\}$. Then, $(V_{\text{rst}},E_{\text{rst}})=G'[V_{\text{rst}}]$
and $(V_{\text{rtt}},E_{\text{rtt}})=G'[V_{\text{rtt}}]$.
For the column terminals, let 
$E_{\text{cst}} = 
\{ (s_{i,j},s'_{i,j})\mid 1\leq i < j \leq k\}$ and 
$E_{\text{ctt}} = 
\{ (t'_{j,i},t_{j,i})\mid 1\leq j < i \leq k\}\,.$
Then, $(V_{\text{cst}},E_{\text{cst}})=G'[V_{\text{cst}}]$
and $(V_{\text{ctt}},E_{\text{ctt}})=G'[V_{\text{ctt}}]$.

The remaining description of the set of arcs of~$G'$ is dedicated to arcs between certain gadgetries.
Here, first we describe all arcs that connect main gadgets with other main gadgets or with dummy gadgets.
We have two arcs which connect two consecutive gadgets in a row $i \in [k] \setminus \{1,k\}$: 
$$ E_{\text{consec}}^{i} = \{ (a_{2}^{i,0},a_{1}^{i,1}),(b_{2}^{i,0},b_{1}^{i,1})\}\cup 
\{(a_{k}^{i,u},a_{1}^{i,u+1}),(b_{k}^{i,u},b_{1}^{i,u+1})\mid u \in  [n]\} \,.$$ 
The two cases $i=1$ and $i=k$ are special because then, the first (or last, respectively) column are missing in the main gadget.
\begin{eqnarray*} E_{\text{consec}}^{1} &=& \{ (a_{2}^{1,0},a_{2}^{1,1}),(b_{2}^{1,0},b_{2}^{1,1})\}\cup 
\{(a_{k}^{1,u},a_{2}^{1,u+1}),(b_{k}^{1,u},b_{2}^{1,u+1})\mid u \in [n-1]\}\\
&\cup &\{(a_{k}^{1,n},a_{1}^{1,n+1}),(b_{k}^{1,n},b_{1}^{1,n+1})\}\,\text{ and}\\
 E_{\text{consec}}^{k} &=& \{ (a_{2}^{k,0},a_{1}^{k,1}),(b_{2}^{k,0},b_{1}^{k,1})\}\cup
\{(a_{k-1}^{k,u},a_{1}^{k,u+1}),(b_{k-1}^{1,u},b_{1}^{k,u+1})\mid u \in  [n]\}\,.\end{eqnarray*}

Moreover, we have some \emph{skipping arcs} between the gadgets, namely, for each row  $i \in [k] \setminus \{1,k\}$, we set
    $$ E_{\text{skip}}^{i} = \{(a_{2}^{i,0},b_{1}^{i,2})\} \cup
    \{(a_{k}^{i,u},b_{1}^{i,u+2})\mid u \in [n-1]\}\,,$$
again with two special cases for $i=1$ and $i=k$, resulting in     $$E_{\text{skip}}^{1}= \{(a_{2}^{1,0},b_{2}^{1,2})\}\cup 
    \{(a_{k}^{1,u},b_{2}^{1,u+2})\mid u \in [n-2]\}\cup\{(a_{k}^{1,n-1},b_{1}^{1,n+1})\}\,\text{ and}  $$
$$E_{\text{skip}}^{k}= \{(a_{2}^{k,0},b_{1}^{k,2})\}\cup 
\{(a_{k-1}^{k,u},b_{1}^{k,u+2})\mid u \in [n-1]\}  \,.$$
The arcs from $\bigcup_{i\in [k]}(E_{\text{consec}}^{i}\cup E_{\text{skip}}^{i}) $
are also illustrated in \autoref{fig:I-SPP-verifier-a}.

To model the edges of the original graph~$G$, we need the following arcs 
$$E_{\text{con}} = 
\{(b_{j}^{i,u}, a_{i}^{j,v} )\mid uv \in E, 1\leq i < j \leq k \}\,.$$ These arcs are also illustrated in \autoref{fig:I-SPP-verifier}. Apart from some shortcut arcs defined below, we have all arcs of $G'[V_{\text{gadgets}}]$ collected in
$$ E_{\text{gadgets}} = \left(\bigcup_{i \in [k]} \left(\bigcup_{u \in \{0\} \cup [n+1]}  E^{i,u}\right) \cup E_{\text{skip}}^{i} \cup E_{\text{consec}}^{i}\right) \cup E_{\text{con}}\,.$$  Next, we have some arcs that connect the row start and row target terminal to the dummy gadget:
    $$E_{\text{stcon}}^{i} = 
    \{ (s'_i, a_1^{i,0}), (b_2^{i,n+1},t'_i)\mid i \in [k]\}\,.$$ Similarly, we have some arcs that connect column terminals to the main gadgets: for $i \in [k]$, define $$E_{\text{upcon}}^{i} =
    \{(s'_{i,j},a_{j}^{i,u})\mid 1\leq i < j \leq k, u \in [n]\}\,,$$ 
    $$ E_{\text{downcon}}^{i} = 
    \{(b_{j}^{i,u},t'_{j,i})\mid 1\leq j < i \leq k, u \in [n]\}\,.$$ We then set for all these terminal connectors 
    $$ E_{\text{tercon}} = \bigcup_{i \in k}\left( E_{\text{stcon}}^{i} \cup E_{\text{upcon}}^{i} \cup E_{\text{downcon}}^{i} \right)\,.$$ We have some arcs that help enforce the start and target vertices of the paths. Their subscripts indicate if they \underline{f}orce some \underline{r}ow/column \underline{s}tart to some \underline{r}ow/column \underline{t}arget (as in the first case) or other combinations. In the following, let $i\in [k]$.
    $$E_{\text{frsrt}}^{i} = 
    \{(s_i,t_l), (s_i,t'_l), (s'_i,t_l), (s'_i,t'_l)\mid  i < l \leq k\}\,,$$ 
    with $E_{\text{frsrt}}^{k}=\emptyset$ as a corner case,
    $$ E_{\text{frsct}}^{i} = 
    \{(s_i,t_{j,l}), (s_i,t'_{j,l}) , (s'_i,t_{j,l}), (s'_i,t_{j,l})\mid i \leq l \leq k, j\in [l-1]\}\,,$$     with $E_{\text{frsct}}^{1}=\emptyset$ as a corner case, 
$$E_{\text{fcsrt}}^{i} =  
\{(s_{i,j},t_{l}),(s'_{i,j},t_{l}), (s_{i,j},t'_{l}),(s'_{i,j},t'_{l})\mid i < j \leq k, i \leq l \leq k\}\,, $$ 
    with $E_{\text{fcsrt}}^{k}=\emptyset$ as a corner case, 
    and

 \begin{eqnarray*} E_{\text{fcsct}} ^{i}&=&\{ (s_{i,j},t_{l,h}), (s_{i,j},t'_{l,h}), (s'_{i,j},t_{l,h}), (s'_{i,j},t'_{l,h})\mid \\&&\phantom{XXXXXXXX}i < j \leq k,i\leq h\leq k, l \in [h-1],(i,j)\neq(h,l)\}\,.
 \end{eqnarray*}
    We also connect the dummy gadget to some target vertices as follows: 
    $$E_{\text{fbrt}}^{i} = 
    \{ (b_1^{i,0},t_l), (b_1^{i,0},t'_l), (b_2^{i,0},t_l), (b_2^{i,0},t'_l)\mid i \leq l \leq k\}\,,$$  
    $$ E_{\text{fbct}}^{i} = 
    \{(b_1^{i,0},t_{j,l}), (b_1^{i,0},t'_{j,l}) , (b_2^{i,0},t_{j,l}), (b_2^{i,0},t_{j,l})\mid i \leq l \leq k,j\in [l-1]\}\,.$$  

    So we can collect all these enforcing arcs into one arc set: 
    $$ E_{\text{force}} = \bigcup_{i \in [k]} \left(E_{\text{frsrt}}^{i} \cup E_{\text{frsct}}^{i} \cup E_{\text{fbrt}}^{i} \cup E_{\text{fbct}}^{i} \cup  E_{\text{fcsrt}}^{i} \cup E_{\text{fcsct}}^{i}\right)\,.$$ 
    
To prove \autoref{obs:top_row}, we need the following arcs (for row $i \in [k]$) that connect start terminals to the rightmost dummy gadget: 
    $$ E_{\text{spfrst}}^{i} =  
    \{ (s_i,a_{q}^{l,n+1}), (s'_i,a_{q}^{l,n+1}) \mid i \leq l \leq k,q \in [k+2]\}\,,$$ 
    $$ E_{\text{spfcst}}^{i} = 
    \{(s_{i,j},a_{q}^{l,n+1}), (s'_{i,j},a_{q}^{l,n+1})\mid i \leq l \leq k, i<j\leq k, q \in [k+2]\}\,.$$
    Finally, we also connect the bottom main gadget vertices (and those of the leftmost dummy gadget) to the top vertices of the rightmost dummy gadget, this way defining the set $ E_{\text{spfbot}}^{i}$ as 
    $$ 
    \bigcup_{q \in [k + 2]}\bigcup_{i < j \leq k} 
    \left(\left\{(b_l^{i,u}, a_q^{j,n+1})\mid u \in [n],l \in [k] \setminus \{i\}\right\} \cup \{(b_1^{i,0},a_q^{j,n+1}),(b_2^{i,0},a_q^{j,n+1})\}\right)\,. $$ We collect these arc sets into the set of \emph{\underline{s}hortest \underline{p}ath en\underline{f}orcers} that provide the shortcuts mentioned above.
    $$ E_{\text{spf}} = \bigcup_{i \in [k]} \left(E_{\text{spfcst}}^{i} \cup E_{\text{spfrst}}^{i} \cup E_{\text{spfbot}}^{i}\right)\,.$$  
  We can now define $$ E' = E_{\text{rst}} \cup E_{\text{rtt}} \cup E_{\text{cst}} \cup E_{\text{ctt}} \cup E_{\text{gadgets}} \cup E_{\text{tercon}} \cup E_{\text{force}} \cup E_{\text{spf}}\,.$$

\noindent
This concludes the formal description of $G'=(V',E')$.

    \smallskip\noindent
    \textbf{Overview of the construction:} We create an array of $k\times n$ identical gadgets for the construction, with each gadget representing a vertex in the original graph. We can visualize this array as having $k$ {rows} and $n$ columns.  If the \textsc{DAGSPP} instance $(G',k')$ is a \yes-instance, then we show that each row has a so-called \emph{selector} in the solution. Here, a selector is a path that traverses all but one gadget in a row, hence skipping exactly one of the gadgets. The  vertices in~$G$ corresponding to the skipped gadgets form a clique of size~$k$ in~$G$. To ensure that all selected vertices form a clique in~$G$, we have so-called \emph{verifiers} in the \textsc{DAGSPP} instance. {Verifiers} are the paths that are used for each pair of rows to ensure that the vertices corresponding to the selected gadgets in these rows are adjacent in~$G$. This way, we also do not have to check separately that the selected vertices are distinct, see \autoref{fig:I-SPP-verifier}. 
    \smallskip

    \noindent
    \textbf{Construction details:} The array of gadgets is drawn with row numbers increasing downward and column numbers increasing to the right. Arcs between columns go down and arcs within the same row go to the right. To each row~$i$, with $i \in [k]$, we add row start terminals, connected by an arc $(s_i,s'_i)\in E_{\text{rst}}$, and  row target terminals, connected by an arc $(t'_i,t_i)\in E_{\text{rtt}}$. Next, add arcs starting from $s_i,s'_i$ to $t_l$ and $t'_l$, with $l>i$, see $E_{\text{frsrt}}^{i}$. Also, for each row, we have $k - i$ \emph{column start terminals}, arcs $(s_{i,j},s'_{i,j})\in E_{\text{cst}}$ with $i < j$, and $i - 1$ \emph{column target} terminals,  arcs $(t'_{j,i},t_{j,i})\in E_{\text{rtt}}$ with $j \in [i-1]$. Also, add arcs from vertices $s_{i,j}$ and $s'_{i,j}$ to $t_l$, $t'_l$ if $l\geq i$ and to column target terminals (see $E^i_{\text{fcsct}}$ in the formal construction). 

 \smallskip
 \noindent 
    \textbf{Gadgets} are denoted by $G^{i,u}$, $i \in [k]$, corresponding to $u \in V$. Each gadget~$G^{i,u}$ consists of $k-1$ top-level vertices $a_j^{i,u}$ and of $k-1$ bottom-level  vertices $b_j^{i,u}$, connected by arcs $(a_j^{i,u}, b_j^{i,u})\in E_{\text{down}}^{i,u}$, with $j \in [k] \setminus \{i\}$. Due to the natural ordering of the vertices within a gadget, we can also speak of the first vertex on the top level of a gadget or of the last vertex on the bottom level of a gadget. On both levels, the vertices are also connected following their natural ordering. More technically speaking, this means that within gadget~$G^{i,u}$, with $u\in V$, there is an arc from the first vertex of the upper level to the second vertex of the upper level, from the second vertex of the upper level to the third vertex of the upper level, etc., up to an arc from the penultimate vertex of the upper level to the last vertex of the upper level. The formal description of the edge set $E^{i,u}$ is a bit more technical, because one index of the set $[k]$ is missing in the column indexing. The arcs within a gadget described so far lead to a grid-like drawing of each gadget, see \autoref{fig:I-SPP-W-gadget}.
    
\smallskip
\noindent    
    To each row~$i\in [k]$, we also add \emph{dummy gadgets}
    $G^{i,0}$ and $G^{i,n+1}$. Gadget $G^{i,0}$ consists of two arcs, $(a_1^{i,0},a_2^{i,0})$ and $(b_1^{i,0},b_2^{i,0})$, also add arcs from $b_1^{i,0}$ and $b_2^{i,0}$ to $t_h, t'_h$ and to $t'_{l,h}$, $t_{l,h}$ where $l< h$ and $i \leq h$, see the formal definition of $E_{\text{fbrt}}^{i}$ and of $E_{\text{fbct}}^{i}$. Gadget $G^{i,n+1}$ consists of a directed $P_{k+2}$ and an arc $(b_1^{i,n+1},b_2^{i,n+1})$; the $P_{k+2}$-vertices are named $a_j^{i,n+1}$, with $j \in [k+2]$, where $a_{k+2}^{i,n+1}$ has out-degree zero. The arcs inside the dummy gadgets are formally defined as the sets $E^{i,0}$ and $E^{i,n+1}$ above.
    We can speak of $$T^{i}:=\{a_r^{i,u}\mid r\in [k]\setminus\{i\},u\in [n]\}\cup\{a_1^{i,0},a_2^{i,0},a_j^{i,n+1}\mid j \in [k+2]\}$$ as being on the  $i^{\text{th}}$ \emph{top level}, while the vertices $$B^{i}:=\{b_r^{i,u}\mid r\in [k]\setminus\{i\},u\in [n]\}\cup\{b_1^{i,0},b_2^{i,0},b_1^{i,n+1},b_2^{i,n+1}\}$$ form the $i^{\text{th}}$ \emph{bottom level}, where $i\in [k]$.

Moreover, there is an arc from the last vertex of the upper level of $G^{i,u-1}$ to the first vertex of the upper level of $G^{i,u}$ and from the last  vertex of the upper level of $G^{i,u}$ to the first vertex of the upper level of $G^{i,u+1}$, etc.
Analogously, the vertices of the lower level of the gadgets are connected. For a formal treatment, we refer to the definition of $E_{\text{consec}}^{i}$. These notions are illustrated in \autoref{fig:I-SPP-W-gadget} and \autoref{fig:I-SPP-verifier-a}. In the figures, for readability we omit the superscripts and simply write $(a_l, b_l) = (a_l^{i,u} , b_l^{i,u} )$ where $l \in  [k]$. Also, notice that the mentioned vertex names $a_1$, $b_1$, $a_k$ and $b_k$ are rather meant in a symbolic fashion. Clearly, by our construction, the first vertex of the top level of $G_{1,u}$ would be $a_2^{1,u}$, to explicitly mention one exceptional case.

\begin{figure}
\centering
\begin{minipage}{.47\textwidth}
  \centering
  \tikzset{every picture/.style={line width=0.75pt}} 

\begin{tikzpicture}[
roundnode/.style={circle, draw=black,  thick, minimum size=5mm, scale=0.8},
squarednode/.style={rectangle, draw=red!60, fill=red!5, very thick, minimum size=5mm},
dot/.style={fill=black,circle,minimum size=1pt}
]

\def\x{0.6}
\def\y{-0.1}

\node[rectangle,draw,minimum width = 6cm, 
    minimum height = 2cm, rounded corners=5pt] (b) at (\x,\y) {};
    \node[roundnode] at (\x - 1.5cm ,\y + 0.5cm )     (a1){}  ;
    \node[roundnode] at (\x - 0.5 cm ,\y + 0.5cm )     (a2){}  ;

    \node[dot,scale=0.3] at (\x - 0.1cm ,\y + 0.5cm) (d1){};
    \node[dot,scale=0.3] at (\x - 0.0cm ,\y + 0.5cm) (d1){};
    \node[dot,scale=0.3] at (\x + 0.1cm ,\y + 0.5cm) (d1){};
    
 \node[roundnode] at (\x + 0.5cm ,\y + 0.5cm )     (ai1){}  ;
    \node[roundnode] at (\x + 1.5cm ,\y + 0.5cm )     (ai){}  ;
    
     \node[dot,scale=0.3] at (\x + 2cm ,\y + 0.5cm) (d1){};
    \node[dot,scale=0.3] at (\x + 2.1cm ,\y + 0.5cm) (d1){};
    \node[dot,scale=0.3] at (\x + 2.2cm ,\y + 0.5cm) (d1){};

    \node[roundnode] at (\x + 2.7cm ,\y + 0.5cm )     (ak){}  ;

    \node[roundnode] at (\x + 1.5cm ,\y + 1.5cm )     (si){}  ;
    \node[roundnode] at (\x + 1.5cm ,\y + 2.3cm )     (sij){}  ;
    \node[roundnode] at (\x + 2.7cm ,\y + 1.5cm )     (s'k){}  ;
    \node[roundnode] at (\x + 2.7cm ,\y + 2.3cm )     (sk){}  ;

    \node[roundnode] at (\x - 1.5cm ,\y - 0.5cm )     (b1){}  ;
    \node[roundnode] at (\x - 1.5cm ,\y - 1.5cm )     (t'1){}  ;
    \node[roundnode] at (\x - 1.5cm ,\y - 2.3cm )     (t1){}  ;

    \node[roundnode] at (\x - 0.5cm ,\y - 0.5cm  )     (b2){}  ;
    \node[roundnode] at (\x - 0.5cm ,\y - 1.5cm )     (t'2){}  ;
    \node[roundnode] at (\x - 0.5cm ,\y - 2.3cm )     (t2){}  ;

    \node[dot,scale=0.3] at (\x - 0.1cm ,\y - 0.5cm) (d1){};
    \node[dot,scale=0.3] at (\x + 0cm ,\y - 0.5cm) (d1){};
    \node[dot,scale=0.3] at (\x + 0.1cm ,\y - 0.5cm) (d1){};
    
      \node[roundnode] at (\x + 0.5cm ,\y - 0.5cm )     (bi1){}  ;
      \node[roundnode] at (\x + 0.5cm ,\y - 1.5cm )     (t'i1){}  ;
      \node[roundnode] at (\x + 0.5cm ,\y - 2.3cm )     (ti1){}  ;
      
    \node[roundnode] at (\x + 1.5cm ,\y - 0.5cm )     (bi){}  ;
    
     \node[dot,scale=0.3] at (\x + 2.cm ,\y-0.5cm) (d1){};
    \node[dot,scale=0.3] at (\x + 2.1cm ,\y -0.5cm) (d1){};
    \node[dot,scale=0.3] at (\x + 2.2cm ,\y -0.5cm) (d1){};

    \node[roundnode] at (\x + 2.7cm ,\y - 0.5cm )     (bk){} ; 
    
    \node[] at (\x - 1.1cm , \y + 0.2cm) {$a_1$};
     \node[] at (\x - .15cm , \y + 0.2cm) {$a_2$};
     \node[] at (\x + 1cm , \y + 0.2cm) {$a_{i-1}$};
     \node[] at (\x+0.9cm  , \y + 2.3cm) {$s_{i,i+1}$};
     \node[] at (\x+0.9cm  , \y + 1.5cm) {$s'_{i,i+1}$};
     \node[] at (\x + 2.05cm , \y + 0.2cm) {$a_{i+1}$};
         \node[] at (\x + 3cm , \y + 0.2cm) {$a_{k}$};
           \node[] at (\x+3.2cm  , \y + 1.5cm) {$s'_{i,k}$};
           \node[] at (\x+3.2cm  , \y + 2.3cm) {$s_{i,k}$};
     
     \node[] at (\x - 1.1cm , \y - 0.8cm) {$b_1$};
     \node[] at (\x - 2cm , \y - 1.5cm) {$t'_{1,i}$};
     \node[] at (\x - 2cm , \y - 2.3cm) {$t_{1,i}$};
     \node[] at (\x - .15cm , \y - 0.8cm) {$b_2$};
     \node[] at (\x -1cm , \y - 1.5cm) {$t'_{2,i}$};
     \node[] at (\x -1cm , \y - 2.3cm) {$t_{2,i}$};
     
     \node[] at (\x + 1cm , \y - 0.8cm) {$b_{i-1}$};
     \node[] at (\x+1.2cm  , \y - 1.5cm) {$t'_{i-1,i}$};
     \node[] at (\x+1.2cm  , \y - 2.3cm) {$t_{i-1,i}$};
     
      \node[] at (\x + 2.05cm , \y - 0.8cm) {$b_{i+1}$};
     \node[] at (\x + 3cm , \y - 0.8cm) {$b_{k}$};

      \draw[-{Stealth[scale=1]}] (a1) to (a2);
    \draw[-{Stealth[scale=1]}] (a1) to (b1);
    \draw[-{Stealth[scale=1]}] (a2) to (b2);
    \draw[-{Stealth[scale=1]}] (ai1) to (ai);
    \draw[-{Stealth[scale=1]}] (ai) to (bi);
    \draw[-{Stealth[scale=1]}] (ai1) to (bi1);
    \draw[-{Stealth[scale=1]}] (bi1) to (bi);

    \draw[-{Stealth[scale=1]}] (si) to (ai);
    \draw[-{Stealth[scale=1]}] (sij) to (si);
    \draw[-{Stealth[scale=1]}] (s'k) to (ak);
    \draw[-{Stealth[scale=1]}] (sk) to (s'k);

     \draw[-{Stealth[scale=1]}] (b1) to (t'1);
    \draw[-{Stealth[scale=1]}] (b2) to (t'2);
    \draw[-{Stealth[scale=1]}] (bi1) to (t'i1);
    \draw[-{Stealth[scale=1]}] (t'1) to (t1);
    \draw[-{Stealth[scale=1]}] (t'2) to (t2);
    \draw[-{Stealth[scale=1]}] (t'i1) to (ti1);

    \draw[-{Stealth[scale=1]}] (b1) to (b2);
    \draw[-{Stealth[scale=1]}] (ak) to (bk);
    \draw[-{Stealth[scale=1]}] (\x - 2.5cm, \y + 0.5cm) to (a1);
    \draw[-{Stealth[scale=1]}] (\x - 2.5cm, \y - 0.5cm) to (b1);

    \draw[-{Stealth[scale=1]}] (ak) to (\x + 3.8cm, \y + 0.5cm);
    \draw[-{Stealth[scale=1]}] (bk) to (\x + 3.8cm, \y - 0.5cm);


 \end{tikzpicture}
  \captionof{figure}{Gadget $G^{i,u}$, $0< u \leq n$, $i \in [k]$}
  \label{fig:I-SPP-W-gadget}
\end{minipage}%
\qquad 
\begin{minipage}{.46\textwidth}
  \centering
  \tikzset{every picture/.style={line width=0.75pt}} 

\begin{tikzpicture}[
roundnode/.style={circle, draw=black,  thick, minimum size=5mm, scale=0.7},
squarednode/.style={rectangle, draw=red!60, fill=red!5, very thick, minimum size=5mm},
dot/.style={fill=black,circle,minimum size=1pt}
]

\def\x{0.6}
\def\y{0}

\def\l{2.6}
\def\m{-2.5}

\node[rectangle,draw,minimum width = 3cm, 
    minimum height = 1.5cm, rounded corners=5pt] (b) at (\x,\y) {};

     \node[roundnode] at (\x - 0.5cm ,\y + 0.5cm )     (a1){}  ;
      \node[dot,scale=0.3] at (\x - 0.1cm ,\y + 0.5cm) (d1){};
    \node[dot,scale=0.3] at (\x - 0.0cm ,\y + 0.5cm) (d1){};
    \node[dot,scale=0.3] at (\x + 0.1cm ,\y + 0.5cm) (d1){};
    \node[roundnode] at (\x  cm ,\y + 0.5cm )     (aj){}  ;

    \node[dot,scale=0.3] at (\x + 0.9cm ,\y + 0.5cm) (d1){};
    \node[dot,scale=0.3] at (\x + 1.cm ,\y + 0.5cm) (d1){};
    \node[dot,scale=0.3] at (\x + 1.1cm ,\y + 0.5cm) (d1){};
    \node[roundnode] at (\x + 1.5cm ,\y + 0.5cm )     (ak){}  ;
    \node[roundnode] at (\x - 0.5 cm ,\y - 0.5cm )     (b1){}  ;
    \node[dot,scale=0.3] at (\x - 0.1cm ,\y - 0.5cm) (d1){};
    \node[dot,scale=0.3] at (\x - 0.0cm ,\y - 0.5cm) (d1){};
    \node[dot,scale=0.3] at (\x + 0.1cm ,\y - 0.5cm) (d1){};

    \node[roundnode] at (\x cm ,\y - 0.5cm )     (bj){}  ;
     \node[dot,scale=0.3] at (\x + 0.9cm ,\y - 0.5cm) (d1){};
    \node[dot,scale=0.3] at (\x + 1.cm ,\y - 0.5cm) (d1){};
    \node[dot,scale=0.3] at (\x + 1.1cm ,\y - 0.5cm) (d1){};
    \node[roundnode] at (\x +1.5 cm ,\y - 0.5cm )     (bk){}  ;

    \draw[-{Stealth[scale=1]}] (a1) to (b1);
    \draw[-{Stealth[scale=1]}] (aj) to (bj);
    \draw[-{Stealth[scale=1]}] (ak) to (bk);

    \node[] at (\x - 0.2cm , \y + 0.2cm) {$ a_{1}$};
    \node[] at (\x + 0.9cm , \y + 0.2cm) {$ a_{j}$};
    \node[] at (\x + 1.74cm , \y + 0.2cm) {$ a_{k}$};

    \node[] at (\x - 0.5cm , \y - 1cm) {$ b_{1}$};
    \node[] at (\x + 0.6cm , \y - 1cm) {$ b_{j}$};
    \node[] at (\x + 1.9cm , \y - 0.5cm) {$ b_{k}$};

    \node[] at (\x + 2.5cm , \y ) {$ G^{i,u}$};

    \pgfmathsetmacro {\xa }{ \x + 2.2cm}
    \pgfmathsetmacro {\ya }{ \y - 2.1cm}

    \node[rectangle,draw,minimum width = 3cm, 
    minimum height = 1.5cm, rounded corners=5pt] (b) at (\xa + 0.6cm ,\ya + 0cm) {};

    \node[roundnode] at (\xa - 0.5cm ,\ya + 0.5cm )     (a12){}  ;
      \node[dot,scale=0.3] at ( \xa - 0.1cm , \ya + 0.5cm) (d1){};
    \node[dot,scale=0.3] at (\xa - 0.0cm , \ya + 0.5cm) (d1){};
    \node[dot,scale=0.3] at (\xa + 0.1cm , \ya + 0.5cm) (d1){};
    \node[roundnode] at (\xa + 0.5cm , \ya + 0.5cm )     (aj2){}  ;

    \node[dot,scale=0.3] at (\xa + 0.9cm , \ya + 0.5cm) (d1){};
    \node[dot,scale=0.3] at (\xa + 1.cm , \ya + 0.5cm) (d1){};
    \node[dot,scale=0.3] at (\xa + 1.1cm , \ya + 0.5cm) (d1){};
    \node[roundnode] at (\xa + 1.5cm , \ya + 0.5cm )     (ak2){}  ;
    \node[roundnode] at (\xa - 0.5 cm , \ya - 0.5cm )     (b12){}  ;
    \node[dot,scale=0.3] at (\xa - 0.1cm , \ya - 0.5cm) (d1){};
    \node[dot,scale=0.3] at (\xa - 0.0cm , \ya - 0.5cm) (d1){};
    \node[dot,scale=0.3] at (\xa + 0.1cm , \ya - 0.5cm) (d1){};

    \node[roundnode] at (\xa + 0.5 cm  , \ya - 0.5cm )     (bj2){}  ;
    
     \node[dot,scale=0.3] at (\xa + 0.9cm , \ya - 0.5cm) (d1){};
    \node[dot,scale=0.3] at (\xa + 1.cm , \ya - 0.5cm) (d1){};
    \node[dot,scale=0.3] at (\xa + 1.1cm , \ya - 0.5cm) (d1){};
    \node[roundnode] at (\xa +1.5 cm , \ya - 0.5cm )     (bk2){}  ;

    \draw[-{Stealth[scale=1]}] (a12) to (b12);
    \draw[-{Stealth[scale=1]}] (aj2) to (bj2);
    \draw[-{Stealth[scale=1]}] (ak2) to (bk2);

    \node[] at (\xa - 0.2cm ,  \ya + 0.2cm) {$ a_{1}$};
    \node[] at (\xa + 0.9cm ,  \ya + 0.2cm) {$ a_{i}$};
    \node[] at (\xa + 1.8cm ,  \ya + 0.2cm) {$ a_{k}$};

    \node[] at (\xa - 0.5cm ,  \ya - 1cm) {$ b_{1}$};
    \node[] at (\xa + 0.7cm ,  \ya - 1cm) {$ b_{i}$};
    \node[] at (\xa + 1.5cm ,  \ya - 1cm) {$ b_{k}$};

    \draw[-{Stealth[scale=1]}] (bj) to (aj2);

    \node[] at (\xa - 1.5cm , \ya - 0cm   ) {$ G^{j,v}$};


    \end{tikzpicture}
  \captionof{figure}{$G^{i,u}$ connected to $G^{j,v}$, $i<j$}
  \label{fig:I-SPP-verifier}
\end{minipage}
\end{figure}

 \begin{figure}[tbh]
    \centering
    \tikzset{every picture/.style={line width=0.75pt}} 

\begin{tikzpicture}[
roundnode/.style={circle, draw=black,  thick, minimum size=5mm, scale=0.5},
squarednode/.style={rectangle, draw=red!60, fill=red!5, very thick, minimum size=5mm},
dot/.style={fill=black,circle,minimum size=1pt}
]

\def\x{0}
\def\y{0}
 \pgfmathsetmacro {\ya }{ \y - 1cm }
 \def\sc{0.8}

    minimum height = 1.5cm, rounded corners=5pt] (b) at (\x + 1.25cm,\y - 0.5cm) {};
    \draw[rounded corners = 5pt] (\x - 0.5cm, \y + 0.2 cm) rectangle (\x + .5cm, \y - 0.5 cm) {};
    
    \draw[rounded corners = 5pt] (\x + 0.75cm, \y + 0.5 cm) rectangle (\x + 1.75cm, \ya - 0.5 cm) {};
    \draw[rounded corners = 5pt] (\x + 2.25cm, \y + 0.5 cm) rectangle (\x + 3.75cm, \ya - 0.5 cm) {};
    \draw[rounded corners = 5pt] (\x + 4.25cm, \y + 0.5 cm) rectangle (\x + 5.75cm, \ya - 0.5 cm) {};
    \draw[rounded corners = 5pt] (\x + 6.25cm, \y + 0.5 cm) rectangle (\x + 7.75cm, \ya - 0.5 cm) {};
    \draw[rounded corners = 5pt] (\x + 2.25cm, \y + 0.5 cm) rectangle (\x + 3.75cm, \ya - 0.5 cm) {};
    \draw[rounded corners = 5pt] (\x + 9.25cm, \y + 0.5 cm) rectangle (\x + 10.75cm, \ya - 0.5 cm) {};

\node[roundnode] at (\x - 0.3cm  ,\y  )     (si){}  ;
\node[roundnode] at (\x + 0.3cm  ,\y  )     (sid){}  ;

\node[roundnode] at (\x + 1cm  ,\y  )     (a10){}  ;
\node[roundnode] at (\x + 1.5cm  ,\y  )     (a20){}  ;

\node[roundnode] at (\x + 2.5cm  ,\y  )     (a11){}  ;
\node[roundnode] at (\x + 3.5cm  ,\y  )     (ak1){}  ;

\node[roundnode] at (\x +4.5cm  ,\y  )     (a12){}  ;
\node[roundnode] at (\x + 5.5cm  ,\y  )     (ak2){}  ;

\node[roundnode] at (\x + 6.5cm  ,\y  )     (a13){}  ;
\node[roundnode] at (\x + 7.5cm  ,\y  )     (ak3){}  ;

 \node[dot,scale=0.3] at (\x + 8.4cm ,\y - 0.5cm) (d1){};
    \node[dot,scale=0.3] at (\x + 8.5cm ,\y - 0.5cm) (d1){};
    \node[dot,scale=0.3] at (\x + 8.6cm ,\y - 0.5cm) (d1){};

\node[roundnode] at (\x + 9.5cm  ,\y  )     (a1n){}  ;
\node[roundnode] at (\x + 10.5cm  ,\y  )     (akn){}  ;


\node[roundnode] at (\x + 1cm  ,\ya + 0cm )     (b10){}  ;
\node[roundnode] at (\x + 1.5cm  ,\ya  + 0cm)     (b20){}  ;

\node[roundnode] at (\x + 2.5cm  ,\ya + 0cm  )     (b11){}  ;
\node[roundnode] at (\x + 3.5cm  ,\ya + 0cm )     (bk1){}  ;

\node[roundnode] at (\x +4.5cm  ,\ya + 0cm )     (b12){}  ;
\node[roundnode] at (\x + 5.5cm  ,\ya + 0cm )     (bk2){}  ;

\node[roundnode] at (\x + 6.5cm  ,\ya + 0cm )     (b13){}  ;
\node[roundnode] at (\x + 7.5cm  ,\ya + 0cm )     (bk3){}  ;

\node[roundnode] at (\x + 9.5cm  ,\ya + 0cm )     (b1n){}  ;
\node[roundnode] at (\x + 10.5cm  ,\ya + 0cm )     (bkn){}  ;

\node[roundnode] at (\x + 11cm  ,\ya + 0cm )     (tpi){}  ;
\node[roundnode] at (\x + 11.5cm  ,\ya + 0cm )     (ti){}  ;

\draw[-{Stealth[scale=\sc]}] (si) to (sid);
\draw[-{Stealth[scale=\sc]}] (sid) to (a10);
\draw[-{Stealth[scale=\sc]}] (a10) to (a20);
\draw[-{Stealth[scale=\sc]}] (b10) to (b20);
\draw[-{Stealth[scale=\sc]}] (a20) to (a11);
\draw[-{Stealth[scale=\sc]}] (b20) to (b11);
\draw[-{Stealth[scale=\sc]}] (ak1) to (a12);
\draw[-{Stealth[scale=\sc]}] (bk1) to (b12);
\draw[-{Stealth[scale=\sc]}] (ak2) to (a13);
\draw[-{Stealth[scale=\sc]}] (bk2) to (b13);
\draw[-{Stealth[scale=\sc]}] (a20) .. controls (3.25,-1) and (3.25,0) .. (b12);
\draw[-{Stealth[scale=\sc]}] (ak1) .. controls (5.25,-1) and (5.25,0) .. (b13);
\draw[-{Stealth[scale=\sc]}] (si) to (sid);
\draw[-{Stealth[scale=\sc]}] (b1n) to (bkn);
\draw[-{Stealth[scale=\sc]}] (bkn) to (tpi);
\draw[-{Stealth[scale=\sc]}] (tpi) to (ti);

\draw[-{Stealth[scale=\sc + 0.2]}, bend left = 24,dashed,thick] (\x , \y + 0.25cm) to  (\x + 10 cm , \y + 0.5cm);

  \node[] at (\x - 0.3cm , \y - 0.3 cm) {\small $s_i$};
  \node[] at (\x + 0.3cm , \y - 0.3 cm) {\small $s'_i$};

  \node[] at (\x + 1cm, \y + 0.3 cm) {\small $a_1$};
  \node[] at (\x + 1.5cm , \y + 0.3 cm) {\small $a_2$};

   \node[] at (\x + 2.5cm, \y + 0.3 cm) {\small $a_1$};
  \node[] at (\x + 3.5cm , \y + 0.3 cm) {\small $a_k$};

  \node[] at (\x + 4.5cm, \y + 0.3 cm) {\small $a_1$};
  \node[] at (\x + 5.5cm , \y + 0.3 cm) {\small $a_k$};

  \node[] at (\x + 6.5cm, \y + 0.3 cm) {\small $a_1$};
  \node[] at (\x + 7.5cm , \y + 0.3 cm) {\small $a_k$};

  \node[] at (\x + 9.5cm, \y + 0.3 cm) {\small $a_1$};
  \node[] at (\x + 10.4cm , \y + 0.3 cm) {\small $a_{k+2}$};

  \node[] at (\x + 1cm, \y - 1.3 cm) {\small $b_1$};
  \node[] at (\x + 1.5cm , \y - 1.3 cm) {\small $b_2$};

   \node[] at (\x + 2.5cm, \y - 1.3 cm) {\small $b_1$};
  \node[] at (\x + 3.5cm , \y - 1.3 cm) {\small $b_k$};

  \node[] at (\x + 4.5cm, \y - 1.3 cm) {\small $b_1$};
  \node[] at (\x + 5.5cm , \y - 1.3 cm) {\small $b_k$};

  \node[] at (\x + 6.5cm, \y - 1.3 cm) {\small $b_1$};
  \node[] at (\x + 7.5cm , \y - 1.3 cm) {\small $b_k$};

  \node[] at (\x + 9.5cm, \y - 1.3 cm) {\small $b_1$};
  \node[] at (\x + 10.4cm , \y - 1.3 cm) {\small $b_2$};

   \node[] at (\x + 11cm, \y - 0.7cm) {\small $t'_i$};
  \node[] at (\x + 11.5cm , \y - 0.7 cm) {\small $t_i$};

    \node[] at (\x + 1.2cm , \y - 1.8cm) {$G^{i,0}$};
    \node[] at (\x + 3cm , \y - 1.8cm) {$G^{i,1}$};
    \node[] at (\x + 5cm , \y - 1.8cm) {$G^{i,2}$};
    \node[] at (\x + 7cm , \y - 1.8cm) {$G^{i,3}$};
    \node[] at (\x + 10cm , \y - 1.8cm) {$G^{i,n+1}$};


 \end{tikzpicture}
    \caption{The $i^{th}$ row (only the first two skipping arcs are shown), the dashed line indicates arcs (shortest paths enforcers)} from $s_i, s'_i$ to $a_j^{i,n+1}$, with $j \in [k+2]$.
    \label{fig:I-SPP-verifier-a}
    
\end{figure}

\smallskip 
\noindent
A \textbf{selector} is a path that starts at $s_i$, enters its row at the top level, and exits it at the bottom level, ending at $t_i$ and skipping exactly one gadget~$G^{i,u}$, with $i \in [k]$ and $u \in V$. In order to implement this, we add the arcs $(s'_i, a_1^{i,0})$ and $(b_2^{i,n+1}, t'_i)$ for row~$i$, see the definition of $E_{\text{stcon}}^{i}$, as well as \emph{skipping arcs} that allow to skip a gadget $G^{i,u}$ for $u\in V$. Such an arc connects the last vertex of the top level of $G^{i,u-1}$ to  the first vertex of the bottom level of $G^{i,u+1}$, see \autoref{fig:I-SPP-verifier-a} for an illustration and the definition of $ E_{\text{skip}}^{i}$ for a formal treatment.

\noindent
A \textbf{verifier} is a path that routes through one of the vertices of the skipped gadget and connects column terminals $s_{i,j}$ to $t_{i,j}$. In order to implement this, we add the arcs $(b_{1}^{i,0}, t'_{i,j})$, $(b_{1}^{i,0}, t_{i,j})$,$(b_{2}^{i,0}, t_{i,j})$ and $(b_{2}^{i,0}, t'_{i,j})$ to every gadget in rows~$i$ and~$j$. See the formal definition of $E_{\text{fbct}}^i$. To connect the  gadget in row~$i$ and~$j$, for every edge $uv$ in $G$, we add an arc 
$(b_j^{i,u}, a_i^{j,v})$ for each $i<j$, see \autoref{fig:I-SPP-verifier} for an illustration and the definition of $E_{\text{con}}$ for formalities.

    To force the start and end vertices of paths in the spp of $G'$ of size~$k'$, we add the following arcs, called \emph{shortest paths enforcers}, see the definition of $E_{\text{spf}}$. For every row $i$, from all the bottom vertices of gadget $G^{j,u}$, where $j < i$ and $0 \leq u \leq n$, from every vertex $s_l$ and $s'_l$ of a row start terminal and from every vertex $s_{l,h}$, $s'_{l,h}$ of a column start terminal, where $l \leq i$ and  $l < h$, add arcs to $a_j^{i,n+1}$, with $j \in [k+2]$, see \autoref{fig:I-SPP-verifier-a}. The idea is, except for $T^{i}$, for all the vertices which can reach $a_j^{i,n+1}$, we add an arc connecting it to $a_j^{i,n+1}$.

\smallskip
\noindent

We are now going to show a number of properties of our construction.

\smallskip
\noindent
Observe that $G'$ is a DAG because all arcs either go from left to right or from top to bottom.
\noindent
Next, using \autoref{obs:in_out_deg} and \autoref{obs:top_row},  we can deduce that, if $G'$ has an spp of size $k'$, then the start vertex and end vertex of each path in the solution have the desired properties.

\begin{claim} \label{obs:in_out_deg}
        $G'$ has $k' - k$ vertices with in-degree zero and $k' - k$ vertices with out-degree zero.
    \end{claim}
    
\begin{proof}
    First, recall that  $k' = \frac{k\cdot (k-1)}{2} + 3k$. Hence, we have to show that we have $ \frac{k\cdot (k-1)}{2} + 2k$ many vertices with in-degree zero. Namely, for each row $i$, with $i \in [k]$, we have one such vertex as its {row start terminal}~$s_i$, $k-i$ {column start terminals} and one vertex of the dummy gadget $b_1^{i, 0}$. We also have $\frac{k\cdot (k-1)}{2} + 2k$ vertices with out-degree zero, for each row $i$, $i \in [k]$, one as the row's target terminal~$t_i$, $i-1$ as its column target terminals and as a vertex of the dummy gadget $a_{k+2}^{i,n+1}$, where  $|V|=n$. 
         \renewcommand{\qedsymbol}{$\Diamond$}
    \end{proof}
    
\smallskip
\noindent
Hence, we know that $k' - k$ start and $k' - k$ end vertices of the solution are fixed. The next claim shows that each row $i$ has one more start vertex of some path, hence fixing all the starting vertices of all the paths in the solution.

    \begin{claim} \label{obs:top_row}
        If a solution~$\mathbb{P}$ of the created  \textsc{DAGSPP} instance~$G'$ is of size $k'$ then, for each row $i \in [k]$, there is a path in $\mathbb{P}$ starting from $v\in T^{i}$ which covers all the vertices of $\{a_j^{i,n+1}\mid j \in [k + 2]\}$.
    \end{claim} 
    \begin{proof}
    First, we will prove that there is a path in $\mathbb{P}$ starting from $v\in T^{i}$ which covers at least one vertex from $\{a_j^{i,n+1}\mid j \in [k + 2]\}$. We will prove this by contradiction. Assume that there is an spp $\mathbb{P}$ such that there exists a row $i$, with $i \in [k]$, such that there is no shortest path starting from~$v$, where $v \in T^{i}$, which covers any of the vertices $a_j^{i,n+1}$, with $j \in [k+2]$. Then, to cover the vertices $a_j^{i,n+1}$, with $j \in [k+2]$, $\mathbb{P}$ needs $k+2$ paths, as every vertex that can reach $a_l^{i,n+1}$, where $l \in [k+2]$, is connected to all the vertices $a_j^{i,n+1}$, where $j \in [k+2]$. Hence, we require $k+2$ paths to just cover $a_j^{i,n+1}$ for $j \in [k+2]$, but this leaves $G'$ with $k'-k-1$ out-degree zero vertices. This contradicts that $\mathbb{P}$ is of size $k'$. Next, it is easy to see that if a path starting at $v \in T^i$ ends at some $a_{q}^{i,n+1}$ where $q \in [k+1]$, it must cover all the vertices $a_{l}^{i,n+1}$ where $l < q$, as there is only one shortest path joining $v$ and $a_{q}^{i,n+1}$. Finally, we can observe that we can extend the path starting at $v \in T^i$ to end at $a_{k+2}^{i,n+1}$, with $i \in [k]$, otherwise it would violate the small size of the partition.    
    \renewcommand{\qedsymbol}{$\Diamond$}
    \end{proof}
\smallskip
\noindent

Hence, if $G'$ has an spp of size $k'$, then the paths in the solution must start at: for $i \in[k]$, $s_i$, $b_1^{i,0}$, $v$ ($v\in T^{i}$) and $s_{i,j}$, with $i< j $.  In the next claims, we will show observations about the end vertices of these paths.

    \begin{claim} \label{obs:siti}
        If $G'$ has an spp of size $k'$, then a path (in this \textsc{DAGSPP} solution) starting at $s_i$ has to end at $t_i$, with $i \in [k]$.
    \end{claim}
    
\begin{proof}
 Let $\mathbb{P}$ be a \textsc{DAGSPP} solution of size~$k'$. The shortest path from $\mathbb{P}$ that covers~$t_i$ cannot have started at $s_j$, where $j < i$, as this path would consist of two vertices due to the arc $(s_j,t_i)$ and hence $\mathbb{P}$ would also contain a path starting at $s'_j$, which is a contradiction to its small size. Similar arguments work if the path would start at any {column start terminal} $s_{l,m}$ or at any $b_1^{l,0}$, with $l < m$, $l \leq i$. Also, the path does not start at $v \in T^{m}$, with $m \leq i$, as this has to go through $a_1^{m,n+1}$ and there is no path from $a_1^{m,n+1}$ to $t_i$. Hence, the path should start at vertex~$s_i$.\renewcommand{\qedsymbol}{$\Diamond$}
    \end{proof}

    \begin{claim}  \label{obs:sijtij}
        If $G'$ has an spp of size $k'$, then any path in this partition starting at $s_{i,j}$ has to end at $t_{i,j}$, with $i < j$.
    \end{claim}
    \begin{proof}
        After \autoref{obs:top_row} and \autoref{obs:siti}, we are left with $k +  \frac{k \dot (k-1)}{2}$ paths to cover the remaining graph. In the remaining graph, there are $k +  \frac{k \dot (k-1)}{2}$ vertices with in-degree $0$, these are $(1) \: s_{i,j}$, for $i < j$ and $j \leq k$ and $(2) \: b_1^{i,0}, i \in [k]$. The path starting at $b_1^{l,0}$ where $l \in [k]$ and $l \leq i$ cannot cover  $t_{i,j}$, as the path would be the arc $(b_1^{l,0}, t_{i,j})$, hence creating one more in-degree $0$ vertex, i.e., $b_2^{l,0}$.  Also, path starting at $s_{i,j}$ cannot end at $t_{m,l}$ where $1 < m < l \leq k $ for $k \neq i$ and $l \neq j$, due to the arcs $E_{\text{fcsct}}$ or as there is no path from $s_{i,j}$ to reach it. Therefore, a path starting from $s_{i,j}$ has to end at $t_{i,j}$, with $i < j$. 
    \renewcommand{\qedsymbol}{$\Diamond$}\end{proof}

Therefore, if $G'$ has an spp of size $k'$, then for each row $i$, with $i \in [k]$, any path starting at $s_i$ has to skip exactly one gadget and end at $t_i$. If we do not skip any gadget, $s_{i,j}$ cannot be connected to $t_{i,j}$, with $j >i$. Once we skip a gadget, we are at the bottom level of the row, and hence we cannot skip again. 
Through this skipped gadget, $s_{i,j}$ is connected to $t_{i,j}$. The bottom of the row is covered by a path starting from $b_1^{i, 0}$ and the top of the row is covered by a path starting at $v \in T^{i}$ after the skipped gadget and ending at $a_{k+2}^{i,n+1}$.

Finally, we have the following claim about $G'$ that follows by construction.

    \begin{claim} \label{obs:connect}
    There exists a path between $s_{i,j}$ and $t_{i,j}$ through two gadgets $G^{i,u}$ and $G^{j,v}$, where $i>j$, if and only if there is an edge $uv$ in the graph~$G$. 
    \end{claim}

\begin{claim} \label{claim:W-1-main-claim}  
$G$ has a $k$-clique if and only if $G'$ has a shortest path partition of size $k'$.
\end{claim}

\begin{proof}
 If vertices $v_1,v_2,\dots, v_k$ form a clique in $G$, then we construct an spp of size~$k'$ in $G'$ as follows: for each row $i$, with $i\in[k]$, we take a path starting at $s_i$ and leading to $t_i$ while skipping gadget $G_{i,v_i}$; we cover all the vertices on the left of the skipped gadget with a single path starting at $b_1^{i,0}$, and similarly, we cover all the vertices on the right of the skipped gadget with a single path ending at $a_{k+2}^{i,n+1}$. Then, we connect $s_{i,j}$ to $t_{i,j}$, with $i < j$, through the skipped gadget (this is possible because of \autoref{obs:connect}).
 
 For the other direction, assume we have an spp $\mathbb{P}$ of size~$k'$ in $G'$: include the vertices corresponding to any gadget which was skipped by a path which connects $s_i$ to $t_i$ for $i \in [k]$ into the $k$-element set~$C$. This forms a clique in~$G$, as any two vertices of $\{v_i,v_j\} \in C$, are distinct and have an edge between them because of the {verifiers}.\renewcommand{\qedsymbol}{$\Diamond$}
\end{proof}
    
\noindent    
As the construction is polynomial-time, \W{1}-hardness follows.   
\end{proof}

\section{\textsf{XP} Algorithms for SPP} \label{xp-algo}
We now present our \textsf{XP} algorithms to prove the following result. We note that the result for \textsc{USPP} was also proved in~\cite{dumas2024graphs} (with a different method: they first give an upper bound on the treewidth of the graph, and then they apply Courcelle's theorem).
The basic idea is to use existing algorithmic results for \textsc{Disjoint (Shortest) Paths} from the literature.

\begin{theorem}\label{lem:SPPinXP}
The following problems are in \XP when parameterized by the number of paths: \textsc{USPP}, \textsc{DAGSPP}, and \textsc{DSPP} when restricted to planar directed graphs, to directed graphs of bounded directed treewidth, or when $k=2$.\footnote{See~\cite{DBLP:journals/jct/JohnsonRST01} for the definition of \emph{directed treewidth}.}
\end{theorem}

\noindent For our \XP algorithms, the following combinatorial result is crucial, that allows us to reduce the problem to \textsc{DP} (without the shortest path requirement).

\begin{lemma}\label{lemm:DSP}
Let~$G=(V,E)$ be a graph (directed or undirected). Then, $V$ can be partitioned into $k$ vertex-disjoint shortest paths if and only if there are $k$ vertex-disjoint paths between some $s_i$ and $t_i$, for $1\leq i\leq k$, such that $\sum_{i=1}^{k} d(s_i,t_i) = |V|-k$.
\end{lemma}
\begin{proof} 
Let $G=(V,E)$ be a directed or an undirected graph.
If $V$ can be partitioned into $k$ vertex-disjoint shortest paths $P_i$,  for $1\leq i\leq k$, then we can take their endpoints as $s_i$ and $t_i$. As the paths are shortest and form a partition,  $\sum_{i=1}^{k} d(s_i,t_i) = |V|-k$ must hold.
Conversely, if there are $k$ vertex-disjoint  paths in~$G$, connecting, say,  $s_i$ and $t_i$, for $1\leq i\leq k$, then the condition $\sum_{i=1}^{k} d(s_i,t_i) = |V|-k$ can only be satisfied if each vertex of the graph belongs to exactly one of the paths and if all the paths are shortest paths.
Hence, $V$ can be partitioned into $k$ vertex-disjoint shortest paths.
\end{proof}

\autoref{lemm:DSP} allows us to prove \autoref{lem:SPPinXP} by enumerating all possible $X:=\{(s_1,t_1), \dots, (s_k,t_k) \}$, resulting in an instance  of \textsc{DSP}. Now, we can either apply known algorithms with \XP running time (or better), for either \textsc{DP} or \textsc{DSP}. We give details of this proof strategy in the following.

\begin{proof}[Proof of \autoref{lem:SPPinXP}] Let $G=(V,E)$, together with an integer~$k$, be an instance of \textsc{SPP} (i.e., \textsc{DSPP}, \textsc{DAGSPP} or \textsc{USPP}, the arguments are complete analogous for these cases). 
By \autoref{lemm:DSP}, given an instance $(G,k)$ of \textsc{SPP}, with $G=(V,E)$, it is sufficient to iterate (in \XP time) through all possible $X:=\{(s_1,t_1), \dots, (s_k,t_k) \}$ where $s_i,t_i \in V$ and $\sum_{i=1}^{i=k} (d(s_i,t_i) + 1) = |V|$. The latter can be checked in polynomial time using standard algorithms for computing the distance.
Then, for each such $X$ we form an instance $(G,X)$ of the version of \textsc{DP} (or \textsc{DSP}) corresponding to our initial instance (i.e., \textsc{DDP}, \textsc{DAGDP} or \textsc{UDP}, or alternatively, \textsc{DDSP}, \textsc{DAGDSP} or \textsc{UDSP}). Now, $(G,k)$ is a \yes-instance of \textsc{SPP}, respectively, if and only if some constructed $(G,X)$ instance is a \yes-instance of \textsc{DP} (or of \textsc{DSP}).

We next see in which cases \textsc{DP} or \textsc{DSP} can be solved in \XP time by known algorithms.

For DAGs, Fortune et al.~\cite[Theorem~3]{ForHopWyl80} showed that, for arbitrary fixed directed pattern graphs, subgraphs homeomorphic to them can be detected in DAGs in polynomial time (even when the node-mapping is fixed).
Hence, taking as a pattern graph a collection of $k$ isolated directed edges and as node-mapping the chosen terminals, the according homeomorphism problem is identical to \textsc{DAGDP} (with fixed $k$). Hence, the instance $(G,X)$ can be solved in time $\Oh(|V|^{f(|X|)})=\Oh(|V|^{f(k)})$ for some function~$f$.  Alternatively, an \XP algorithm for \textsc{DAGDSP} is also given in~\cite{DBLP:conf/esa/Berczi017}.

For planar (directed) graphs, we refer to~\cite{DBLP:conf/focs/CyganMPP13} for an \FPT algorithm for \textsc{DDP}, or to~\cite{DBLP:conf/esa/Berczi017} for an \XP algorithm for \textsc{DDSP}. An \XP algorithm is given for \textsc{DDP} on graphs of bounded directed treewidth~\cite{DBLP:journals/jct/JohnsonRST01}. For the case of $k=2$, see~\cite{DBLP:conf/esa/Berczi017}. Note that results for other classes exist and can be applied, see e.g.~\cite{DBLP:journals/jgt/Bang-JensenCM17}.

For the undirected case, we can refer to the seminal paper of Robertson and Seymour~\cite{ROBERTSON199565}, yielding a cubic-time algorithm on input $(G,X)$ for \textsc{UDP}, an algorithm that was later improved to be quadratic in~\cite{kawarabayashi2012disjoint}, yielding again the \XP-claim. Alternatively, \XP algorithms for \textsc{UDSP} were also given in~\cite{DBLP:journals/siamdm/BentertNRZ23,lochet2021polynomial}.

These results establish the statement.
\end{proof}
    
Notice that these algorithmic results rule out \paraNP-hardness results (in contrast to those existing for PP and IPP) for the corresponding problems.

\section{Neighborhood Diversity and Vertex Cover Parameterizations} \label{nd-param}

One of the standard structural parameters studied within parameterized complexity is the \emph{vertex cover number}, i.e., the size of the smallest vertex cover that a graph has. As graphs with bounded vertex cover number are highly restricted,  less restrictive parameters that generalize vertex cover are interesting, as \emph{neighborhood diversity} is, introduced by Lampis~\cite{lampis2012algorithmic}.

\subsection{Neighborhood diversity for \textsc{USPP} and \textsc{UIPP}}

Crucial to our \FPT-results is the following interesting combinatorial fact.

\begin{proposition} \label{nd:diam}
If $G$ is connected, then any induced path has length at most $\nd(G)$. In particular, $\diam(G) \leq \nd(G)$.
\end{proposition}
\begin{proof}
Let $P$ be an induced path with $k$ vertices in $G$ (note that any shortest path, in particular a diametrical path, is also an induced path). Note that if $k\geq 4$, no two vertices of $P$ are twins in $G$, and so, every vertex of $P$ needs to be in a different neighborhood diversity class, and so, the claim is true (in fact even stronger, $k\leq\nd(G)$). Hence, if $\nd(G)\geq 2$, we are done: $G$ does not contain an induced path of length at least $\nd(G)+1$ (i.e. with $k\geq \nd(G)+2$ vertices). If $k\leq 3$, then the endpoints of $P$ may be twins and thus, may belong to the same neighborhood diversity class, if that class forms an independent set (but no other pair of vertices of $P$ may lie in the same class). Thus, if $\nd(G)\leq 2$, $G$ may contain an induced with $\nd(G)+1$ vertices, but no more.
\end{proof}

\noindent We will also use the following result.

\begin{theorem}[{\cite[Theorem 6.5]{CygFKLMPPS2015}}]\label{thm:ILP}
An \textsc{Integer Linear Programming} instance of size $L$ with
$p$ variables can be solved using $$\Oh(p^{2.5p+o(p)}\cdot (L + \log M_x) \log(M_x M_c ))$$ arithmetic operations and space polynomial in $L + log M_x$, where $M_x$ is an upper bound on the absolute value a variable can take in a solution, and $M_c$ is the largest absolute value of a coefficient.
\end{theorem}

    Now, call two induced paths $P_i$ and $P_j$ (viewed as sets of vertices) \emph{equivalent}, denoted by $P_i\equiv P_j$, in a graph~$G$ with $d=\nd(G)$ if  $|P_i \cap C_{l}| = |P_j \cap C_{l}|$ for all $l$ with $1 \leq l \leq d$, where $C_1,\dots,C_d$ denote the neighborhood diversity classes.
Any two equivalent induced paths have the same length.
\autoref{nd:diam} and our discussions imply that there are 
$2^{\Oh(\nd(G))}$ many equivalence classes of induced paths $(+)$.

\begin{theorem}\label{thm:USPP-ND}
    \textsc{UIPP} and \textsc{USPP} are \FPT when parameterized by
    neighborhood diversity.
\end{theorem}

\begin{proof}As we can solve \textsc{UIPP} and \textsc{USPP} separately on each connected component, we can assume in the following that the input graph is connected. 


Recall that there exists a partition of $V(G)$ into $d$ $\nd$-classes $C_1, \dots , C_d$, and each $C_i$ for $i \in [d]$ either induces a clique or an independent set and all its vertices are mutual twins. As mentioned before, such a partition can be computed in linear time using~\cite[Algorithm 2]{DBLP:conf/stacs/HabibPV98}. 

Compute (in time $2^{\Oh(\nd(G))}$, by $(+)$) the set of all $2^{\Oh(\nd(G))}$ induced path equivalence classes; this can be represented by a set $\cal P$ of \emph{types} of induced paths (each represented by one representative induced path), in which any two induced paths are not equivalent. For \textsc{USPP}, we discard those classes that do not yield shortest paths.
Construct and solve the following Integer Linear Program (\text{ILP}), with a variable $z_P\geq 0$ corresponding to each type of path $P\in \cal P$.
Each $P\in \cal P$ is characterized by a vector $(P^1,\dots,P^d)$ with $P^j=\vert C_j\cap P\vert $ (note that, as argued in the proof of \autoref{nd:diam}, $P^j\leq 2$).

     $$\text{minimize} 
     \sum_{P\in\cal P}z_P
     $$     $$\text{subject to }\sum_{P\in \cal P}z_P\cdot P^j=|C_j|\text{  for all }j\in [d]$$

The variable $z_P$ encodes how many
induced paths equivalent to~$P$ are taken into the solution. Hence, the objective function expresses minimizing the number of induced paths used in the partition. The constraints ensure the path partitioning.

\begin{claim}\label{claim:ND-ILP}
There exists an ipp (respectively spp) of $G$ with $k$ induced (respectively, shortest) paths if and only if the objective function attains the value $k$ in the ILP described above.
\end{claim}
\begin{proof}
  For the forward direction, consider the induced/shortest paths $\mathbb{P}= \{P_1,P_2, \dots, P_k\}$ as a solution. To construct a solution of the \text{ILP}, for every $P \in \cal P$, $z_P$ equals the number of  paths in $\mathbb P$ that are {equivalent} to $P$. Clearly, $\sum_{P\in \cal P} z_P = k$. The 
  constraint is satisfied as every vertex is covered exactly once by $\mathbb P$.

For the other direction, suppose that there exist integers $\{z_{P}\mid P\in \mathcal{P}\}$  which satisfy the \textsc{ILP} and $\sum_{P\in\cal P}z_P = k$. We can produce an ipp/spp of $G$ with $k$ paths as follows. As long as there exists a set $P\in\cal P$ with $z_P > 0$, construct a path by picking exactly $P^j$ vertices from each $C_j$. Since the vertices in $C_j$ form either a clique or an independent set of mutual twins, this choice can safely be arbitrary. Since $P$ corresponds to a valid type of induced/shortest path, the set of selected vertices forms such a path. Remove these vertices from~$G$, and set $z_P := z_P - 1$. This algorithm produces exactly $k$ paths which form an ipp/spp that covers each vertex exactly once because of the satisfied constraints of the ILP.\renewcommand{\qedsymbol}{$\Diamond$}
\end{proof}

Notice that the number of variables of the ILP is exactly $\vert \mathcal{P}\vert$, which is $2^{\Oh(\nd(G))}$, by $(+)$. Now, we apply Theorem~\ref{thm:ILP} with $p=2^{\Oh(\nd(G))}$, $L=\Oh(n)$, $M_x=M_c=n$ to get our \FPT\ claim.
\end{proof}

 


It is worth mentioning that by Theorem~\ref{thm:ILP}, the time consumption of the proposed algorithms is at worst double-exponential in the parameter: ignoring lower-order terms, the growth is $\Oh\left(2^{2.5 d\cdot 2^d}\right)$, while the space consumption is rather single-exponential, $\Oh(2^d)$. In certain cases, however, the number of induced/shortest path equivalence classes of a graph (that determines that number of variables of the derived ILP) need not be exponential in the neighborhood diversity, which would immediately improve the running time estimates for the proposed algorithm. Nonetheless, it would be interesting to give better parameterized algorithms for these path partitioning problems.

Moreover, using the mentioned estimate of Lampis~\cite[Lemma 2]{lampis2012algorithmic} that a graph with vertex cover number $k$ has neighborhood diversity at most $2^k+k$, we can immediately infer the following result, where the dependence on the parameter is even worse than that for neighborhood diversity.

\begin{corollary}\label{cor:vc-param}
\textsc{USPP} and \textsc{UIPP} are \FPT when parameterized by
vertex cover number.
\end{corollary}

Note that we also have a similar result for \textsc{UPP}, either  by a more general approach in~\cite{GajLamOrd2013}, or based on a more direct and simple reasoning, as shown later in \autoref{thm:vc-param}.


\subsection{Neighborhood diversity of directed graphs for \textsc{DSPP} and \textsc{DIPP}}

Now we want to try to adapt these neighborhood diversity results for the directed case. For this reason we want to introduce the new notion of \emph{directed neighborhood diversity} (dnd for short) that might be of independent interest. To this end, let $N^-(v)\coloneqq \{ x\in V \mid (x,v)\in E\}$ and $N^+(v)\coloneqq \{ x\in V \mid (v,x)\in E\}$ be the in- and out-neighborhood of~$v$, respectively. For two vertices $v,u\in V$, we define $v\sim_{dnd}u$ to hold if the following three conditions are satisfied:
\begin{enumerate}
    \item $N^-(v)\setminus \{u\}=N^-(u)\setminus \{v\}$,
    \item $N^+(v)\setminus \{u\}=N^+(u)\setminus \{v\}$, and
    \item $(v,u)\in E \Leftrightarrow (u,v) \in E$.
\end{enumerate}

\begin{lemma}
    For any directed graph $G=(V,E)$, the relation $\sim_{dnd}$ is an equivalence relation.
\end{lemma}
\begin{proof}
Trivially, $\sim_{dnd}$ is reflexive and symmetric. We are left to show transitivity. Let $v,u,w\in V$ be different vertices with $v \sim_{dnd} u$ and $u \sim_{dnd} w$.
    
First assume that $(v,u),(u,v)\in E$. This implies $v\in N^-(u) \setminus \{ w \}= N^-(w)\setminus \{ u\}$ and $v\in N^+(u) \setminus \{ w \}= N^+(w)\setminus \{ u\}$. Hence, $(v,w),(w,v)\in E$. As $N^-(u) \setminus \{ v \}= N^-(v)\setminus \{ u\}$ , $(u,w),(w,u)\in E$.
    Thus, \begin{equation*}
        \begin{split}
            N^-(v) \setminus \{ w \} &= ((N^-(v)\setminus \{u\})\cup \{u\})\setminus \{w\} = ((N^-(u) \setminus \{v\})\cup \{u\})\setminus \{w\}\\ 
            & = ((N^-(u) \setminus \{w\})\cup \{u\})\setminus \{v\} = ((N^-(w) \setminus \{u\})\cup \{u\})\setminus \{v\}\\
            &= N^-(w)\setminus \{v\}.
        \end{split}
    \end{equation*}
Analogously, $N^+(v) \setminus \{ w \} = N^+(w)\setminus \{v\}$.    Therefore, $v \sim_{dnd} w$.

Now assume $(v,u),(u,v)\notin E$. Then, $v\notin N^-(u) \setminus \{ w \}= N^-(w)\setminus \{ u\}$ and $v \notin N^+(u) \setminus \{ w \}= N^+(w)\setminus \{ u\}$. Thus, $(v,w),(w,v)\notin E$.  Since $N^-(u) \setminus \{ v \}= N^-(v)\setminus \{ u\}$ , $(u,w),(w,u)\notin E$. Hence, 
    \begin{equation*}
        \begin{split}
            N^-(v) \setminus \{ w \} &= (N^-(v)\setminus \{u\})\setminus \{w\} = (N^-(u) \setminus \{v\})\setminus \{w\}\\ 
            & = (N^-(u) \setminus \{w\})\setminus \{v\} = (N^-(w) \setminus \{u\})\setminus \{v\}\\
            &= N^-(w)\setminus \{v\}.
        \end{split}
    \end{equation*}
    Therefore, $\sim_{dnd}$ is transitive and hence an equivalence relation.
\end{proof}

Note that the third condition in the definition of relation $\sim_{dnd}$ is needed to obtain an equivalence relation. To see this, consider for example the directed graph on vertices $u,v,w$ consisting of a transitive triangle, where $u$ is the source and $w$ is the sink. Without the third condition, we would have $u\sim v$ and $v\sim w$, but $u\not\sim w$, so $\sim$ would not be an equivalence relation.

Now we want to consider the equivalence classes under this equivalence relation. The previous proof provides  an observation for these.
\begin{corollary}
    Let $G=(V,E)$ be a directed graph and $C_1,\ldots, C_d$ be the equivalence classes with respect to $\sim_{dnd}$. Then for each $i \in \{1,\ldots,d\}$, $C_i$ is either an independent set or a bi-directed clique.
\end{corollary}

A similar result for the neighborhood diversity on undirected graphs is well-known. There, an equivalence class is either an independent set or a clique. But this is not the only similarity, as we see next. 

\begin{remark}
If we take an undirected graph $G=(V,E)$ and build the directed graph $G'=(V,E')$ with $E'\coloneqq \{(v,u), (u,v) \mid \{v,u\}\in E\}$, then $\text{$\sim_{nd}$}=\text{$\sim_{dnd}$}$, where $\sim_{nd}$ is the underlying relation of the undirected neighborhood diversity. This holds, since $N(v) = N^-(v) = N^+(v)$ and $\{v,u\} \in E$ if and only if $(v,u) \in E'$ if and only if $(u,v) \in E'$, which is the case if and only if $\{u,v\}\in E$.  
\end{remark}

As we can see, the definition above is a natural way to generalize the neighborhood diversity relation from undirected to directed graphs. Therefore, we define the \emph{directed neighborhood diversity} of a graph $G=(V,E)$ (denoted by $\dnd(G)$) as the number of equivalence classes under $\sim_{dnd}$.

\begin{lemma}\label{lem:path_dnd_classes}
Let $G=(V,E)$ be a directed graph with directed neighborhood diversity classes $C_1,\ldots,C_{\dnd(G)}$. For any induced path $P=v_1,\ldots,v_{\ell}$ in~$G$, there is at most one $i\in \{1,\ldots,\dnd(G)\}$ such that $\vert C_i \cap P \vert >1$. In this case, $\{v_1,v_{\ell}\} = C_i \cap P$. 
\end{lemma}
\begin{proof}
Let $P=v_1\ldots v_{\ell}$ be an induced directed path and $C_i$ be a directed neighborhood diversity class with $\{v_j,v_k\} \subseteq C_i \cap P$ and $v_j \neq v_k$. Without loss of generality, $j<k$. Assume $j \neq 1$. Then $v_{j-1}\in N^-(v_j) \setminus \{v_k\} = N^-(v_k) \setminus \{v_j\}$. This contradicts the fact that $P$ is induced. Therefore, $j=1$. Analogously, $k=\ell$ (substitute $v_{j-1}$ by $v_{k+1}$ and $N^-$ by $N^+$).     
\end{proof}

\noindent
This provides the following analogue of \autoref{nd:diam} for directed graphs.

\begin{corollary}
    Let $G$ be a directed graph. Any induced path has length at most $\dnd(G)$.
\end{corollary}

This is a tight bound. Indeed, for any given $d \in \mathbb{N}$, there exists a directed graphs $G_d$ with $\diam(G_d)=\dnd(G)$: For $d=1$, just use a path of two vertices. For the other cases, define $G_d=(V_d,E_d)$ with $V_d\coloneqq \{v, v', v_1, \ldots, v_{d-1}\}$ and $E_d \coloneqq \{ (v_i,v_{i+1})\mid i\in \{1,\ldots, d-2\}\}\cup \{(v_{d-1}, v),(v_{d-1}, v'),(v,v_1),(v',v_1)\}$. The directed neighborhood diversity classes are $C_i = \{v_i\}$ for $i\in \{1,\ldots, d-1\}$ and $C_d=\{v,v'\}$. The only path between $v$ and $v'$ is $v,v_1,\ldots,v_{d-1},v'$.

By this corollary, we can obtain an analogue of \autoref{thm:USPP-ND} for directed graphs as follows:


\begin{theorem}\label{thm:dspp+dipp-dnd}
  \textsc{DSPP} and \textsc{DIPP} are \FPT when parameterized by the directed neighborhood diversity of the input graph.
\end{theorem}
\begin{proof}
The proof is exactly the same as the one of \autoref{thm:USPP-ND}, using the crucial fact from \autoref{lem:path_dnd_classes} which implies that there are at most $2^{\Oh(\dnd(G))}$ types of induced/shortest directed paths (that can be computed in $2^{\Oh(\dnd(G))}n$ time), and since each directed neighborhood diversity class is a set of twins, the choice of some vertex in the class is completely arbitrary. Thus, the same ILP as in the proof of \autoref{thm:USPP-ND} provides the desired \FPT algorithm.
\end{proof}

Let us compare $\dnd(G)$ with $\nd(U(G))$ of the underlying undirected graph $U(G)$.
By the very definitions, we observe:
\begin{proposition}\label{prop:dnd-und}
For every directed graph $G$, $\nd(U(G))\leq \dnd(G)$.
\end{proposition}

Note that equality in \autoref{prop:dnd-und} can be strict, consider for example any tournament $G$ of order~$n$. Since there is exactly one arc between every pair $u,v$ of vertices of~$G$, the third condition for $u\sim_{dnd}v$ is never satisfied, so $\dnd(G)=n$. However, $U(G)$ is a complete graph, so $\nd(U(G))=1$.

Similarly to the undirected case, the directed neighborhood diversity number of  the underlying graph is upper-bounded by an exponential function of the vertex cover number of the underlying undirected graph. However, the details are a bit different, as shown next.

\begin{proposition}
For every directed graph $G$, $\dnd(G)\leq 4^{\vc(U(G))}+\vc(U(G))$.  
\end{proposition}

\begin{proof} Namely, consider all (not necessarily proper) 4-colorings of the vertex cover set.
They should come with the following meaning (with respect to a vertex~$u$ in the independent set):
\begin{itemize}
    \item $0$: not touched by any edge incident to $u$,
    \item $1$: in $N^+(u)\setminus N^-(u)$,
    \item $2$: in $N^-(u)\setminus N^+(u)$,
    \item $3$: in $N^+(u)\cap N^-(u)$.
\end{itemize}
This idea already shows the claim. 
Namely, at worst each vertex in a minimum vertex cover forms its own equivalence class, and two vertices $u,u'$ in the independent set (the complement of the vertex cover) are equivalent if and only if they induce the same 4-coloring as defined above. As there are at most  $4^{\vc(U(G))}$ many such 4-colorings, this upper-bounds the number of equivalence classes with independent set vertices only.
\end{proof}

This immediately provides the following result, which we will refine and make more explicit in the next subsection.

\begin{corollary}\label{cor:dspp+dipp-vc}
 \textsc{DSPP}  and \textsc{DIPP} are \FPT  when parameterized by the vertex cover number of the underlying undirected graph.
\end{corollary}

We could also define the directed neighborhood diversity by only in-arcs or out-arcs. But different from the undirected case, we must specify if we consider the open or closed neighborhood for the diversity. If we would define the relation~$\sim$ for two vertices $v,u\in V$ only by $N^-(v) \setminus \{u\} = N^-(u)\setminus \{v\}$ (completely in analogy to the undirected case), then the transitive triangle would again be a counter-example for the required transitivity. As for $\sim_{dnd}$, we could include the condition $(u,v)\in E$ if and only if $(v,u) \in E$. But even for this modification, there is a  counter-example for transitivity of this relation~$\sim$:  $G=(\{u,v,w\},E)$ with $E \coloneqq \{(w,v), (u,v), (v,w)\}$. Here, $v\sim w$ and $v \sim u$ but $w \not\sim u$. 

So either we try to discuss further adaptations of the notion of neighborhood diversity towards directed graphs, 
or we can base a definition on the open or closed neighborhoods, e.g., $u\sim_{\text{open}^-}v\iff N^-(u)=N^-(v)$. But we did not find a way such that these definitions help generalize our algorithms. For example, consider for the open neighborhood (the closed one works analogously) $G=(\{v,u,w,x,y\},E)$ with $E\coloneqq \{(y,v),(y,u),(v,w),(v,x)\}$. The three $\sim_{\text{open}^-}$-classes are $C_1\coloneqq\{y\}$, $C_2\coloneqq \{v,u\}$, $C_3 \coloneqq \{w,x\}$. So our algorithm would advise to use a path $C_1,C_2,C_3$ and $C_2,C_3$. But this is not possible.  


\subsection{More on vertex cover parameterization}

We now study the parameterization by the vertex cover number of the (underlying) graph, which provides better running times than via the more general (directed) neighborhood diversity parameterization of \autoref{cor:dspp+dipp-vc}.

\begin{proposition}\label{thm:vc-param}
\textsc{UPP}, \textsc{USPP} and \textsc{UIPP} are \FPT  when parameterized by vertex cover number, as they admit single-exponential size kernels. The same holds for directed graphs for \textsc{DPP}, \textsc{DSPP} and \textsc{DIPP}, for the vertex cover number of the underlying graph.
\end{proposition}
\begin{proof}
Assume that $\vc(G)$ is our problem parameter, given a graph $G=(V,E)$, and let $C\subseteq V$ be a minimum vertex cover (which can be computed in \FPT time by \cite{CheKanXia2010,HarNar2024}). 
Now we claim that if there are more than $2\vc(G)$ many vertices in the independent set $ I=V\setminus C$ that have the same type, then at least one of them must form a single-vertex path in a minimum pp $\mathbb{P}$ of $G$. Namely, consider a vertex $x_0\in I$ with $x_1,\dots,x_{2\vc(G)}\in I$ being $2\vc(G)$ many other vertices of the same type, i.e., $N_G(x_0)=N_G(x_i)$ for $i=1,\dots,2\vc(G)$. Any of these vertices~$x_i$ (with $i=0,\dots,2\vc(G)$) that is not forming a  single-vertex path in $\mathbb{P}$ demands at least one neighbor $y_i\in N_G(x_0)$ as a neighbor on its path. This neighbor $y_i$ can only be `shared' by one other vertex $x_i'\in I$ on this path, i.e., there might be at most one $i'\neq i$ such that $y_i=y_{i'}$. But as $|N_G(x_0)|\leq \vc(G)$ and as there are (at least) $2\vc(G)+1$ many vertices of the same type as $x_0$, so by pigeon hole, at least one vertex must form a single-vertex path in~$\mathbb{P}$.

Implementing this reduction rule leaves us with $|I|\leq 2\vc(G)\cdot 2^{\vc(G)}$, as the number of neighborhood equivalence classes within $I$ is clearly bounded by $2^{\vc(G)}$. Hence, we can assume that the size of $G$ is bounded by a function in $\vc(G)$ which shows that \textsc{UPP}, \textsc{UIPP} and \textsc{USPP} parameterized by vertex cover, have a kernel and are hence in \FPT.

Note that the exact same argument holds for directed graphs as well, based on our reasoning in the preceding subsection.
\end{proof}

We 
can  also obtain a direct \FPT algorithm for \textsc{UPP}, leading to the following result.
It is not that clear how this result can be adapted to the other path partition variants (\textsc{UIPP} and \textsc{USPP}) studied in this paper.

\begin{proposition}\label{prop:UPP-VC-FPT}
\textsc{UPP} is \FPT  when parameterized by
vertex cover number, testified by an algorithm running in time $\Oh^*\left(\vc(G)!\cdot 2^{\vc(G)}\right)$ on input~$G$.
\end{proposition}
\begin{proof}
The proposed algorithm runs as follows on the input graph $G=(V,E)$:\\[1ex] (0) Compute a minimum vertex cover~$C$ of the input graph; this can be done in time $\Oh(c^{\vc(G)})$ for some $c<1.3$, see~\cite{CheKanXia2010,HarNar2024}. We call the complement of~$ C$ as~$I$, being an independent set.\\ We also keep track of a `record partition' $\mathbb{P}$ that we initialize with the path partition where each path is consisting of one vertex.\\[1ex]
(1) Iterate through all permutations of the vertices in~$C=\{v_1,\dots,v_r\}$, where $r=\vc(G)$. The idea is that each such permutation describes a sequence of vertices in the order in which each of the paths that are considered in a partition traverse the $C$-vertices in this order.\\[1ex]
(2) For each such permutation~$\pi$, say, $v_{\pi(1)},\dots,v_{\pi(r)}$, 
walk through all $(r-1)$-dimensional bit-vectors $\bar{b}=b_1b_2\cdots b_{r-1}$. These bit-vectors describe where the sequence $v_{\pi(1)},\dots,v_{\pi(r)}$ is connected (or not) on a path. More precisely,  $b_j=1$ if and only if $v_{\pi(j)}$ and $v_{\pi(j+1)}$ are on the same path in the path partition~$\mathbb{P}(\pi,\bar b)$ that we are constructing. Therefore, a bit-vector with $\ell$ zeros will describe a collection of paths with $\ell+1$ paths that contain vertices of~$C$. 
Also, if $P\in \mathbb{P}(\pi,\bar b)$, then there exist $i\leq j$ such that, if we view $P$ as a set of vertices, $P\cap C=\{v_{\pi(i)},\dots,v_{\pi(j)}\}$ and for all indices $\iota\in [i,j-1]$, $b_\iota=1$, while $b_{i-1}=0$ or $i=1$ and $b_j=0$ or $j-1=r$.

\begin{remark}
Under the assumption that $I=V\setminus C$ might be much larger than~$C$, as suggested by the exponential-size kernel obtained in \autoref{thm:vc-param}, this does not relate $\ell$ to the solution size $k$, as then $k$ may be in the order of~$|I|$ rather than just 
in the order of $\ell\leq r=|C|$. Yet, we can say that, given a bit-vector~$\bar b$ with $\ell$ zeros and assuming that $I$ is sufficiently large, we will obtain at least $(\ell+1)+(|I|-2(\ell+1))=|I|-(\ell+1)$ many paths, a number that should be smaller than~$k$, as otherwise we can abort. Namely, by our construction, we have $\ell+1$ paths containing vertices from~$C$. In order to get as few paths as possible in our path partition, as many vertices from~$I$ as possible should be integrated into these paths. In each of these paths, we can integrate at most as many vertices as there are vertices belonging to~$C$, plus one, because vertices from~$I$ can occur only as `intermediate vertices' between two vertices from~$C$, or at the very ends of the paths, because $I$ is an independent set in~$G$ by assumption. Hence, at most 
$2r+(\ell+1)$ 
many vertices can be integrated into these $\ell+1$ many paths. Of these vertices, at most $r+(\ell+1)$ 
belong to~$I$, but `the rest' of~$I$ must be covered by single-vertex-paths. However, we could possibly also get more paths in a solution that is consistent with $\bar b$, namely, if no vertices from~$I$ are integrated in the paths that contain vertices from~$C$. In this extreme case, we are looking at a solution with $|I|+(\ell+1)$ many paths. In order to integrate as many vertices from~$I$ into   paths that contain vertices from~$C$, we will suggest to consider a maximum matching in an auxiliary graph, as described in the following.
\end{remark}

\noindent
(3) For all $\pi$ and $\bar{b}$,
create an auxiliary edge-weighted bipartite graph $G'=G'(\pi,\bar b)$ with edge set $E'$ and vertex set $V'$ containing $V$, such that $C$ is part of one set~$C'$ of the bipartition. 
Similarly, let $I'=V'\setminus C'$ contain~$I$.
Moreover, we will define a weight function $w:E'\to\{1,2\}$.
More precisely, to obtain $C'$, add to $C$ (always) the vertex $v_{\pi(1)}'$ and further vertices $v_{\pi(i+1)}'$ if $b_i=0$, and possibly add
vertices $u_i$ into~$I$ if $b_i=1$ to build $I'$.
For a worked-out example explaining our construction, we refer to \autoref{exampleE}.

First assume that $b_i=1$, i.e., $v_{\pi(i)}$ and $v_{\pi(i+1)}$ should lie on a path in $\mathbb{P}(\pi,\bar b)$.
Connect $v_{\pi(i)}$ and $x\in I$ by an edge of weight~$2$ if $v_{\pi(i)}x\in E$ and $v_{\pi(i+1)}x\in E$. 
Moreover, if there is no edge incident to  $v_{\pi(i)}$ so far but $v_{\pi(i)}v_{\pi(i+1)}\in E$, then introduce a new vertex $u_i$ (as part of~$I'$) and the edge $v_{\pi(i)}u_i$ of weight~$1$. 

Secondly, assume that $b_i=0$, i.e., $v_{\pi(i)}$ and $v_{\pi(i+1)}$ should not lie on a path in $\mathbb{P}(\pi,\bar b)$.
Still, the path that ends at $v_{\pi(i)}$ (concerning the $C$-vertices) could contain one more vertex (in~$I$). Therefore, we also connect $v_{\pi(i)}$ and $x\in I$ by an edge of weight~$1$ if $v_{\pi(i)}x\in E$.  Similarly,  connect $v_{\pi(r)}$ and $x\in I$ by an edge of weight~$1$ if $v_{\pi(r)}x\in E$.
Also the path that starts at $v_{\pi(i+1)}$ (concerning the $C$-vertices) could contain one more vertex (in~$I$).
To cover this case, we also connect $v_{\pi(i+1)'}$ and $x\in I$ by an edge of weight~$1$ if $v_{\pi(i+1)}x\in E$.
Similarly, we connect $v_{\pi(1)'}$ and $x\in I$ by an edge of weight~$1$ if $v_{\pi(1)}x\in E$.

This concludes the description of the auxiliary graph~$G'=(V',E')$ and the weight function $w$.

Now, compute a maximum weighted matching~$M$ in $G'$. If the vertices $v_{\pi(i)}$, $v_{\pi(i+1)}$,  \dots, $v_{\pi(j)}$ are matched, then this clearly corresponds to a path within~$G$ that traverses $v_{\pi(i)}$, $v_{\pi(i+1)}$,  \dots $v_{\pi(j)}$ in that order, possibly using vertices from~$I$ on its way, including possibly starting or ending in~$I$.
Let us make this more precise.

First, we test if there is any edge $v_{\pi(i+1)}'x$ in the matching, while $v_{\pi(i)}x$ was also an edge in $G'$. This would mean that we found a situation where we guessed (by the bit-vector~$\bar b$, as $b_i=0$) that one path ends at $v_{\pi(i)}$ and the next path starts at $v_{\pi(i+1)}$. Our matching proves that we could connect both paths, i.e., the path partition $\mathbb{P}(\pi,\bar b)$ would not be minimal. Also, there will be another setting of the bit-vector~$\bar b$ that covers this (better) case. Therefore, we need not consider this case any longer. 
Similarly, we can skip the case when $v_{\pi(1)}'x$ in the matching, while $v_{\pi(r)}x$ was also an edge in~$G'$.

Now, we describe $\mathbb{P}(\pi,\bar b)$ more precisely.
Our first path~$P_1$ should contain $v_{\pi(1)}$. If $v_{\pi(1)'}$ is matched with some $x_0\in I$, then $P_1$ starts with $x_0$, with the second vertex then being $x_1=v_{\pi(1)}$, while otherwise it starts with~$x_1$.
If $b_1=0$, then it might be the case that $x_1$ is matched with $x_2\in I$, where~$P_1$ will then end, or $x_1$ is not matched, which means that $P_1$ already ends with $x_1$.
If $b_1=1$, then we expect that $x_1$ is matched with some vertex (otherwise, the guessed bit-vector would be wrong and we can stop any further computation of this case here).
If $x_1$ is matched with some $x_2\in I$, this means that (within~$G$) we find the edges $x_1x_2$ and $x_2x_3$, with $x_3=v_{\pi(2)}$.
Otherwise, $x_1$ is matched with~$u_1$. By construction, this means that there was no $x\in I$ with $x\in N_G(x_1)\cap N_G(v_{\pi(2)})$, but $x_1\in N_G(v_{\pi(2)})$. Hence, we can take $x_2=v_{\pi(2)}$ as the next vertex on~$P_1$.
Continuing this way, we will construct the first path~$P_1$, whose end is determined by the first~$j$ with $b_j=0$. Then, the second path~$P_2$ will start with~$v_{\pi(j)}$ if $v_{\pi(j)}'$ is unmatched and with some $y\in V\setminus C$ if $y$ is matched with $v_{\pi(j)}'$ that is constructed in the very same fashion.

Notice that all vertices of $I$ outside the matching belong to paths of length zero.
As the matching was maximum, we see the minimum number of paths (respecting the choice of $\bar b$)
constructed in this way. More precisely, we claim the following property.

\begin{claim}
Let $M$ is a maximum weight matching of $G'(\pi,\bar b)=(V',E')$ with weight function~$w$ as defined above. Then, the path partition $\mathbb{P}(\pi,\bar b)$ contains $|V|-w(M)$ many paths. 
\end{claim}

\begin{proof}
It is best to think of this claim from the perspective of how many paths can be `saved' within $\mathbb{P}(\pi,\bar b)$ in comparison to the solution that consists of single-vertex paths only. If a vertex in~$C\cup I$ is not matched, then it means that it forms a single-vertex path. 
If we find an edge $e=v_{\pi(i)}x\in M$ with $x\in I$, then this means that we see the sub-path $v_{\pi(i)}xv_{\pi(i+1)}$ in $\mathbb{P}(\pi,\bar b)$, somehow `merging' three potential single-vertex paths and therefore reducing the number of paths by~$2$, which is the weight of~$e$. If we find an edge $e=v_{\pi(i)}x\in M$ where $b_i=0$, i.e., $w(e)=1$, then the path does not continue to $v_{\pi(i+1)}$ but will end in~$x$. Nonetheless, $x$ does not form a single-vertex path, so again we reduce the number of paths by the weight~$1$ of~$e$. A similar observation can be made concerning edges between some $v_{\pi(i)}'$ and some $x\in I$, as they treat 'the other end' of the paths. Finally, one can consider the edges $e=v_{\pi(i)}u_i\in M$. Again, $w(e)=1$. They reflect the possibility to connect two paths that would otherwise end in  $v_{\pi(i)}$ and start in $v_{\pi(i+1)}$, respectively. Hence, also in this case we reduce the number of paths by the weight~$1$ of~$e$.\renewcommand{\qedsymbol}{$\Diamond$}
\end{proof}

If the path partition $\mathbb{P}(\pi,\bar b)$ contains less paths than the current record partition $\mathbb{P}$, we will update~$\mathbb{P}$ accordingly.
After enumerating all permutations and bit vectors, $\mathbb{P}$ should satisfy the following claim.

\begin{claim}
The given graph~$G$ contains a path partition with at most $k$ many paths if and only if our algorithm outputs a record partition $\mathbb{P}$ with at most~$k$ many paths.
\end{claim}

\begin{proof}
It should be clear by what we argued above that our algorithm only outputs valid path partitions. Hence, if it outputs a record partition $\mathbb{P}$ with at most~$k$ many paths, then our graph contains such a pp.

Conversely, let $\mathbb{P}$ be a pp with the minimum number of paths. Let $k$ be the number of these paths. If $P_1,\dots,P_k$ are the paths in $\mathbb{P}$ ordered by decreasing length, then there is some smallest number $k'\leq k$ such that the paths $P_{k'+1},\dots,P_k$ contain only vertices from~$I$. Then, by reading the vertices from $C=\{v_1,\dots,v_r \}$ on the paths $P_1,\dots,P_{k'}$ `from left to right', we can define a permutation $\pi$ on $\{1,\dots,r\}$.
Moreover, we can define a bit-vector $\bar b$ that describes how many vertices of $C$ belong to $P_1,\dots,P_{k'}$. This bit-vector has $k'-1$ many zeros.
One can verify that this solution corresponds to a weighted maximum matching in the auxiliary graph $G'(\pi,\bar b)$ of weight $|V|-k$.
\renewcommand{\qedsymbol}{$\Diamond$}
\end{proof}

Notice that Step~(0) is executed only once, so that for the running time, it is decisive that we cycle through all permutations and through all settings of the bit-vectors.\end{proof}

\noindent
We are now illustrating the construction of the previous proof by an example.

\begin{figure}
    \centering
\begin{subfigure}[b]{.98\textwidth}
    \centering
	\begin{tikzpicture}[transform shape]
	\tikzset{scale=1,every node/.style={scale=0.5, fill = black,circle,minimum size=0.1cm}}
			
			\node[draw] (v1) at (0.5,1) {}; 
			\node[draw] (v2) at (2.5,1) {};
			\node[draw] (v3) at (4,1) {}; 
			\node[draw] (v4) at (5.5,1) {};	
			\node[draw] (v5) at (7,1) {}; 
			\node[draw] (v6) at (8,1) {};	
			\node[draw] (v7) at (9,1) {}; 
			\node[draw] (w1) at (0,-1) {}; 
			\node[draw] (w2) at (1,-1) {};	
			\node[draw] (w3) at (2,-1) {}; 
			\node[draw] (w4) at (3,-1) {};	
			\node[draw] (w5) at (4,-1) {}; 
			\node[draw] (w6) at (5,-1) {};	
			\node[draw] (w7) at (6,-1) {};
			\node[draw] (w8) at (7,-1) {}; 
			\node[draw] (w9) at (8,-1) {};		
			\node[draw] (w10) at (9,-1) {};
			
			\node[scale=1.66,fill=none]  at (0.5,1.2) {$v_1$};
			\node[scale=1.66,fill=none]  at (2.5,1.2) {$v_2$};
			\node[scale=1.66,fill=none]  at (4,1.2) {$v_3$}; 
			\node[scale=1.66,fill=none]  at (5.5,1.2) {$v_4$};	
			\node[scale=1.66,fill=none]  at (7,1.2) {$v_5$}; 
			\node[scale=1.66,fill=none]  at (8,1.2) {$v_6$};		
			\node[scale=1.66,fill=none]  at (9,1.2) {$v_7$}; 
			\node[scale=1.66,fill=none]  at (0,-1.2) {$w_1$}; 
			\node[scale=1.66,fill=none]  at (1,-1.2) {$w_2$};	
			\node[scale=1.66,fill=none]  at (2,-1.2) {$w_3$}; 
			\node[scale=1.66,fill=none]  at (3,-1.2) {$w_4$};	
			\node[scale=1.66,fill=none]  at (4,-1.2) {$w_5$}; 
			\node[scale=1.66,fill=none]  at (5,-1.2) {$w_6$};
			\node[scale=1.66,fill=none]  at (6,-1.2) {$w_7$};
			\node[scale=1.66,fill=none]  at (7,-1.2) {$w_8$}; 
			\node[scale=1.66,fill=none]  at (8,-1.2) {$w_9$};
			\node[scale=1.66,fill=none]  at (9,-1.2) {$w_{10}$};
   \draw[left,dotted,->] (w3) arc (0:280:.21);
    \draw[left,dashed,->] (w4) arc (0:280:.21);
			\path (v1) edge [dash dot] (w1);
			\path (v1) edge [dash dot] (w2);
			\path (v2) edge [dashed] (w3);
			\path (v2) edge [dotted] (w4);
			\path (v2) edge [dash dot] (v3);
			\path (v3) edge [dash dot] (w5);
			\path (v4) edge [dash dot] (w6);
			\path (v4) edge [dash dot] (w7);
			\path (v4) edge [bend left] (v6);
			\path (v5) edge [dash dot] (w8);
			\path (v5) edge [dash dot] (w9);
			\path (v5) edge (w10);
			\path (v5) edge (v6);
			\path (v6) edge (w2);
			\path (v6) edge [dash dot] (w9);
			\path (v6) edge [dashed] (w10);
			\path (v6) edge [dotted] (v7);
			\path (v7) edge [dash dot] (w10);
        \end{tikzpicture}

    \subcaption{The original graph $G$}
    \label{fig:griginalGraph}
\end{subfigure}\\
\begin{subfigure}[b]{.98\textwidth}
    \centering
    	
	\begin{tikzpicture}[transform shape]

	\tikzset{scale=1, every node/.style={scale=0.5, fill = black,circle,minimum size=0.1cm}}
			
			\node[draw] (v1) at (0.5,1) {}; 
			\node[draw] (v2) at (2.5,1) {};
			\node[draw] (v3) at (4,1) {}; 
			\node[draw] (v4) at (5.5,1) {};	
			\node[draw] (v5) at (7,1) {}; 
			\node[draw] (v6) at (8,1) {};	
			\node[draw] (v7) at (9,1) {}; 
			\node[draw] (v'1) at (0.5,-3) {}; 
			\node[draw] (v'2) at (2.5,-3) {};
			\node[draw] (v'4) at (5.5,-3) {};	
			\node[draw] (v'5) at (7,-3) {}; 
		    \node[draw] (u2) at (3.25,3) {};
		    \node[draw] (u5) at (7.5,3) {}; 
		    \node[draw] (u6) at (8.5,3) {};  
			\node[draw] (w1) at (0,-1) {}; 
			\node[draw] (w2) at (1,-1) {};	
			\node[draw] (w3) at (2,-1) {}; 
			\node[draw] (w4) at (3,-1) {};	
			\node[draw] (w5) at (4,-1) {}; 
			\node[draw] (w6) at (5,-1) {};	
			\node[draw] (w7) at (6,-1) {};
			\node[draw] (w8) at (7,-1) {}; 
			\node[draw] (w9) at (8,-1) {};		
			\node[draw] (w10) at (9,-1) {};
			\node[scale=1.66,fill=none]  at (0.8,1) {$v_1$};
			\node[scale=1.66,fill=none]  at (0.8,-3) {$v'_1$};
			\node[scale=1.66,fill=none]  at (2.8,1) {$v_2$};
			\node[scale=1.66,fill=none]  at (2.8,-3) {$v'_2$};
			\node[scale=1.66,fill=none]  at (4.3,1) {$v_3$}; 
			\node[scale=1.66,fill=none]  at (5.8,1) {$v_4$}; 
			\node[scale=1.66,fill=none]  at (5.8,-3) {$v'_4$};	
			\node[scale=1.66,fill=none]  at (7.3,1) {$v_5$};	
			\node[scale=1.66,fill=none]  at (7.3,-3) {$v'_5$}; 
			\node[scale=1.66,fill=none]  at (8.3,1) {$v_6$};		
			\node[scale=1.66,fill=none]  at (9.3,1) {$v_7$}; 
			\node[scale=1.66,fill=none]  at (0.3,-1) {$w_1$}; 
			\node[scale=1.66,fill=none]  at (1.3,-1) {$w_2$};	
			\node[scale=1.66,fill=none]  at (2.3,-1) {$w_3$}; 
			\node[scale=1.66,fill=none]  at (3.3,-1) {$w_4$};	
			\node[scale=1.66,fill=none]  at (4.3,-1) {$w_5$}; 
			\node[scale=1.66,fill=none]  at (5.3,-1) {$w_6$};
			\node[scale=1.66,fill=none]  at (6.3,-1) {$w_7$};
			\node[scale=1.66,fill=none]  at (7.3,-1) {$w_8$}; 
			\node[scale=1.66,fill=none]  at (8.3,-1) {$w_9$};
			\node[scale=1.66,fill=none]  at (9.3,-1) {$w_{10}$};
		    \node[scale=1.66,fill=none]  at (3.55,3) {$u_2$};
		    \node[scale=1.66,fill=none]  at (7.8,3) {$u_5$}; 
		    \node[scale=1.66,fill=none]  at (8.8,3) {$u_6$}; 
			
			\path (v1) edge [left,dashed] node [fill=none,scale=1.66] {1} (w1);
			\path (v1) edge [right,dotted] node [fill=none,scale=1.66] {1} (w2);
			\path (v'1) edge [left,dotted] node [fill=none,scale=1.66] {1} (w1);
			\path (v'1) edge [right,dashed] node [fill=none,scale=1.66] {1} (w2);
			\path (v'2) edge [left,dash dot] node [fill=none,scale=1.66] {1} (w3);
			\path (v'2) edge [right] node [fill=none,scale=1.66] {1} (w4);
			\path (v2) edge [left,dash dot] node [fill=none,scale=1.66] {1} (u2);
			\path (v3) edge [left,dash dot] node [fill=none,scale=1.66] {1} (w5);
			\path (v4) edge [left,dashed] node [fill=none,scale=1.66] {1} (w6);
			\path (v4) edge [right, dotted] node [fill=none,scale=1.66] {1} (w7);
			\path (v'4) edge [left,dotted ] node [fill=none,scale=1.66] {1} (w6);
			\path (v'4) edge [right,dashed] node [fill=none,scale=1.66] {1} (w7);
			\path (v5) edge [left,dash dot] node [fill=none,scale=1.66] {2}(w9);
			\path (v5) edge [below] node [fill=none,scale=1.66] {2}(w10);
			\path (v'5) edge [left, dash dot] node [fill=none,scale=1.66] {1} (w8);
			\path (v'5) edge [above] node [fill=none,scale=1.66] {1} (w9);
			\path (v'5) edge [right] node [fill=none,scale=1.66] {1} (w10);
			\path (v5) edge [left] node [fill=none,scale=1.66] {1} (u5);
			\path (v6) edge [above,dashed] node [fill=none,scale=1.66] {2}(w10);
			\path (v6) edge [left, dotted] node [fill=none,scale=1.66] {1} (u6);
			\path (v7) edge [right, dotted] node [fill=none,scale=1.66] {1}(w10);

        \end{tikzpicture}
    \subcaption{The resulting bipartite graph $G'$}
    \label{fig:transformGraph}
\end{subfigure}
    \caption{Drawings illustrating \autoref{exampleE}}\label{fig:ExampleE}
\end{figure}

\begin{example}\label{exampleE}
Consider $G=(V,E)$, with $V$ being decomposed as $C\cup I$, where $C=\{v_1,v_2,\dots,v_7\}$ and $I=\{w_1,w_2,\dots,w_{10}\}$. The set $C$ is a vertex cover of~$G$, with $E$ containing the edges $\{v_2v_3,v_4v_6,v_6v_6,v_6v_7\}$ plus several edges connecting $C$ and~$I$, as shown in \autoref{fig:ExampleE}.

Let us consider the identity as permutation and consider the bit-vector $\bar b=(0,1,0,0,1,1)$.
This lets us consider the graph~$G'$ as depicted in \autoref{fig:ExampleE}.
Notice that we use four layers for drawing~$G'$.
On the topmost layer, we show the three vertices $u_2,u_5,u_6$ that we have to introduce thanks to $\bar b$. On the second layer, we draw the vertices of~$C$ given in the ordering of our permutation, the identity. On the third layer, we draw the vertices from~$I$. Finally, our algorithm enforces us to introduce one more primed vertex whenever a new path is started with respect to~$C$. As $\bar b$ contains three zeros, this means that we have to introduce four vertices, which are $v_1',v_2',v_4',v_5'$.
All edges in this bipartite graph $G'$ that are incident to vertices on the topmost or lowermost layer have weight~$1$. Only edges connecting the second and the third layer may have either weight~$1$ or~$2$.
In our example, one can find two distinctively different matchings of maximum weight:
\begin{itemize}
    \item $M_1=\{v_1w_1,v_1'w_2,v_2u_2,v_2'w_3,v_3w_5,v_4w_6,v_4'w_7,v_5'w_8,v_5w_9,v_6w_{10}\}$, and
    \item $M_2=\{v_1'w_1,v_1w_2,v_2u_2,v_2'w_3,v_3w_5,v_4'w_6,v_4w_7,v_5'w_8,v_5w_9,v_6u_6,v_7w_{10}\}$.
\end{itemize}
We find $w(M_1)=w(M_2)=12$ and hence, both matchings should correspond to a collection of five paths in the original graph~$G$, as $G$ has 17 vertices. More precisely, $M_1$ corresponds to $$P_1=\{w_1v_1w_2,w_3v_2v_3w_5,w_4,w_6v_4w_7,w_8v_5w_9v_6w_{10}v_7\}\,,$$ while $M_2$ shows $P_2=\{w_1v_1w_2,w_3,w_4v_2v_3w_5,w_6v_4w_7,w_8v_5w_9v_6v_7w_{10}\}$. In \autoref{fig:ExampleE}, we mark $M_1$ with dashed lines and $M_2$ with dotted lines; if an edge appears in both matchings, the line is dashed-dotted. We follow a similar convention for marking the corresponding paths $P_1$ and $P_2$. We mark single-vertex paths by small self-loops, see $w_4$ in $P_1$ and $w_3$ in $P_2$.

We could also consider a different bit-vector, say, $\bar b'=(0,0,0,0,1,1)$.
Although we would now have one more path involving $C$-vertices, their overall number would not change. As the reader can verify, a possible solution obtained by our matching approach is now
$$P_1'=\{w_1v_1w_2,w_3v_2w_4,w_5v_3,w_6v_4w_7,w_8v_5w_9v_6w_{10}v_7\}\,.$$ Again, there would be a variation $P_2'$ similar as before.

Observe that the edge $v_4v_6$ does not play a big role in our construction in this example. We can take this to an extreme by considering the permutation $\pi=(7,1,2,6,3,5,4)$. Then, there are no two vertices $v_{\pi(i)}$ and $v_{\pi(i+1)}$ that are connected by an edge. Even worse, there is no vertex from~$I$ adjacent both to $v_{\pi(i)}$ and $v_{\pi(i+1)}$.
Therefore, the only reasonable bit-vector is the all-zero-vector~$\vec{0}$. This means we have seven paths containing $C$-vertices in any path partition of $G'(\pi,\vec{0})$. But in fact, our matching approach also yields a matching with 10 edges in $G'(\pi,\vec{0})$, with each edge having weight~$1$. To realize this, we could take a path partition that contains four paths of length one: $v_7w_{10}$, $v_6w_9$, $v_3w_5$, and $v_5w_8$, plus three paths of length two: $w_1v_1w_2$, $w_3v_2w_4$, and $w_6v_4w_7$. 

But the first choice obviously gave a better result for our aim of minimizing the number of paths in a path partition of~$G$.
\end{example}

\section{Duals and Distance to Triviality}\label{sec:duals}

Given a typical graph problem that is (as a standard) parameterized by solution size~$k$, it takes as input a graph $G$ of order~$n$ and~$k$, then its \emph{dual parameter} is $k_d=n-k$. This applies in particular to our problems \textsc{UPP}, \textsc{UIPP} and \textsc{USPP}. As these problems always have as a trivial solution the number~$n$ of vertices (i.e., $n$ trivial paths), we can also interpret this dual parameterization as a parameterization led by the idea of \emph{distance from triviality}; see~\cite{GuoHufNie2004}.
Moreover, all our problems can be algorithmically solved for each connected component separately, so that we can assume, without loss of generality, that we are dealing with connected graphs. Namely, if $(G,k)$ with $G=(V,E)$ is a graph and $C\subsetneq V$ describes a connected component, then we can solve $(G[C],k')$ and $(G-C,k-k')$ independently for all $1\leq k'<k$.
We now prove that our problems, with dual parameterizations, are in \FPT by providing a kernelization algorithm. The following combinatorial claims are crucial.

\begin{lemma}  \label{obs:match}
Let $G$ be an undirected graph of order~$n$.
    If $G$ has a matching that covers $2k$ vertices, then $(G,n-k)$ is a \yes-instance of \textsc{UPP}, \textsc{UIPP} and \textsc{USPP}. 
\end{lemma}
\begin{proof}
Let $M$ be a matching of $G$ such that $M$ covers $2k$ vertices. We construct a solution $\mathbb{P}$ of \textsc{PP}, \textsc{IPP} and \textsc{SPP} as follows. For every edge $(u,v) \in M$, include the path $\{u,v\}$ into $\mathbb{P}$. Put the rest of the vertices  not yet covered by~$M$ as single-vertex paths into the solution~$\mathbb{P}$ . We see that the path collection~$\mathbb{P}$ covers every vertex exactly once in~$G$ and has cardinality $n-k$.      
\end{proof}

\noindent
The preceding lemma has the following interesting consequence.
\begin{corollary}\label{cor:degree}  
If $G=(V,E)$ is a 
graph that possesses some $X \subseteq V$ with  $|X| \geq 2k$ such that $\deg(v) \geq  2k$ for every $v \in X$, then $G$ has a matching of size~$k$ and hence $(G,n-k)$ is a \yes-instance of \textsc{UPP}, \textsc{UIPP} and \textsc{USPP}. 
\end{corollary}

This consequence, as well as the following combinatorial observation, has no direct bearing on our algorithmic result, but may be of independent interest.

\begin{lemma} \label{obs:diam-bound}
    If $G$ is a connected graph with $\diam(G) > k$, then $(G,n-k)$ is a \yes-instance of \textsc{UPP}, \textsc{UIPP} and \textsc{USPP}.   
\end{lemma}

More interestingly, \autoref{obs:match} is also valid for directed graphs in the following form as we can observe the directions of the matching edges as prescribed by the given directed graph.
\begin{lemma}  \label{obs:match-dir}
Let $G$ be a directed graph of order~$n$. 
If the underlying undirected graph $U(G)$ has a matching that covers $2k$ vertices, then $(G,n-k)$ is a \yes-instance of \textsc{DPP}, \textsc{DIPP} and \textsc{DSPP}. 
\end{lemma}

Our combinatorial thoughts, along with \autoref{cor:vc-param} and \autoref{thm:vc-param}, allow us to show the following algorithmic result, the main result of this section.

\begin{theorem} \label{thm:dual-param} \textsc{UPP}, \textsc{UIPP} and \textsc{USPP}, as well as \textsc{DPP}, \textsc{DIPP} and \textsc{DSPP}, can be solved in \FPT time with respect to the dual parameterization.
\end{theorem}
\begin{proof}
We first handle the undirected graph case. Let $G=(V,E)$ and $k$ together be an instance of \textsc{UPP}, \textsc{UIPP}, or \textsc{USPP}, where $|V|=n$.
As said above, we can assume that $G$ is connected. 
Next, compute (in polynomial time) a maximum matching~$M$ in~$G$. If $|M|>k$, then \autoref{obs:match} justifies why $(G,n-k)$ is a \yes-instance. Hence, we can assume that any matching of $G$ has size at most~$k$. 
This implies (as matchings provide factor-2 approximations for \textsc{Minimum Vertex Cover}) that the vertex cover number $\vc(G)$ is at most $2k$. This allows us to first compute $\vc(G)$ in \FPT time; see~\cite{CheKanXia2010}.
Now, \autoref{cor:vc-param} and \autoref{thm:vc-param} show the claim in the case of \textsc{UIPP} and \textsc{USPP}, as well as for \textsc{UPP}.

For the case of directed graphs, we can first argue concerning the underlying undirected graph as in \autoref{obs:match-dir}. Then, we also make use of the part of \autoref{thm:vc-param} for directed graphs.
\end{proof}

\section{Covering and Edge-Disjoint Variants}\label{sec:ED}

As we mentioned in the introduction, the variants \textsc{PC}, \textsc{SPC} and \textsc{IPC} of the problems, where the paths cover the vertices but do not need to form a partition, are also studied. Many of our results do apply to these problems as well.

Alternatively, in some applications, paths may be required to be edge-disjoint but not necessarily vertex-disjoint, see for example~\cite{Slivkins2010}, where the problem \textsc{Edge-Disjoint Paths} is studied. Thus, one may define the problems \textsc{Edge-Disjoint Path Cover} (\textsc{ED-PC}), \textsc{Edge-Disjoint Shortest Path Cover} (\textsc{ED-SPC}) and \textsc{Edge-Disjoint Induced Path Cover} (\textsc{ED-IPC}) (as before we include the prefixes \textsc{U}, \textsc{D}, or \textsc{DAG} in the problem name abbreviations if we want to emphasize that the inputs are undirected, directed, or directed acyclic). These can be thought of as intermediate problems between the partition and the covering variants, indeed a (unrestricted, induced or shortest) path partition is an (unrestricted, induced or shortest) edge-disjoint path cover, which is also a (unrestricted, induced or shortest) path cover.

Some of our results apply to these three problems as well with slight modifications of the proofs. Hence, we only exhibit these modifications in the following.

\subsection{NP-completeness results for DAGs and bipartite undirected graphs}

\begin{theorem}
    \textsc{DAGSPC}, \textsc{DAGIPC}, \textsc{ED-DAGSPC} and \textsc{ED-DAGIPC} are \NP-hard even when the inputs are restricted to planar bipartite DAGs of maximum degree~4. 
\end{theorem}

\begin{proof}
We follow the same construction as in \autoref{thm:DAGSPP-NPhard}.  Observe that, by the same argument as in the proof of \autoref{claim:DAG-reduction-mainclaim}, each shortest or induced path in any solution of \textsc{(ED-)DAGSPC/(ED-)DAGIPC} must contain exactly three vertices (and thus forms a vertex-partition of the graph), and each gadget $H(v_i)$ is partitioned into $P_3$-paths in one of the two ways as shown in \autoref{twopart}. The rest of the argumentation is similar to \autoref{thm:DAGSPP-NPhard}. 
\end{proof} 

\begin{theorem}
    \textsc{USPC} and \textsc{ED-USPC} are \NP-hard, even for bipartite 5-degenerate graphs of diameter~4.
\end{theorem}
\begin{proof}
We follow the same construction as \autoref{thm:SPP-NP-hardness} (also see \autoref{bipartitefig}). We can observe that \autoref{obs:diameter}, \autoref{obs:two_paths_in_partition} and \autoref{obs:bipartite_degenerate} are true due to the construction. The rest of the proof is similar to \autoref{thm:SPP-NP-hardness}.


\end{proof}

\subsection{W-hardness for DAGs}

Our W-hardness result is only adaptable (as far as we can see) for the edge-disjoint path cover variants: for the unrestricted path cover variants, crucial elements of the proof do not work, as the same edges can be used by multiple solution paths. Also, recall that \textsc{DAGPC} is solvable in polynomial time~\cite{fulkerson1956note}.

The only problem with the construction of \autoref{thm:w-hard-dag} is that now paths can start from $s_{i,j}$, with $i,j \in [k]$, and go through to $a_{l}$,  $i \leq l \leq k $, even if $a_{l}$ belongs to a different path. To prevent this, we modify $G^{i,u}$, where $0< u \leq n$ and $i \in [k]$, in the construction as shown in \autoref{fig:UPD-I-SPP-W-gadget}. The rest of the proof is similar to \autoref{thm:w-hard-dag}. This construction is an adaptations of \cite{Slivkins2010}, where it was proved that \textsc{ED-DP} is \W{1}-hard when parameterized by solution size.

\begin{theorem}
     \textsc{ED-DAGSPC} and \textsc{ED-DAGIPC} (parameterized by solution size) are \W{1}-hard.
\end{theorem} 

\begin{figure}
\centering

  \tikzset{every picture/.style={line width=0.75pt}} 

\begin{tikzpicture}[
roundnode/.style={circle, draw=black,  thick, minimum size=5mm, scale=0.8},
squarednode/.style={rectangle, draw=red!60, fill=red!5, very thick, minimum size=5mm},
dot/.style={fill=black,circle,minimum size=1pt}
]

\def\x{0.6}
\def\y{-0.1}

\node[rectangle,draw,minimum width = 6cm, 
    minimum height = 4cm, rounded corners=5pt] (b) at (\x,\y) {};
    \node[roundnode] at (\x - 1.5cm ,\y + 0.5cm )     (a1){}  ;
    \node[roundnode] at (\x - 0.5 cm ,\y + 0.5cm )     (a2){}  ;
    \node[roundnode] at (\x - 1.5cm ,\y + 1.5cm )     (a11){}  ;
    \node[roundnode] at (\x - 0.5 cm ,\y + 1.5cm )     (a12){}  ;
    \node[roundnode] at (\x + 0.5cm ,\y + 1.5cm )     (a13){}  ;
    \node[roundnode] at (\x + 1.5cm ,\y + 1.5cm )     (a14){}  ;
     \node[roundnode] at (\x + 2.7cm ,\y + 1.5cm )     (a15){}  ;
    \node[roundnode] at (\x - 1.5cm ,\y - 1.5cm )     (b11){}  ;
    \node[roundnode] at (\x - 0.5 cm ,\y - 1.5cm )     (b12){}  ;
     \node[roundnode] at (\x + 0.5cm ,\y - 1.5cm )     (b13){}  ;
    \node[roundnode] at (\x + 1.5 cm ,\y - 1.5cm )     (b14){}  ;
    \node[roundnode] at (\x + 2.7 cm ,\y - 1.5cm )     (b15){}  ;
    \node[dot,scale=0.3] at (\x - 0.1cm ,\y + 0.5cm) (d1){};
    \node[dot,scale=0.3] at (\x - 0.0cm ,\y + 0.5cm) (d1){};
    \node[dot,scale=0.3] at (\x + 0.1cm ,\y + 0.5cm) (d1){};
 \node[roundnode] at (\x + 0.5cm ,\y + 0.5cm )     (ai1){}  ;
    \node[roundnode] at (\x + 1.5cm ,\y + 0.5cm )     (ai){}  ;
     \node[dot,scale=0.3] at (\x + 2cm ,\y + 0.5cm) (d1){};
    \node[dot,scale=0.3] at (\x + 2.1cm ,\y + 0.5cm) (d1){};
    \node[dot,scale=0.3] at (\x + 2.2cm ,\y + 0.5cm) (d1){};
    \node[roundnode] at (\x + 2.7cm ,\y + 0.5cm )     (ak){}  ;
    \node[roundnode] at (\x + 1.5cm ,\y + 2.5cm )     (si){}  ;
    \node[roundnode] at (\x + 2.7cm ,\y + 2.5cm )     (sk){}  ;
 \node[roundnode] at (\x - 1.5cm ,\y - 0.5cm )     (b1){}  ;
    \node[roundnode] at (\x - 1.5cm ,\y - 2.5cm )     (t1){}  ;
    \node[roundnode] at (\x - 0.5cm ,\y - 0.5cm  )     (b2){}  ;
    \node[roundnode] at (\x - 0.5cm ,\y - 2.5cm )     (t2){}  ;
    \node[dot,scale=0.3] at (\x - 0.1cm ,\y - 0.5cm) (d1){};
    \node[dot,scale=0.3] at (\x + 0cm ,\y - 0.5cm) (d1){};
    \node[dot,scale=0.3] at (\x + 0.1cm ,\y - 0.5cm) (d1){};
      \node[roundnode] at (\x + 0.5cm ,\y - 0.5cm )     (bi1){}  ;
      \node[roundnode] at (\x + 0.5cm ,\y - 2.5cm )     (ti1){}  ;
    \node[roundnode] at (\x + 1.5cm ,\y - 0.5cm )     (bi){}  ;
     \node[dot,scale=0.3] at (\x + 2.cm ,\y-0.5cm) (d1){};
    \node[dot,scale=0.3] at (\x + 2.1cm ,\y -0.5cm) (d1){};
    \node[dot,scale=0.3] at (\x + 2.2cm ,\y -0.5cm) (d1){};
    \node[roundnode] at (\x + 2.7cm ,\y - 0.5cm )     (bk){} ; 
    \node[] at (\x - 1.0cm , \y + 0.2cm) {$a_1$};
    \node[] at (\x - 1.1cm , \y + 1.2cm) {$a'_1$};
     \node[] at (\x - .15cm , \y + 0.2cm) {$a_2$};
     \node[] at (\x - .15cm , \y + 1.2cm) {$a'_2$};
     \node[] at (\x + 1cm , \y + 0.2cm) {$a_{i-1}$};
     \node[] at (\x+0.9cm  , \y + 2.5cm) {$s_{i,i+1}'$};
     \node[] at (\x + 2.05cm , \y + 0.2cm) {$a_{i+1}$};
     \node[] at (\x + 2.05cm , \y + 1.2cm) {$a'_{i+1}$};
     \node[] at (\x + 3cm , \y + 0.2cm) {$a_{k}$};
     \node[] at (\x + 3cm , \y + 1.2cm) {$a'_{k}$};    
     \node[] at (\x+3.2cm  , \y + 2.5cm) {$s_{i,k}'$};
     \node[] at (\x - 1.1cm , \y - 0.8cm) {$b_1$};
     \node[] at (\x - 1.1cm , \y - 1.8cm) {$b'_1$};
     \node[] at (\x - 2cm , \y - 2.5cm) {$t_{1,i}'$};
     \node[] at (\x - .15cm , \y - 0.8cm) {$b_2$};
     \node[] at (\x - .15cm , \y - 1.8cm) {$b'_2$};
     \node[] at (\x -1cm , \y - 2.5cm) {$t_{2,i}'$};
     \node[] at (\x + 1.0cm , \y - 0.5cm) {$b_{i-1}$};
     \node[] at (\x+1.2cm  , \y - 2.5cm) {$t_{i-1,i}'$};
     \node[] at (\x + 2.05cm , \y - 0.8cm) {$b_{i+1}$};
     \node[] at (\x + 2.05cm , \y - 1.8cm) {$b'_{i+1}$};
     \node[] at (\x + 3cm , \y - 0.8cm) {$b_{k}$};
     \node[] at (\x + 3cm , \y - 1.8cm) {$b'_{k}$};
     
    \draw[-{Stealth[scale=1]}] (a1) to (b1);
    \draw[-{Stealth[scale=1]}] (a2) to (b2);
    \draw[-{Stealth[scale=1]}] (ai1) to (a14);
    \draw[-{Stealth[scale=1]}] (ai) to (bi);
    \draw[-{Stealth[scale=1]}] (ai1) to (bi1);
    \draw[-{Stealth[scale=1]}] (b13) to (bi);
    \draw[-{Stealth[scale=1]}] (a11) to (a1);
    \draw[-{Stealth[scale=1]}] (a1) to (a12);
    \draw[-{Stealth[scale=1]}] (a12) to (a2);
    \draw[-{Stealth[scale=1]}] (a13) to (ai1);
    \draw[-{Stealth[scale=1]}] (a14) to (ai);
    \draw[-{Stealth[scale=1]}] (a15) to (ak);
    \draw[-{Stealth[scale=1]}] (si) to (a14);
    \draw[-{Stealth[scale=1]}] (sk) to (a15);
    \draw[-{Stealth[scale=1]}] (b11) to (t1);
    \draw[-{Stealth[scale=1]}] (b12) to (t2);
    \draw[-{Stealth[scale=1]}] (b13) to (ti1);
    \draw[-{Stealth[scale=1]}] (b1) to (b11);
    \draw[-{Stealth[scale=1]}] (b11) to (b2);
    \draw[-{Stealth[scale=1]}] (b2) to (b12);
    \draw[-{Stealth[scale=1]}] (bi1) to (b13);
    \draw[-{Stealth[scale=1]}] (bi) to (b14);
    \draw[-{Stealth[scale=1]}] (bk) to (b15);
    \draw[-{Stealth[scale=1]}] (ak) to (bk);
    \draw[-{Stealth[scale=1]}] (\x - 2.5cm, \y + 1.5cm) to (a11);
    \draw[-{Stealth[scale=1]}] (\x - 2.5cm, \y - 0.5cm) to (b1);
    \draw[-{Stealth[scale=1]}] (ak) to (\x + 3.8cm, \y + 1.5cm);
    \draw[-{Stealth[scale=1]}] (b15) to (\x + 3.8cm, \y - 0.5cm);

 \end{tikzpicture}
  \captionof{figure}{New gadget $G^{i,u}$, $u\in [n]$, $i \in [k]$.}
  \label{fig:UPD-I-SPP-W-gadget}
  \end{figure}
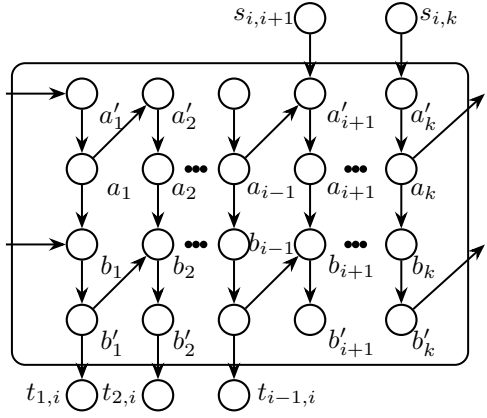

\subsection{\XP algorithms}

For our next result, we need the following theorem, indeed the techniques used in Section~\ref{xp-algo} are not applicable since they are based on using (vertex-) disjoint path problems.

\begin{theorem} [\cite{dumas2024graphs}] \label{k-short}
Let $G$ be an undirected graph whose vertex set can be covered by at most $k$ (not necessarily disjoint) shortest paths. Then the pathwidth of $G$ is in $O(k\cdot 3^{k})$.
\end{theorem}

The following corollary was already observed in~\cite{dumas2024graphs} for the problems \textsc{USPC} and \textsc{USPP}. The same argument shows the following.

\begin{corollary} \label{cor-sol}
    \textsc{ED-USPC} is in \XP when parameterized by solution size. 
\end{corollary} 
\begin{proof}
Here we cannot use the technique of \autoref{lem:SPPinXP} based on \textsc{DP} and its variants. \autoref{cor-sol}, follows from the same proof used in~\cite{dumas2024graphs} for \textsc{USPC}: first bound the pathwidth of the graph using \autoref{k-short}. Then, as for \autoref{lem:SPPinXP}, guess the set $X$ of endpoints of the solution paths (enumerating all the solutions is done in \XP time). That way, one can compute the sum $s$ of the pairwise distances between the endpoints in $X$. Then, one can use the counting version of Courcelle's theorem by expressing the problem in Counting Monadic Second Order Logic (CMSOL) and solve it in \FPT time: one expects a partition whose total size is equal to $s$.
\end{proof}

\subsection{Neighborhood diversity and vertex cover parameterizations}

\begin{theorem}\label{ed-nd}
    \textsc{USPC}, \textsc{UIPC}, \textsc{ED-USPC} and \textsc{ED-UIPC} (as well as \textsc{DSPC}, \textsc{DIPC}, \textsc{ED-DSPC} and \textsc{ED-DIPC}) are in \FPT when parameterized by (directed) neighborhood diversity. 
\end{theorem}
\begin{proof}
This can be obtained with the same proof as the one of \autoref{thm:USPP-ND} by modifying the \textsc{ILP} as follows. Each path $P \in \cal P$ is characterized by a vector $(t_{P,1,1}, \ldots , t_{P,k,k})$ with $t_{P,i,j} = |E_{i,j} \cap P|$, where $i , j \leq k$. Here $E_{i,j}$ refers to the edges between the $i^{th}$ and $j^{th}$ equivalence class.

For the covering versions, we use:

$$\text{minimize} 
     \sum_{P\in\cal P}z_P
     $$     $$\text{subject to }\sum_{P\in \cal P}z_P\cdot P^j \geq |C_j|\text{  for all }j\in [d]\,,$$

while for the edge-disjoint variants, we write:

           $$\text{minimize} 
     \sum_{P\in\cal P}z_P
     $$     $$\text{subject to }\sum_{P\in \cal P}z_P\cdot P^j \geq|C_j|\text{  for all }j\in [d]$$
     $$ \text{and subject to }\sum_{P\in \cal P}z_P \cdot t_{P,i,j} \leq|E_{i,j}|\text{ for all }i,j\in [d]\,. $$
     For the directed cases, we refer to the arguments leading to \autoref{thm:dspp+dipp-dnd}.   
    \end{proof}

\begin{corollary} \label{ed-vc}
    \textsc{USPC}, \textsc{UIPC}, \textsc{ED-USPC} and \textsc{ED-UIPC} (as well as \textsc{DSPC}, \textsc{DIPC}, \textsc{ED-DSPC} and \textsc{ED-DIPC}) are in \FPT when parameterized by the vertex cover number of the given (or of the underlying undirected) graph.
\end{corollary}

\subsection{Dual parameters}


\begin{theorem}
\textsc{UPC}, \textsc{ED-UPC}, \textsc{ED-USPC} and \textsc{ED-UIPC} (as well as the directed counterparts \textsc{ED-DSPC}, \textsc{ED-DIPC}, \textsc{ED-DSPC} and \textsc{ED-DIPC}) can be solved in \FPT time with respect to the dual parameterization. 
\end{theorem}
\begin{proof}
These results are obtained using the technique of \autoref{thm:dual-param} in view of \autoref{ed-vc}, using \autoref{thm:vc-param} for \textsc{(ED-)UPC} and \autoref{ed-nd} for \textsc{(ED-)USPC} and \textsc{(ED-)UIPC}.
\end{proof}

\section{Conclusions}\label{conclusion}

We have explored the algorithmic complexity of the three problems \textsc{PP}, \textsc{IPP} and \textsc{SPP}, and as witnessed by \autoref{tab:survey-PPP}, our results show some interesting algorithmic differences between these three problems.

Many interesting questions remain to be investigated. For example, what is the parameterized complexity of \textsc{SPP} on undirected graphs, parameterized by the number of paths? Is it \W{1}-hard, like for DAGs? This was asked in~\cite{dumas2024graphs}, and our \W{1}-hardness result for DAGs can be seen as a first step towards an answer.

Another interesting question is whether \textsc{SPP} is \XP on directed graphs? This is the case for undirected graphs, for DAGs and for directed graphs that are planar or have bounded directed treewidth, but we do not know about the general case. In the conference version of this paper, we had asked whether \textsc{IPP} is in \XP for DAGs when parameterized by solution size, or if it is \paraNP-hard like in the general (and undirected) case. During the revision of this journal paper, the latter possibility has been confirmed in~\cite{LafMou2024}. We have added this reference also in Table~\ref{tab:survey-PPP} because it highlights that in terms of complexity, there are indeed many parallels between undirected graphs and DAGs concerning SPP/IPP problems.

We do not know whether \textsc{IPP} is \NP-hard on undirected bipartite graphs. Other classes, like planar graphs, could also be interesting to study in this regard.

We have seen that \textsc{IPP} and \textsc{SPP} admit \FPT algorithms when parameterized by (directed) neighborhood diversity. Can we obtain such algorithms for other (e.g., more general) parameters? 
This holds on undirected graphs (\textsc{UPP}) for the more general parameter modular-width~\cite{GajLamOrd2013}; we do not know if this stays true for directed graphs.

Regarding the parameterized algorithms that we obtained for vertex cover and neighborhood diversity, can they be improved? For \textsc{PP} parameterized by vertex cover, we managed to reduce the dependency to $\Oh^*\left(\vc!\cdot 2^{\vc}\right)$, which is $\Oh^*\left(2^{\vc\log\vc}\right)$ (\autoref{prop:UPP-VC-FPT}). It is not clear if this can be improved to, say, $\Oh^*\left(2^{\vc}\right)$. Also, it would be interesting to obtain a similar running time (or better) for \textsc{SPP} and \textsc{IPP} as well. Furthermore, we do not know whether polynomial kernelizations exist for these problems.

Similarly, in the light of our \FPT algorithms for the dual parameterizations of \textsc{PP}, \textsc{IPP} and \textsc{SPP}, we can ask whether they admit a polynomial kernel.

As we have indicated, some of our results also apply to the covering variants of the three problems (where the paths need not be disjoint). On the other hand, approximation algorithms have been designed for \textsc {SPC}, the covering version of \textsc{SPP}, in~\cite{TG21} (for general undirected graphs) and~\cite{foucaud2022} (for chordal graphs and related undirected graph classes). Are there such algorithms for \textsc{SPP} (and \textsc{IPP}) as well?

Another line of research is \emph{reconfiguration}.
In particular, shortest path reconfiguration in undirected graphs has been extensively studied, see \cite{GajJKL2022} for a recent piece of work that also gives a nice survey, but directed graphs have attracted less attention. Also, the reconfiguration of path partitions as studied in this paper has not been touched so far, although this has natural applications in reconfiguring supply networks.
A possible definition of a single reconfiguration step could be to select two paths $P,Q$ from the current path partition $\mathbb{P}$ and build two other paths $P',Q'$ partitioning the graph induced by the vertex sets of $P,Q$, so that one obtains another path partition $\mathbb{P}'$ by replacing $P,Q$ in $\mathbb{P}$ by $P',Q'$.

We also note that some of the techniques used in this paper for the \textsc{IPP} problem likely can be applied to the \textsc{Hole Partition} problem (where we want to partition a graph into induced cycles). The packing version of this problem was studied in~\cite{marx20}.

\textsc{SPP} could also be studied for edge-weighted graphs, as studied for example in~\cite{DBLP:conf/esa/Berczi017}.

\section*{Acknowledgments}
This paper is an extended version of a paper presented at CIAC 2023~\cite{FernauFMPN23}. The research of the second author was financed by the French government IDEX-ISITE initiative 16-IDEX-0001 (CAP 20-25), the International Research Center ``Innovation Transportation and Production Systems'' of the I-SITE CAP 20-25, and by the ANR project GRALMECO (ANR-21-CE48-0004).
The research of the last author was funded by a DAAD-WISE Scholarship 2022.

\bibliography{references}

\begin{thebibliography}{10}

\bibitem{AF84}
M.~Aigner and M.~Fromme.
\newblock A game of cops and robbers.
\newblock {\em Discrete Applied Mathematics}, 8(1):1--12, 1984.

\bibitem{DBLP:journals/dam/AndreattaM95}
G.~Andreatta and F.~Mason.
\newblock Path covering problems and testing of printed circuits.
\newblock {\em Discrete Applied Mathematics}, 62(1-3):5--13, 1995.

\bibitem{MD-PP}
J.~Ara{\'{u}}jo, J.~Bensmail, V.~Campos, F.~Havet, A.~K. Maia~de Oliviera,
  N.~Nisse, and A.~Silva.
\newblock On finding the best and worst orientations for the metric dimension.
\newblock {\em Algorithmica}, 85(10):2962–3002, 2023.

\bibitem{DBLP:journals/algorithmica/AraujoCMSS20}
J.~Ara{\'{u}}jo, V.~A. Campos, A.~K. Maia, I.~Sau, and A.~Silva.
\newblock On the complexity of finding internally vertex-disjoint long directed
  paths.
\newblock {\em Algorithmica}, 82(6):1616--1639, 2020.

\bibitem{DBLP:journals/jgt/Bang-JensenCM17}
J{\o}rgen Bang{-}Jensen, Tilde~My Christiansen, and Alessandro Maddaloni.
\newblock Disjoint paths in decomposable digraphs.
\newblock {\em Journal of Graph Theory}, 85(2):545--567, 2017.

\bibitem{DBLP:journals/algorithmica/BelmonteHKKKKLO22}
R.~Belmonte, T.~Hanaka, M.~Kanzaki, M.~Kiyomi, Y.~Kobayashi, Y.~Kobayashi,
  M.~Lampis, H.~Ono, and Y.~Otachi.
\newblock Parameterized complexity of {(A,{\(\ell\)})}-path packing.
\newblock {\em Algorithmica}, 84(4):871--895, 2022.

\bibitem{DBLP:journals/siamdm/BentertNRZ23}
M.~Bentert, A.~Nichterlein, M.~Renken, and P.~Zschoche.
\newblock Using a geometric lens to find {$k$}-disjoint shortest paths.
\newblock {\em SIAM Journal on Discrete Mathematics}, 37(3):1674--1703, 2023.

\bibitem{DBLP:conf/esa/Berczi017}
K.~B{\'{e}}rczi and Y.~Kobayashi.
\newblock The directed disjoint shortest paths problem.
\newblock In K.~Pruhs and C.~Sohler, editors, {\em 25th Annual European
  Symposium on Algorithms, {ESA}}, volume~87 of {\em LIPIcs}, pages
  13:1--13:13. Schloss Dagstuhl - Leibniz-Zentrum f{\"{u}}r Informatik, 2017.

\bibitem{berge1983path}
C.~Berge.
\newblock Path partitions in directed graphs.
\newblock In {\em North-Holland Mathematics Studies}, volume~75, pages 59--63.
  Elsevier, 1983.

\bibitem{boesch1974covering}
F.~T. Boesch, S.~Chen, and J.~A.~M. McHugh.
\newblock On covering the points of a graph with point disjoint paths.
\newblock In {\em Graphs and Combinatorics}, pages 201--212. Springer, 1974.

\bibitem{boesch1977covering}
F.~T. Boesch and J.~F. Gimpel.
\newblock Covering points of a digraph with point-disjoint paths and its
  application to code optimization.
\newblock {\em Journal of the ACM}, 24(2):192--198, 1977.

\bibitem{dagPC}
M.~C{\'{a}}ceres, M.~Cairo, B.~Mumey, R.~Rizzi, and A.~I. Tomescu.
\newblock Sparsifying, shrinking and splicing for minimum path cover in
  parameterized linear time.
\newblock In {\em {ACM-SIAM} Symposium on Discrete Algorithms, {SODA}}, pages
  359--376. {SIAM}, 2022.

\bibitem{foucaud2022}
D.~Chakraborty, A.~Dailly, S.~Das, F.~Foucaud, H.~Gahlawat, and S.~K. Ghosh.
\newblock Complexity and algorithms for isometric path cover on chordal graphs
  and beyond.
\newblock In {\em Proceedings of the 33rd International Symposium on Algorithms
  and Computation, {ISAAC}}, volume 248 of {\em LIPIcs}, pages 12:1--12:17.
  Schloss Dagstuhl - Leibniz-Zentrum f{\"{u}}r Informatik, 2022.

\bibitem{DBLP:conf/mfcs/ChakrabortyMOPR24}
D.~Chakraborty, H.~M{\"{u}}ller, S.~Ordyniak, F.~Panolan, and M.~Rychlicki.
\newblock Covering and partitioning of split, chain and cographs with isometric
  paths.
\newblock In R.~Kr{\'{a}}lovic and A.~Kucera, editors, {\em 49th International
  Symposium on Mathematical Foundations of Computer Science, {MFCS}}, volume
  306 of {\em LIPIcs}, pages 39:1--39:14. Schloss Dagstuhl - Leibniz-Zentrum
  f{\"{u}}r Informatik, 2024.

\bibitem{DBLP:journals/siamdm/ChangK96}
G.~J. Chang and D.~Kuo.
\newblock The {L}(2,1)-labeling problem on graphs.
\newblock {\em SIAM Journal on Discrete Mathematics}, 9(2):309--316, 1996.

\bibitem{CheKanXia2010}
J.~Chen, I.~A. Kanj, and G.~Xia.
\newblock Improved upper bounds for vertex cover.
\newblock {\em Theoretical Computer Science}, 411(40--42):3736--3756, 2010.

\bibitem{DBLP:journals/iandc/ChenCKLXZ24}
Y.~Chen, Z.{-}Z. Chen, C.~Kennedy, G.~Lin, Y.~Xu, and A.~Zhang.
\newblock Approximating the directed path partition problem.
\newblock {\em Information and Computation}, 297:105150, 2024.

\bibitem{CLRS3}
T.~H. Cormen, C.~E. Leiserson, R.~L. Rivest, and C.~Stein.
\newblock {\em Introduction to Algorithms, Third Edition}.
\newblock The MIT Press, 3rd edition, 2009.

\bibitem{corneil2013ldfs}
D.~G. Corneil, B.~Dalton, and M.~Habib.
\newblock {LDFS}-based certifying algorithm for the minimum path cover problem
  on cocomparability graphs.
\newblock {\em SIAM Journal on Computing}, 42(3):792--807, 2013.

\bibitem{CygFKLMPPS2015}
M.~Cygan, F.~Fomin, \L. Kowalik, D.~Lokshtanov, D.~Marx, M.~Pilipczuk,
  M.~Pilipczuk, and S.~Saurabh.
\newblock {\em Parameterized Algorithms}.
\newblock Springer, 2015.

\bibitem{DBLP:conf/focs/CyganMPP13}
M.~Cygan, D.~Marx, M.~Pilipczuk, and M.~Pilipczuk.
\newblock The planar directed $k$-vertex-disjoint paths problem is
  fixed-parameter tractable.
\newblock In {\em 54th Annual {IEEE} Symposium on Foundations of Computer
  Science, {FOCS}}, pages 197--206. {IEEE} Computer Society, 2013.

\bibitem{DamDKS92}
P.~Damaschke, J.~S. Deogun, D.~Kratsch, and G.~Steiner.
\newblock Finding {H}amiltonian paths in cocomparability graphs using the bump
  number algorithm.
\newblock {\em Order}, 8(4):383--391, 1992.

\bibitem{D17}
R.~Diestel.
\newblock {\em Graph Theory, 5th Edition}, volume 173 of {\em Graduate Texts in
  Mathematics}.
\newblock Springer, 2017.

\bibitem{dilworth2009decomposition}
R.~P. Dilworth.
\newblock A decomposition theorem for partially ordered sets.
\newblock In {\em Classic Papers in Combinatorics}, pages 139--144. Springer,
  2009.

\bibitem{dujmovic2020planar}
V.~Dujmovi{\'c}, G.~Joret, P.~Micek, P.~Morin, T.~Ueckerdt, and D.~R. Wood.
\newblock Planar graphs have bounded queue-number.
\newblock {\em Journal of the ACM}, 67(4):1--38, 2020.

\bibitem{dumas2024graphs}
M.~Dumas, F.~Foucaud, A.~Perez, and I.~Todinca.
\newblock On graphs coverable by $k$ shortest paths.
\newblock {\em SIAM Journal on Discrete Mathematics}, 38(2):1840--1862, 2024.

\bibitem{dyer1986planar}
M.~E. Dyer and A.~M. Frieze.
\newblock Planar {3DM} is {NP}-complete.
\newblock {\em Journal of Algorithms}, 7(2):174--184, 1986.

\bibitem{DBLP:conf/iwoca/EtoKLMO24}
H.~Eto, S.~Kawaharada, G.~Lin, E.~Miyano, and T.~Ozdemir.
\newblock Directed path partition problem on directed acyclic graphs.
\newblock In A.~A. Rescigno and U.~Vaccaro, editors, {\em Combinatorial
  Algorithms - 35th International Workshop, {IWOCA}}, volume 14764 of {\em
  LNCS}, pages 314--326. Springer, 2024.

\bibitem{FernauFMPN23}
H.~Fernau, F.~Foucaud, K.~Mann, U.~Padariya, and R.~Rao~K. N.
\newblock Parameterizing path partitions.
\newblock In M.~Mavronicolas, editor, {\em Algorithms and Complexity - 13th
  International Conference, {CIAC}}, volume 13898 of {\em LNCS}, pages
  187--201. Springer, 2023.

\bibitem{FF01}
D.~C. Fisher and S.~L. Fitzpatrick.
\newblock The isometric number of a graph.
\newblock {\em Journal of Combinatorial Mathematics and Combinatorial
  Computing}, 38(1):97--110, 2001.

\bibitem{ForHopWyl80}
S.~Fortune, J.~E. Hopcroft, and J.~Wyllie.
\newblock The directed subgraph homeomorphism problem.
\newblock {\em Theoretical Computer Science}, 10:111--121, 1980.

\bibitem{franzblau_raychaudhuri_2002}
D.~S. Franzblau and A.~Raychaudhuri.
\newblock Optimal {H}amiltonian completions and path covers for trees, and a
  reduction to maximum flow.
\newblock {\em The ANZIAM Journal}, 44(2):193–204, 2002.

\bibitem{fulkerson1956note}
D.~R. Fulkerson.
\newblock Note on {D}ilworth’s decomposition theorem for partially ordered
  sets.
\newblock {\em Proceedings of the American Mathemathical Society},
  7(4):701--702, 1956.

\bibitem{GajLamOrd2013}
J.~Gajarsk{\'{y}}, M.~Lampis, and S.~Ordyniak.
\newblock Parameterized algorithms for modular-width.
\newblock In G.~Z. Gutin and S.~Szeider, editors, {\em Parameterized and Exact
  Computation - 8th International Symposium, {IPEC}}, volume 8246 of {\em
  LNCS}, pages 163--176. Springer, 2013.

\bibitem{GajJKL2022}
K.~Gajjar, A.~V. Jha, M.~Kumar, and A.~Lahiri.
\newblock Reconfiguring shortest paths in graphs.
\newblock In {\em Thirty-Sixth {AAAI} Conference on Artificial Intelligence,
  {AAAI} 2022, Thirty-Fourth Conference on Innovative Applications of
  Artificial Intelligence, {IAAI} 2022, The Twelveth Symposium on Educational
  Advances in Artificial Intelligence, {EAAI} 2022}, pages 9758--9766. {AAAI}
  Press, 2022.

\bibitem{GM60}
T.~Gallai and A.~N. Milgram.
\newblock Verallgemeinerung eines graphentheoretischen {S}atzes von {R}\'edei.
\newblock {\em Acta Scientiarum Mathematicarum}, 21(3-4):181--186, 1960.

\bibitem{GarJoh79}
M.~R. Garey and D.~S. Johnson.
\newblock {\em Computers and Intractability}.
\newblock Freeman, 1979.

\bibitem{DBLP:journals/siamcomp/GareyJT76}
M.~R. Garey, D.~S. Johnson, and R.~E. Tarjan.
\newblock The planar {H}amiltonian circuit problem is {NP}-complete.
\newblock {\em SIAM Journal on Computing}, 5(4):704--714, 1976.

\bibitem{PP-L21}
J.~P. Georges, D.~W. Mauro, and M.~A. Whittlesey.
\newblock Relating path coverings to vertex labellings with a condition at
  distance two.
\newblock {\em Discrete Mathematics}, 135(1-3):103--111, 1994.

\bibitem{GolPauLee2022}
P.~A. Golovach, D.~Paulusma, and E.~J. van Leeuwen.
\newblock Induced disjoint paths in {AT}-free graphs.
\newblock {\em Journal of Computer and System Sciences}, 124:170--191, 2022.

\bibitem{goodman1974hamiltonian}
S.~Goodman and S.~Hedetniemi.
\newblock On the {H}amiltonian completion problem.
\newblock In {\em Graphs and Combinatorics}, pages 262--272. Springer, 1974.

\bibitem{GuoHufNie2004}
J.~Guo, F.~H{\"u}ffner, and R.~Niedermeier.
\newblock A structural view on parameterizing problems: distance from
  triviality.
\newblock In R.~Downey, M.~Fellows, and F.~Dehne, editors, {\em International
  Workshop on Parameterized and Exact Computation IWPEC 2004}, volume 3162 of
  {\em LNCS}, pages 162--173. Springer, 2004.

\bibitem{GDLHSH18}
L.~Guo, Y.~Deng, K.~Liao, Q.~He, T.~Sellis, and Z.~Hu.
\newblock A fast algorithm for optimally finding partially disjoint shortest
  paths.
\newblock In J.~Lang, editor, {\em Proceedings of the Twenty-Seventh
  International Joint Conference on Artificial Intelligence, {IJCAI}}, pages
  1456--1462. ijcai.org, 2018.

\bibitem{DBLP:conf/stacs/HabibPV98}
M.~Habib, C.~Paul, and L.~Viennot.
\newblock A synthesis on partition refinement: {A} useful routine for strings,
  graphs, boolean matrices and automata.
\newblock In M.~Morvan, C.~Meinel, and D.~Krob, editors, {\em {STACS} 98, 15th
  Annual Symposium on Theoretical Aspects of Computer Science}, volume 1373 of
  {\em LNCS}, pages 25--38. Springer, 1998.

\bibitem{HarNar2024}
D.~G. Harris and N.~S. Narayanaswamy.
\newblock A faster algorithm for vertex cover parameterized by solution size.
\newblock In O.~Beyersdorff, M.~M. Kant{\'{e}}, O.~Kupferman, and
  D.~Lokshtanov, editors, {\em 41st International Symposium on Theoretical
  Aspects of Computer Science, {STACS}}, volume 289 of {\em LIPIcs}, pages
  40:1--40:18. Schloss Dagstuhl - Leibniz-Zentrum f{\"{u}}r Informatik, 2024.

\bibitem{hartman1988variations}
I.~B.-A. Hartman.
\newblock Variations on the {G}allai-{M}ilgram theorem.
\newblock {\em Discrete Mathematics}, 71(2):95--105, 1988.

\bibitem{DBLP:journals/jct/JohnsonRST01}
T.~Johnson, N.~Robertson, P.~D. Seymour, and R.~Thomas.
\newblock Directed tree-width.
\newblock {\em Journal of Combinatorial Theory, Series B}, 82(1):138--154,
  2001.

\bibitem{KK12}
K.-I. Kawarabayashi and Y.~Kobayashi.
\newblock A linear time algorithm for the induced disjoint paths problem in
  planar graphs.
\newblock {\em Journal of Computer and System Sciences}, 78(2):670--680, 2012.

\bibitem{kawarabayashi2012disjoint}
K.-I. Kawarabayashi, Y.~Kobayashi, and B.~Reed.
\newblock The disjoint paths problem in quadratic time.
\newblock {\em Journal of Combinatorial Theory, Series B}, 102(2):424--435,
  2012.

\bibitem{Kou2013}
M.~Kouteck{\'{y}}.
\newblock Solving hard problems on neighborhood diversity.
\newblock Master thesis, Charles University in Prague, Chech Republic, Facultas
  Mathematica Physicaque, Department of Applied Mathematics, April 2013.

\bibitem{krishnamoorthy1975np}
M.~S. Krishnamoorthy.
\newblock An {NP}-hard problem in bipartite graphs.
\newblock {\em ACM SIGACT News}, 7(1):26--26, 1975.

\bibitem{kundu1976linear}
S.~Kundu.
\newblock A linear algorithm for the {H}amiltonian completion number of a tree.
\newblock {\em Information Processing Letters}, 5(2):55--57, 1976.

\bibitem{LafMou2024}
M.~Lafond and V.~Moulton.
\newblock Path partitions of phylogenetic networks.
\newblock {\em Theoretical Computer Science}, 1024:114907, 2025.

\bibitem{lampis2012algorithmic}
M.~Lampis.
\newblock Algorithmic meta-theorems for restrictions of treewidth.
\newblock {\em Algorithmica}, 64(1):19--37, 2012.

\bibitem{le2003splitting}
H.-O. Le, V.~B. Le, and H.~M{\"u}ller.
\newblock Splitting a graph into disjoint induced paths or cycles.
\newblock {\em Discrete Applied Mathematics}, 131(1):199--212, 2003.

\bibitem{DBLP:journals/ol/LiYL24}
S.~Li, W.~Yu, and Z.~Liu.
\newblock A local search algorithm for the $k$-path partition problem.
\newblock {\em Optimization Letters}, 18(1):279--290, 2024.

\bibitem{lochet2021polynomial}
W.~Lochet.
\newblock A polynomial time algorithm for the $k$-disjoint shortest paths
  problem.
\newblock In {\em ACM-SIAM Symposium on Discrete Algorithms, {SODA}}, pages
  169--178. SIAM, 2021.

\bibitem{DBLP:journals/ajc/MagnantM09}
C.~Magnant and D.~M. Martin.
\newblock A note on the path cover number of regular graphs.
\newblock {\em Australasian Journal of Combinatorics}, 43:211--218, 2009.

\bibitem{manuel2018revisiting}
P.~D. Manuel.
\newblock Revisiting path-type covering and partitioning problems.
\newblock {\em CoRR ArXiv preprint}, arXiv:1807.10613, 2018.

\bibitem{pmanuelisometric}
P.~D. Manuel.
\newblock On the isometric path partition problem.
\newblock {\em Discussiones Mathematicae Graph Theory}, 41(4):1077--1089, 2021.

\bibitem{MarPSL2023}
B.~Martin, D.~Paulusma, S.~Smith, and E.~J. van Leeuwen.
\newblock Few induced disjoint paths for {$H$}-free graphs.
\newblock {\em Theoretical Computer Science}, 939:182--193, 2023.

\bibitem{marx20}
D.~Marx.
\newblock {Chordless Cycle Packing Is Fixed-Parameter Tractable}.
\newblock In F.~Grandoni, G.~Herman, and P.~Sanders, editors, {\em 28th Annual
  European Symposium on Algorithms, ESA}, volume 173 of {\em Leibniz
  International Proceedings in Informatics (LIPIcs)}, pages 71:1--71:19,
  Dagstuhl, Germany, 2020. Schloss Dagstuhl--Leibniz-Zentrum f{\"u}r
  Informatik.

\bibitem{M10}
M.~Mezzini.
\newblock On the complexity of finding chordless paths in bipartite graphs and
  some interval operators in graphs and hypergraphs.
\newblock {\em Theoretical Computer Science}, 411(7-9):1212--1220, 2010.

\bibitem{monnot2007path}
J.~Monnot and S.~Toulouse.
\newblock The path partition problem and related problems in bipartite graphs.
\newblock {\em Operations Research Letters}, 35(5):677--684, 2007.

\bibitem{DBLP:journals/dm/Muller96a}
H.~M{\"{u}}ller.
\newblock Hamiltonian circuits in chordal bipartite graphs.
\newblock {\em Discrete Mathematics}, 156(1-3):291--298, 1996.

\bibitem{ntafos1979path}
S.~C. Ntafos and S.~L. Hakimi.
\newblock On path cover problems in digraphs and applications to program
  testing.
\newblock {\em IEEE Transactions on Software Engineering}, SE-5(5):520--529,
  1979.

\bibitem{DBLP:journals/ipl/PanC05}
J.{-}J. Pan and G.~J. Chang.
\newblock Isometric-path numbers of block graphs.
\newblock {\em Information Processing Letters}, 93(2):99--102, 2005.

\bibitem{DBLP:journals/dam/PanC05}
J.{-}J. Pan and G.~J. Chang.
\newblock Path partition for graphs with special blocks.
\newblock {\em Discrete Applied Mathematics}, 145(3):429--436, 2005.

\bibitem{DBLP:journals/tcs/PanC07}
J.{-}J. Pan and G.~J. Chang.
\newblock Induced-path partition on graphs with special blocks.
\newblock {\em Theoretical Computer Science}, 370(1-3):121--130, 2007.

\bibitem{pinter1987mapping}
S.~S. Pinter and Y.~Wolfstahl.
\newblock On mapping processes to processors in distributed systems.
\newblock {\em International Journal of Parallel Programming}, 16(1):1--15,
  1987.

\bibitem{ROBERTSON199565}
N.~Robertson and P.~D. Seymour.
\newblock Graph minors {XIII}. {T}he disjoint paths problem.
\newblock {\em Journal of Combinatorial Theory, Series B}, 63(1):65--110, 1995.

\bibitem{PFVLSIbook}
A.~Schrijver, L.~Lovasz, B.~Korte, H.~J. Promel, and R.~L. Graham.
\newblock {\em Paths, Flows, and VLSI-Layout}.
\newblock Springer-Verlag, Berlin, Heidelberg, 1990.

\bibitem{DBLP:journals/tcs/ShihH99}
W.{-}K. Shih and W.{-}L. Hsu.
\newblock A new planarity test.
\newblock {\em Theoretical Computer Science}, 223(1-2):179--191, 1999.

\bibitem{skupien1974path}
Z.~Skupi{\'e}n.
\newblock Path partitions of vertices and hamiltonicity of graphs.
\newblock In {\em Proceedings of the Second Czechoslovakian Symposium on Graph
  Theory, Prague}, 1974.

\bibitem{Slivkins2010}
A.~Slivkins.
\newblock Parameterized tractability of edge-disjoint paths on directed acyclic
  graphs.
\newblock {\em SIAM Journal on Discrete Mathematics}, 24(1):146--157, 2010.

\bibitem{DBLP:journals/tcs/Steiner03}
G.~Steiner.
\newblock On the $k$-path partition of graphs.
\newblock {\em Theoretical Computer Science}, 290(3):2147--2155, 2003.

\bibitem{TG22}
M.~Thiessen and T.~G{\"{a}}rtner.
\newblock Online learning of convex sets on graphs.
\newblock In M.{-}R. Amini, S.~Canu, A.~Fischer, T.~Guns, P.~Kralj Novak, and
  G.~Tsoumakas, editors, {\em Machine Learning and Knowledge Discovery in
  Databases - European Conference, {ECML} {PKDD}, Proceedings, Part {IV}},
  volume 13716 of {\em LNCS}, pages 349--364. Springer, 2022.

\bibitem{TG21}
M.~Thiessen and T.~Gärtner.
\newblock Active learning of convex halfspaces on graphs.
\newblock In M.~Ranzato, A.~Beygelzimer, Y.~Dauphin, P.S. Liang, and J.~Wortman
  Vaughan, editors, {\em Proceedings of the 35th Conference on Neural
  Information Processing Systems, NeurIPS}, volume~34, pages 23413--23425.
  Curran Associates, Inc., 2021.

\bibitem{9568702}
C.~Wang, Y.~Song, G.~Fan, H.~Jin, L.~Su, F.~Zhang, and X.~Wang.
\newblock Optimizing cross-line dispatching for minimum electric bus fleet.
\newblock {\em IEEE Transactions on Mobile Computing}, 22(4):2307--2322, 2023.

\end{thebibliography}

\end{document}